\documentclass[
 reprint,
 amsmath,amssymb,
 pra,superscriptaddress, nofootinbib
]{revtex4-2}
\usepackage[utf8]{inputenc}
\usepackage{graphicx}
\usepackage{dcolumn}
\usepackage{bm}
\usepackage{stmaryrd}
\usepackage{dsfont}
\usepackage{tikz}
\usepackage{tkz-euclide}
\usepackage[colorlinks=true, pdfstartview=FitV, linkcolor=blue, citecolor=blue, urlcolor=blue]{hyperref}
\usepackage{braket}
\usepackage{pgfplots}
\usepackage{enumerate}
\usepackage[caption=false]{subfig}
\usepackage{booktabs}
\usepackage{xcolor}
\usepackage{stmaryrd}
\usepackage{microtype}
\usepackage{bm}
\usepackage{amsmath,amssymb,amsthm,mathtools}
\usepackage{array}
\usepackage{enumitem}
\usepackage{mdframed}

\newtheorem{definition}{Definition}
\newtheorem{lemma}{Lemma}
\newtheorem{proposition}{Proposition}
\newtheorem{theorem}{Theorem}

\theoremstyle{definition}
\newmdtheoremenv[
  linewidth=0.5pt,
  innertopmargin=8pt,
  innerbottommargin=8pt,
  innerleftmargin=8pt,
  innerrightmargin=8pt
]{algorithm}{Algorithm}
\theoremstyle{plain}

\newcommand{\cP}{\mathcal P}
\newcommand{\cC}{\mathcal C}
\newcommand{\CNOT}{\operatorname{CNOT}}
\begin{document}

\title{Exact efficient simulation of noisy logical magic states using Clifford stabilizers}

\author{Yugo Takada}
 \email{yugo.takada1@gmail.com}
\affiliation{%
Graduate School of Engineering Science, The University of Osaka, 1--3 Machikaneyama, Toyonaka, Osaka 560--8531, Japan
}%
\author{Stephen~D.~Bartlett}
\affiliation{%
School of Physics, The University of Sydney, Sydney, NSW 2006, Australia
}%
\author{Dominic~J.~Williamson}
\affiliation{%
School of Physics, The University of Sydney, Sydney, NSW 2006, Australia
}%

\date{\today}

\begin{abstract}
The preparation of high-fidelity logical magic states is a crucial subroutine for universal fault-tolerant quantum computation (FTQC). 
Predicting the performance of FTQC and developing improved protocols rely on numerical methods to classically simulate logical magic state preparation in the presence of noise.  
Clifford logic on Pauli-stabilizer codes with circuit-level Pauli errors can be efficiently simulated using Pauli-stabilizer formalism, but the non-Clifford operations required to prepare logical magic states render generic simulation inefficient. 
We introduce \emph{Clifford-stabilizer simulation}, an exact and efficient algorithm based on updating a Clifford-stabilizer group to simulate noisy preparation protocols for a broad class of logical magic states used to implement non-Clifford gates in the third level of the Clifford hierarchy under circuit-level Pauli errors. 
Clifford-stabilizer simulation applies to a range of operations that commonly appear in preparation protocols for such logical magic states, including Pauli-stabilizer measurements, logical Clifford measurements, and transversal non-Clifford gates.
Our algorithm for Clifford-stabilizer simulation maps a non-Clifford circuit with sampled circuit-level Pauli errors to a Clifford circuit that exactly reproduces its measurement outcome distribution, achieving time and space complexities polynomial in relevant protocol parameters.
We perform exact simulation of magic state cultivation up to fault distance 7 by Clifford-stabilizer simulation. 
Our method provides a route to perform exact benchmarking of large-scale logical magic state preparation protocols required for useful FTQC.
\end{abstract}

\maketitle
\makeatletter
\renewcommand{\l@subsection}[2]{}
\renewcommand{\l@subsubsection}[2]{}
\makeatother
\tableofcontents

\newpage
\section{Introduction}\label{sec:intro}

Quantum computers are capable of performing important calculations for quantum chemistry~\cite{cao2019quantum} and cryptography~\cite{365700,shor1999polynomial} exponentially faster than any known classical method. 
However, quantum devices are inherently noisy, and in order to run large-scale quantum algorithms reliably on faulty components,  fault-tolerant quantum computation (FTQC) makes use of quantum error correction (QEC) to suppress errors~\cite{548464,10.1145/258533.258579,kitaev1997quantum} together with the careful design of fault-tolerant quantum gates to execute logic. 

Among quantum gates, Clifford gates are relatively straightforward to implement fault tolerantly. However, Clifford gates alone are not computationally universal, and universal quantum computation therefore requires non-Clifford gates. A standard approach to implementing such gates is to prepare magic states and consume them through gate teleportation~\cite{PhysRevA.62.052316}. 
Thus, the preparation of high-fidelity logical magic states is crucial in FTQC, but it often constitutes a major part of its space-time overhead, particularly in conventional approaches based on magic state distillation~\cite{PhysRevA.71.022316,PhysRevA.86.032324,PhysRevA.86.052329,Litinski2019magicstate}. Reducing this overhead has motivated the development of alternative protocols, including recent low-overhead schemes such as magic state cultivation~\cite{gidney2024magicstatecultivationgrowing,thxx-njr6,9kys-3whh,gpvl-lg4c,p8tw-6kq9,claes2025cultivatingtstatessurface,ch5r-cnfq,hirano2026efficientmagicstatecultivation,hetenyi2026constantdepthmagicstate}.

Constructing low-overhead logical magic state preparation protocols and evaluating their performance rely heavily on the classical simulation of noisy non-Clifford circuits.
Clifford circuits, which consist of Pauli-stabilizer state preparations, Clifford gates, and Pauli measurements, can be efficiently simulated via the stabilizer formalism~\cite{gottesman1998heisenbergrepresentationquantumcomputers,PhysRevA.70.052328} since they only produce Pauli-stabilizer states whose stabilizers are Pauli operators. This efficient simulation extends to noisy circuits with Pauli errors.
However, exact classical simulation of non-Clifford circuits generally incurs an exponential computational cost, even with a Pauli error model.
Specifically, state-vector simulation can exactly simulate generic non-Clifford circuits, but its time and memory costs grow exponentially with the number of physical qubits.
To enable simulation of larger-scale protocols, prior work has often replaced the target logical magic state preparation protocol with a computationally tractable Clifford analog and used the performance of the latter to estimate that of the target protocol.
In particular, several studies of magic state cultivation~\cite{gidney2024magicstatecultivationgrowing,thxx-njr6,9kys-3whh,gpvl-lg4c,p8tw-6kq9,claes2025cultivatingtstatessurface,ch5r-cnfq,hirano2026efficientmagicstatecultivation,hetenyi2026constantdepthmagicstate} simulate the preparation of the $S|+\rangle$ state as a proxy for the target magic state $T|+\rangle$. However, such an approach provides no general guarantee that the estimated performance accurately reflects that of the target non-Clifford protocol.

Stabilizer-decomposition
methods~\cite{PhysRevLett.116.250501,PhysRevX.6.021043,Bravyi2019simulationofquantum,kissinger2022simulating,aziz2026classicalsimulationslowmagic,kuyanov2026efficientclassicalsimulationlowrankwidth} can avoid an exponential dependence on the physical qubit count, but their costs generally increase exponentially with the number of non-Clifford gates.
Thus, exact simulation of noisy logical magic state preparation protocols is not scalable with these existing methods.
Although recently proposed simulation methods for universal fault-tolerant circuits~\cite{li2025softhighperformancesimulatoruniversal,haenel2026tsimfastuniversalsimulator,chase2026clifftfastexactsimulation,fang2026symftuniversalfaulttolerantquantum,wan2025cuttingstabiliserdecompositionsmagic} have improved the practical reach of near-Clifford simulation, their worst-case complexities are still exponential in parameters characterizing the non-Clifford complexity. 
A notable exception to the generic exponential scaling for logical magic state preparation protocols is the method recently proposed in Ref.~\cite{fby6-xjbm}. For a broad family of noisy protocols that prepare logical magic states for implementing non-Clifford gates in the third level of the Clifford hierarchy, its simulation complexity is polynomial in the number of physical qubits, the number of physical non-Clifford gates, and the stabilizer or Pauli rank of the target logical magic state. However, it still represents the target logical magic state through stabilizer- or Pauli-rank decompositions, so its cost can grow exponentially when the corresponding rank grows exponentially. 
Another efficient method for sampling syndrome statistics under circuit-level Pauli errors has also been developed for a certain class of error-correcting circuits with diagonal gates from the third level of the Clifford hierarchy~\cite{delafuente2026highthresholddecodingnonpaulicodes}. However, the proposed benchmarking procedure yields only upper bounds on logical error rates.

In this work, we introduce \textit{Clifford-stabilizer simulation}, an efficient method for the exact simulation of noisy preparation protocols for a broad class of logical magic states used to implement non-Clifford gates in the third level of the Clifford hierarchy, and with circuit-level Pauli noise.
The starting observation is that many commonly used logical magic states, including $|\bar{T}\rangle$, $|\overline{CCZ}\rangle$, and $|\bar{H}\rangle$ states, are unique common $+1$ eigenstates of Pauli stabilizers together with non-Pauli Clifford operators, which we refer to as \emph{Clifford stabilizers}. We formalize this description using a \textit{Clifford-stabilizer group} and identify sufficient conditions under which its evolution can be systematically tracked under circuit-level Pauli errors and the operations typically appearing in the non-Clifford circuits for preparing logical magic states.
Exploiting this structure, we develop an efficient algorithm to exactly simulate such non-Clifford circuits under circuit-level Pauli errors.
A summary of our main results is provided below.

\subsection{Summary of Main Results}
For a logical magic state described by a Clifford-stabilizer group satisfying suitable algebraic conditions (Definition~\ref{def:admissible}), our Clifford-stabilizer simulation can exactly and efficiently simulate non-Clifford circuits implementing Pauli- and Clifford-stabilizer measurements under circuit-level Pauli noise.
This class includes logical magic states used to implement practically important diagonal, and more general, non-Clifford logical gates in the third level of the Clifford hierarchy. 
We also show that the applicability of the Clifford-stabilizer simulation extends beyond Pauli- and Clifford-stabilizer measurements. In particular, it can also be used to simulate several other operations that commonly appear in logical magic state preparation protocols, including certain injection circuits, logical non-Clifford gates such as transversal non-Clifford gates, code block growth, more general measurements, midcircuit measurements with correction, and ancilla resets.

Our algorithm implementing the Clifford-stabilizer simulation is summarized in the following. For a non-Clifford circuit in the class considered here with sampled circuit-level Pauli errors, we show that there exists a Clifford circuit that exactly reproduces its measurement probabilities.
This equivalence exploits the structure of the evolution of the Clifford-stabilizer groups.
We provide an efficient algorithm for constructing the Clifford circuit, so the simulation of the noisy non-Clifford circuit is reduced to a Clifford simulation preceded by an efficient preprocessing step.
This enables exact and efficient simulation in which the measurement outcomes and the occurrence of a logical error for a sampled noisy non-Clifford circuit are obtained through the simulation of \emph{a single} Clifford circuit, without relying on stabilizer- or Pauli-rank decompositions.
The time and space complexities of our Clifford-stabilizer simulation are thus polynomial in relevant protocol parameters, such as the number of physical qubits, the number of logical qubits, and the number of physical non-Clifford gates, with no direct dependence on the stabilizer or Pauli rank of the target logical magic state.
    
We apply the proposed Clifford-stabilizer simulation method to perform exact simulation of magic state cultivation with fault distances up to $f=7$, reaching regimes that were previously inaccessible in the literature. Note that, although the simulation time per circuit run remains tractable as the protocol size increases, the logical error rate of the $f=7$ protocol is too low to obtain sufficiently many logical-error events by Monte Carlo sampling to make a statistically precise determination of the logical error rate. This limitation is separate from the cost of our simulation method and instead reflects the intrinsic difficulty of estimating extremely low logical error rates in large-scale protocols even if each run is efficiently simulable.

\subsection{Section Outline}
The remainder of this paper is organized as follows.
In Sec.~\ref{sec:prelims}, we review background material relevant to this work, including the Pauli-stabilizer group, its tableau representation, and several standard magic states. 
In Sec.~\ref{sec:clifford_stabilizer}, we introduce the notion of a Clifford-stabilizer group to represent logical magic states.
In Sec.~\ref{sec:clifford_stb_update_diagonal}, we establish systematic rules for updating Clifford-stabilizer groups under circuit-level Pauli errors and Pauli- or Clifford-stabilizer measurements for logical magic states associated with diagonal logical gates.
In Sec.~\ref{sec:simulation}, we apply this framework to exactly and efficiently simulate magic state cultivation, while deferring the description of an efficient algorithm that implements the framework to Sec.~\ref{sec:algorithm}.
In Sec.~\ref{sec:clifford_stb_update}, we generalize the Clifford-stabilizer simulation framework beyond the diagonal case and provide rigorous proofs of the update rules.
In Sec.~\ref{sec:algorithm}, we construct an efficient algorithm for updating the Clifford-stabilizer group.
In Sec.~\ref{sec:scope}, we show that the algorithm can simulate a broader class of operations appearing in logical magic state protocols beyond Pauli- or Clifford-stabilizer measurements.
Finally, in Sec.~\ref{sec:discuss}, we discuss the implications and possible extensions of our results.
Appendix~\ref{app:cultivation} provides further details on the magic state cultivation protocols simulated in Sec.~\ref{sec:simulation}.
Appendix~\ref{app:proofs} provides supplementary proofs of the claims presented in Sec.~\ref{sec:clifford_stb_update}.

\section{Preliminaries}\label{sec:prelims}

In this section, we review relevant background material on the Pauli-stabilizer formalism and magic states. 

\subsection{Pauli-stabilizer group}
\label{subsec:prelims_pauli}
We first define the $n$-qubit Pauli group as
\begin{equation}
\cP_n=\{\pm1, \pm i\} \times \{I,X,Y,Z\}^{\otimes n}.
\end{equation}
Then the Clifford hierarchy is defined recursively by
\begin{equation}
\cC^{(1)}=\cP_n,\qquad
\cC^{(k)}=\{U:UPU^\dagger\in \cC^{(k-1)} \quad \forall P\in\cP_n\}.
\end{equation}
The second level $\cC^{(2)}$ is called the Clifford group, and the gates in the Clifford group are Clifford gates. Common Clifford gates are $H$, $S$, $CZ$, and $\CNOT$ gates.
Gates outside $\cC^{(2)}$ are called non-Clifford gates. Typical non-Clifford gates in the third level $\cC^{(3)}$ are $T$ and $CCZ$ gates.

A Pauli-stabilizer group on $n$ qubits is an abelian subgroup $\mathcal{G}_\mathrm{P}\subset\cP_n$ that does not contain $-I$. A pure $n$-qubit Pauli-stabilizer state is the unique simultaneous $+1$ eigenstate of $n$ independent commuting Pauli generators
\begin{equation}
\mathcal{G}_\mathrm{P}=\langle G_1,\ldots,G_n\rangle,
\qquad G_j|\psi\rangle=|\psi\rangle.
\end{equation}

Now we describe how the Pauli-stabilizer group of a Pauli-stabilizer state is updated by Clifford gates and Pauli measurements.
A Clifford gate $C$ updates a Pauli-stabilizer group by conjugating every generator $G_j$:
\begin{equation}
G_j\leftarrow CG_jC^{\dagger}.
\end{equation}
When a Pauli operator $M$ is measured, there are two cases. If $M$ commutes with all current generators, then the measurement outcome is deterministic and the group remains unchanged. The outcome is found by expressing $M$ as
\begin{equation}
M=\pm\prod_{j}G_j.
\label{eq:pauli_group_update_0}
\end{equation}
If $M$ anticommutes with at least one generator, the outcome is uniformly random. The group is updated by choosing a pivot generator $G_p$ that anticommutes with $M$, and for every other generator $G_j$ that anticommutes with $M$, replace it by
\begin{equation}
G_j\leftarrow G_jG_p.
\label{eq:pauli_group_update_1}
\end{equation}
The product $G_jG_p$ commutes with $M$ because both $G_j$ and $G_p$ anticommute with $M$. Finally, replace the pivot $G_p$ by $M$ with the obtained random sign:
\begin{equation}
G_p\leftarrow (-1)^b M,
\qquad b\in\{0,1\}.
\label{eq:pauli_group_update_2}
\end{equation}

\subsection{Tableau representation}
The tableau representation~\cite{PhysRevA.70.052328} represents a Pauli-stabilizer state using binary representations of its stabilizer and destabilizer generators. This allows the group to be updated efficiently. For $n$ qubits, it contains binary variables $x_{ij},z_{ij}$ for $i\in\{1,\ldots,2n\}$ and $j\in\{1,\ldots,n\}$, together with phase bits $r_i$:
\begin{equation}
\resizebox{\columnwidth}{!}{$
\left(\begin{array}{ccc|ccc|c}
 x_{11}&\cdots&x_{1n}&z_{11}&\cdots&z_{1n}&r_1\\
 \vdots&\ddots&\vdots&\vdots&\ddots&\vdots&\vdots\\
 x_{n1}&\cdots&x_{nn}&z_{n1}&\cdots&z_{nn}&r_n\\\hline
 x_{(n+1)1}&\cdots&x_{(n+1)n}&z_{(n+1)1}&\cdots&z_{(n+1)n}&r_{n+1}\\
 \vdots&\ddots&\vdots&\vdots&\ddots&\vdots&\vdots\\
 x_{(2n)1}&\cdots&x_{(2n)n}&z_{(2n)1}&\cdots&z_{(2n)n}&r_{2n}
\end{array}\right)$}.
\end{equation}
Rows $1$ to $n$ represent destabilizer generators, and rows $n+1$ to $2n$ represent stabilizer generators.  We denote the Pauli operator represented by row $i$ by $R_i$. The destabilizer generators are mutually commuting, and each $R_i$, $i\in\{1,\ldots,n\}$, anticommutes only with the corresponding stabilizer generator $R_{n+i}$. If
\begin{equation}
R_i=\pm P_1P_2\cdots P_n,
\end{equation}
then the local bits determine $P_j$ by
\begin{equation}
(x_{ij},z_{ij})=(0,0),(1,0),(1,1),(0,1)
\mapsto I,X,Y,Z.
\end{equation}
The sign bit $r_i$ is $0$ for a positive sign and $1$ for a negative sign.

In this representation, updating the tableau under a Clifford gate, such as $H$, $S$, or
$\CNOT$, requires $O(n)$ time. 
For a single-qubit $Z$-basis measurement, if the outcome is random, the update costs $O(n^2)$.
If the outcome is deterministic, we can determine the outcome in $O(n^2)$ time using the destabilizers.
Note that, by maintaining an inverse stabilizer tableau~\cite{Gidney2021stimfaststabilizer}, the deterministic outcome can be determined in $O(n)$. 

\subsection{Magic states}
Magic states are resource states used to implement non-Clifford gates through gate teleportation.
More formally, an $n$-qubit magic state is a pure state that is not the unique simultaneous eigenstate of any $n$ independent commuting Pauli operators.
In this work, we focus on magic states that are the $+1$ eigenstates of non-Pauli Clifford operators, a class of resource states used to implement non-Clifford gates in the third level of the Clifford hierarchy.
A canonical single-qubit example is the $|T\rangle$ state:
\begin{equation}
|T\rangle=T|+\rangle=\frac{|0\rangle+e^{i\pi/4}|1\rangle}{\sqrt2},
\end{equation}
which is the unique $+1$ eigenstate of the non-Pauli Clifford operator
\begin{equation}
H_{XY}=\frac{X+Y}{\sqrt2}=TXT^\dagger=e^{-i\pi/4}SX.
\end{equation}
The $|T\rangle$ state is used to implement the $T$ gate through gate teleportation as depicted in Fig.~\ref{fig:teleportation}.
\begin{figure}[b]
    \centering
    \includegraphics[width=0.8\linewidth]{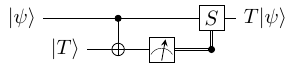}
    \caption{Gate-teleportation circuit for implementing the $T$ gate using the $|T\rangle$ state as a resource state.}
    \label{fig:teleportation}
\end{figure}
A standard three-qubit example is the $|CCZ\rangle$ state:
\begin{equation}
|CCZ\rangle=CCZ|+++\rangle,
\end{equation}
which is the unique simultaneous $+1$ eigenstate of the three non-Pauli Clifford operators
\begin{align}
    &CCZ \, X_1\, CCZ=X_1\, CZ_{2,3}, \\
    &CCZ \, X_2\, CCZ=X_2\, CZ_{1,3}, \\
    &CCZ \, X_3\, CCZ=X_3\, CZ_{1,2}.
\end{align}
Another example is the $|H\rangle$ state:
\begin{equation}
|H\rangle=e^{-i\pi Y/8}|0\rangle=\cos\left(\frac{\pi}{8}\right)|0\rangle+\sin\left(\frac{\pi}{8}\right)|1\rangle.
\end{equation}
The state is the unique $+1$ eigenstate of the non-Pauli Clifford operator
\begin{equation}
H=\frac{X+Z}{\sqrt2}=e^{-i\pi Y/8}Z e^{i\pi Y/8}.
\end{equation}

\section{Clifford-stabilizer group}
\label{sec:clifford_stabilizer}

In this section, we define the notion of a \emph{Clifford-stabilizer group}, which provides a useful representation for logical magic states used to implement non-Clifford gates in the third level of the Clifford hierarchy. This representation plays a crucial role in our simulation method. 
The idea behind this representation is that, for magic states that are $+1$ eigenstates of non-Pauli Clifford operators, these operators can be regarded as stabilizers even though they are not Pauli operators.
We refer to these generalized stabilizer operators as \emph{Clifford stabilizers}. In analogy with a Pauli-stabilizer state, which is characterized as a unique $+1$ eigenstate of its Pauli-stabilizer group, we describe a \emph{Clifford-stabilizer state} as a unique $+1$ eigenstate of a Clifford-stabilizer group, with generators given by a combination of Pauli stabilizers together with Clifford stabilizers. 
Note that the formulations of beyond-Pauli stabilizer group are also introduced in Refs.~\cite{ni2015non,Webster2022xpstabiliser} for diagonal non-Pauli cases and in Refs.~\cite{nest2011monomial,erew2026extremizingmeasuresmagicpure} for more general cases.

More precisely, a Clifford-stabilizer group $\mathcal{G}_\mathrm{C}$ is defined to be a finite subgroup of the Clifford group $\cC^{(2)}$ such that the common $+1$ eigenspace of all elements of $\mathcal{G}_\mathrm{C}$, which we call the stabilizer subspace, is nontrivial. This definition automatically excludes $-I$ from $\mathcal{G}_\mathrm{C}$. 
In general, the dimension of the stabilizer subspace can be more than one. However, to describe a Clifford-stabilizer state, we seek a 1-dimensional trivial representation of the Clifford-stabilizer group.  
Equivalently, we seek a Clifford-stabilizer group with a 1-dimensional stabilizer subspace.
Unlike a Pauli-stabilizer group, a Clifford-stabilizer group can be non-Abelian.  Nevertheless, the elements of the group commute on the stabilizer subspace.

We can see how this is possible even for non-commuting elements.  Consider the multiplicative group commutator for unitaries $A$ and $B$ defined by
\begin{equation}
[A,B]:=ABA^\dagger B^\dagger.
\label{eq:group_commutator_def}
\end{equation}
In the situations considered in this work, when evaluating the commutation of unitaries $A$ and $B$, the current state is either a common $+1$ eigenstate of both $A$ and $B$, or an eigenstate of one of the two operators while the other is a Hermitian unitary. Under either condition, if $[A,B]$ acts as $I$ on a state $|\psi\rangle$, i.e.,
\begin{equation}
[A,B]|\psi\rangle=|\psi\rangle,
\label{eq:commutator_i}
\end{equation}
then 
\begin{equation}
AB|\psi\rangle=BA|\psi\rangle,
\end{equation}
and we say that $A$ and $B$ \emph{effectively commute} on $|\psi\rangle$ even if they do not commute as operators. Similarly, if $[A,B]$ acts as $-I$ on $|\psi\rangle$, namely,
\begin{equation}
[A,B]|\psi\rangle=-|\psi\rangle,
\label{eq:commutator_minus_i}
\end{equation}
then
\begin{equation}
AB|\psi\rangle=-BA|\psi\rangle,
\end{equation}
and we say that $A$ and $B$ \textit{effectively anticommute} on $|\psi\rangle$. 
For the Clifford-stabilizer group $\mathcal{G}_\mathrm{C}$, two elements $G_i, G_j \in \mathcal{G}_\mathrm{C}$ may have the group commutator $[G_i,G_j] \neq I$ as operators, but they satisfy 
\begin{equation}
[G_i,G_j]|\phi\rangle=|\phi\rangle
\end{equation}
for a state $|\phi\rangle$ in the stabilizer subspace.

While some aspects of Clifford-stabilizer states and groups appear very similar to their Pauli-stabilizer counterparts, we remark that Clifford-stabilizer groups generally have other notable properties that are distinct from those of Pauli-stabilizer groups. In particular, in contrast to a Pauli-stabilizer group, for which the number of independent generators determines the dimension of the stabilizer subspace, there is no analogous counting rule based solely on the number of generators.
Moreover, a Clifford-stabilizer group that uniquely stabilizes a state need not contain every Clifford operator that stabilizes that state, as later illustrated by an example in Sec.~\ref{subsec:deterministic}.

An example of a Clifford-stabilizer state is a logical magic state of a Pauli-stabilizer code, which is the unique $+1$ eigenstate of a Clifford-stabilizer group generated by the Pauli stabilizers of the code together with global Clifford stabilizers that fix the logical degrees of freedom.

In what follows, the term \emph{stabilizer} denotes either a Pauli stabilizer or a Clifford stabilizer, and Clifford refers to non-Pauli Clifford, unless explicitly specified otherwise. 

\subsection{Example:  Surface code \texorpdfstring{$|\bar{T}\rangle$ state}{|Tbar> state}}
\label{subsec:clifford_stabilizer_ex_surface}
We consider the $|\bar{T}\rangle$ state of the distance-3 regular (unrotated) surface code shown in Fig.~\ref{fig:surface_code}(a) as an example of a logical magic state uniquely specified by a Clifford-stabilizer group. Thus, this state is a Clifford-stabilizer state.
\begin{figure}[t]
    \centering
    \includegraphics[width=0.95\linewidth]{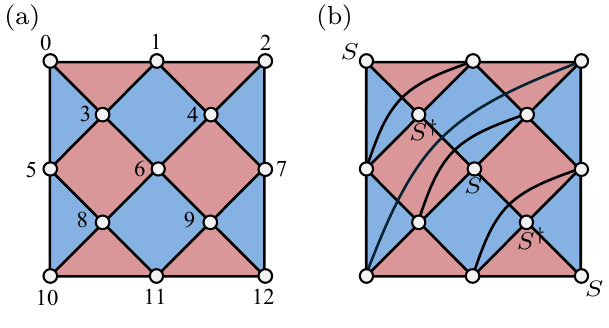}
    \caption{The distance-3 regular (unrotated) surface code. (a) The definition of the code. The circles represent qubits. The red and blue faces correspond to $X$ and $Z$ stabilizers, respectively. (b) The fold-transversal $\bar{S}$ gate. The $S$ and $S^\dagger$ gates act on the diagonal qubits alternately, and $CZ$ gates act on the rest of the qubits. Each $CZ$ gate, depicted as a black curved line, acts on a pair of qubits at mirrored positions.}
    \label{fig:surface_code}
\end{figure}
The surface code has 6 $X$-type and 6 $Z$-type stabilizers:
\begin{align}
G^X_1&=X_0X_1X_3,&
G^X_2&=X_1X_2X_4,\nonumber\\
G^X_3&=X_3X_5X_6X_8,&
G^X_4&=X_4X_6X_7X_9,\nonumber\\
G^X_5&=X_8X_{10}X_{11},&
G^X_6&=X_9X_{11}X_{12},
\end{align}
\begin{align}
G^Z_1&=Z_0Z_3Z_5,&
G^Z_2&=Z_5Z_8Z_{10},\nonumber\\
G^Z_3&=Z_1Z_3Z_4Z_6,&
G^Z_4&=Z_6Z_8Z_9Z_{11},\nonumber\\
G^Z_5&=Z_2Z_4Z_7,&
G^Z_6&=Z_7Z_9Z_{12}.
\end{align}
The $|\bar{T}\rangle$ state of the code is also stabilized by the Clifford operator
\begin{align}
\bar{H}_{XY}=
{}&e^{-i\pi/4}\bar{S}\bar{X}\nonumber\\
={}&e^{-i\pi/4}CZ_{1,5}CZ_{4,8}CZ_{7,11}CZ_{2,10}\nonumber\\
&\times S_0X_0S_3^\dagger X_3S_6X_6S_9^\dagger X_9S_{12}X_{12},
\label{eq:fold_h_operator}
\end{align}
where $\bar{X}$ is the logical $X$ operator and $\bar{S}$ is the fold-transversal logical $S$ operator shown in Fig.~\ref{fig:surface_code}(b). 

The Clifford-stabilizer group $\mathcal{G}_\mathrm{C}$ that uniquely defines the $|\bar{T}\rangle$ state is thus generated by all Pauli stabilizers together with the Clifford stabilizer $\bar{H}_{XY}$:
\begin{equation}
\mathcal{G}_\mathrm{C}=\langle G^X_1,\, G^X_2,\,\ldots, G^X_6,G^Z_1,\, G^Z_2,\,\ldots, G^Z_6, \,\bar{H}_{XY}\rangle.
\end{equation}
The Clifford stabilizer $\bar{H}_{XY}$ effectively commutes with every other generator of the group as follows. For $Z$ stabilizers, \begin{equation}
[G^Z_i,\bar{H}_{XY}]=I,
\qquad i=1,\ldots,6,
\end{equation}
because Pauli $Z$ operators commute with the diagonal Clifford operators. For $X$ stabilizers,
\begin{equation}
[G^X_i,\bar{H}_{XY}]=G^Z_i,
\qquad i=1,\ldots,6.
\label{eq:surface_commutation_x}
\end{equation}
Since $G^Z_i$ is a stabilizer of the $|\bar{T}\rangle$ state for every $i$, each $[G^X_i,\bar{H}_{XY}]$ acts as $I$ on the state, showing that every $X$ stabilizer $G^X_i$ effectively commutes with $\bar{H}_{XY}$ on the state.

\subsection{Example:  The \texorpdfstring{$\llbracket 8,3,2 \rrbracket$ code $|\overline{CCZ}\rangle$ state}{[[8,3,2]] code logical CCZ state}}

\begin{figure}[t]
    \centering
    \includegraphics[width=0.5\linewidth]{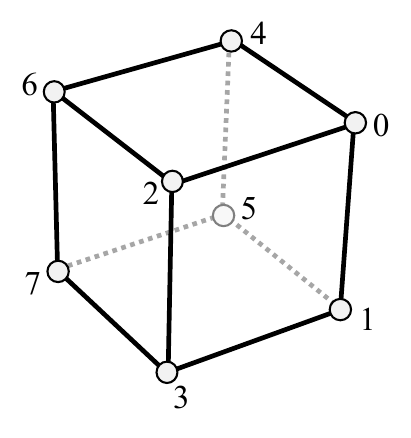}
    \caption{Cubic lattice on which the $\llbracket 8,3,2 \rrbracket$ code is defined. The circles represent qubits.}
    \label{fig:ccz_code}
\end{figure}

The $\llbracket 8,3,2 \rrbracket$ code $|\overline{CCZ}\rangle$ state provides a simple example with multiple Clifford stabilizers. 
The $\llbracket 8,3,2 \rrbracket$ code can be defined on a cubic lattice shown in Fig.~\ref{fig:ccz_code}. It has one $X$ stabilizer
\begin{equation}
G^X_1=X_0X_1X_2X_3X_4X_5X_6X_7
\end{equation}
and four $Z$ stabilizers
\begin{align}
G^Z_1&=Z_0Z_1Z_2Z_3,&
G^Z_2&=Z_0Z_1Z_4Z_5,\nonumber\\
G^Z_3&=Z_0Z_2Z_4Z_6,&
G^Z_4&=Z_0Z_1Z_2Z_3Z_4Z_5Z_6Z_7.
\end{align}
Three pairs of logical Pauli operators are
\begin{align}
\bar{X}_1&=X_0X_1X_2X_3,&
\bar{Z}_1&=Z_0Z_4,\nonumber\\
\bar{X}_2&=X_0X_1X_4X_5,&
\bar{Z}_2&=Z_0Z_2,\nonumber\\
\bar{X}_3&=X_0X_2X_4X_6,&
\bar{Z}_3&=Z_0Z_1.
\end{align}
The code has a transversal realization of the $\overline{CCZ}$ gate
\begin{equation}
\overline{CCZ}=T_0T^\dagger_1T^\dagger_2T_3T^\dagger_4T_5T_6T^\dagger_7.
\end{equation}
The $|\overline{CCZ}\rangle$ state, defined by  $|\overline{CCZ}\rangle=\overline{CCZ}|\overline{+++} \rangle$ of the $\llbracket 8,3,2 \rrbracket$ code, is thus a simultaneous $+1$ eigenstate of all Pauli stabilizers of the code and of the following three Clifford stabilizers:
\begin{align}
G^\mathrm{C}_1&=\overline{CCZ}\bar{X}_1\overline{CCZ}^\dagger=S_0X_0S^\dagger_1X_1S^\dagger_2X_2S_3X_3,\nonumber\\
G^\mathrm{C}_2&=\overline{CCZ}\bar{X}_2\overline{CCZ}^\dagger=S_0X_0S^\dagger_1X_1S^\dagger_4X_4S_5X_5,\nonumber\\
G^\mathrm{C}_3&=\overline{CCZ}\bar{X}_3\overline{CCZ}^\dagger=S_0X_0S^\dagger_2X_2S^\dagger_4X_4S_6X_6.
\end{align}
Hence, the Clifford-stabilizer group $\mathcal{G}_\mathrm{C}$ that uniquely defines the $|\overline{CCZ}\rangle$ state is
\begin{equation}
\mathcal{G}_\mathrm{C}=\langle G^X_1,\, G^Z_1,\, G^Z_2,\,G^Z_3,\, G^Z_4, \, G^C_1, \, G^C_2, \, G^C_3\rangle.
\end{equation}
The group commutators between the $Z$ stabilizers and the Clifford stabilizers are 
\begin{equation}
[G^Z_i,G^\mathrm{C}_j]=I,
\qquad i=1,2,3,4, \quad j=1,2,3,
\end{equation}
showing they commute with each other.
For $X$ stabilizers, the group commutators between them and the Clifford stabilizers are given by
\begin{equation}
[G^X_1,G^\mathrm{C}_i]=G^Z_i,
\qquad i=1,2,3.
\end{equation}
Because the $Z$ stabilizers act as the identity on the $|\overline{CCZ}\rangle$ state, this shows that the $X$ stabilizers and the Clifford stabilizers effectively commute with each other on the state.

\subsection{Example:  Steane code \texorpdfstring{$|\bar{H}\rangle$ state}{|Hbar> state}}
\begin{figure}[t]
    \centering
    \includegraphics[width=0.5\linewidth]{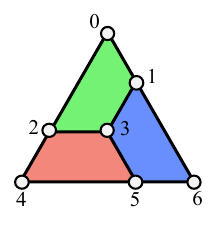}
    \caption{Steane code. The circles denote qubits. The $X$ and $Z$ stabilizer generators are defined to act on the qubits around each face.}
    \label{fig:steane_code}
\end{figure}

The Steane code $|\bar{H}\rangle$ state is an example whose Clifford stabilizers are not products of Pauli and diagonal Clifford operators. The Pauli stabilizers of the Steane code, shown in Fig.~\ref{fig:steane_code}, are
\begin{align}
G^X_1&=X_0X_1X_2X_3,&\nonumber
G^X_2&=X_2X_3X_4X_5,&\\
G^X_3&=X_1X_3X_5X_6,
\end{align}
\begin{align}
G^Z_1&=Z_0Z_1Z_2Z_3,&\nonumber
G^Z_2&=Z_2Z_3Z_4Z_5,&\\
G^Z_3&=Z_1Z_3Z_5Z_6.
\end{align}
The $|\bar{H}\rangle$ state of the code is defined as the simultaneous $+1$ eigenstate of all stabilizers and a logical $\bar{H}$ operator of the code.
The logical $\bar{H}$ operator can be implemented as the following transversal physical Hadamard operators:
\begin{equation}
\bar{H}=H_0H_1\cdots H_6.
\end{equation}
The Clifford-stabilizer group $\mathcal{G}_\mathrm{C}$ that uniquely specifies the $|\bar{H}\rangle$ state is given by
\begin{equation}
\mathcal{G}_\mathrm{C}=\langle G^X_1,\, G^X_2,\, G^X_3,\,G^Z_1,\, G^Z_2, \, G^Z_3, \, \bar{H}\rangle.
\end{equation}
The group commutators between the Pauli stabilizers and the Clifford stabilizer $\bar{H}$ are 
\begin{equation}
[G^X_i,\bar{H}]=G^X_iG^Z_i,
\qquad
[G^Z_i,\bar{H}]=G^X_iG^Z_i,
\end{equation}
where $G^X_iG^Z_i$ is a $Y$-type Pauli stabilizer. 
We can see that the generators in the group effectively commute with each other on the $|\bar{H}\rangle$ state because both $G^X_i$ and $G^Z_i$ are stabilizers of the state for every $i$.

\section{Clifford-stabilizer simulation: diagonal case}\label{sec:clifford_stb_update_diagonal}

We now introduce our Clifford-stabilizer simulation method for exactly simulating noisy logical magic state protocols under a standard circuit-level Pauli noise model. We first recall how ordinary Clifford simulation works. As described in Sec.~\ref{subsec:prelims_pauli}, efficient Clifford simulation is based on tracking the Pauli-stabilizer group that uniquely specifies the current state by updating this group under Clifford gates, state initialization, and measurements.

The core idea of our simulation method for noisy logical magic states is similar in spirit to the Pauli-stabilizer group tracking.
In Sec.~\ref{sec:clifford_stabilizer}, we showed that certain logical magic states can be represented as Clifford-stabilizer states, uniquely specified as the common $+1$ eigenstates of Clifford-stabilizer groups. Building on this observation, we present a systematic procedure for updating the Clifford-stabilizer group of a broad class of logical magic states under circuit-level Pauli errors and operations that appear in typical preparation protocols.
This procedure enables exact and efficient simulation of such protocols in the presence of circuit-level Pauli errors.

Typically, logical magic state protocols such as magic state cultivation~\cite{gidney2024magicstatecultivationgrowing,thxx-njr6,9kys-3whh,gpvl-lg4c,p8tw-6kq9,claes2025cultivatingtstatessurface,ch5r-cnfq,hirano2026efficientmagicstatecultivation,hetenyi2026constantdepthmagicstate} consist of the measurements of Hermitian Pauli or Clifford stabilizers in the Clifford-stabilizer group that defines the target logical magic state. 
These measurements prepare and verify the logical magic states. Accordingly, in this section we focus on the evolution of the Clifford-stabilizer group under circuit-level Pauli errors and Pauli- or Clifford-stabilizer measurements.
Later, in Sec.~\ref{sec:scope}, we show that our algorithm can also accommodate additional operations that commonly appear in more general protocols, such as unitary gates for injection, logical non-Clifford gates including transversal non-Clifford gates, and $\CNOT$ gates and new stabilizer measurements used for code growth.

Below, we show that for a broad class of logical magic states, the Clifford-stabilizer group can be updated in a manner closely analogous to the update of an ordinary Pauli-stabilizer group. In this section, we focus on the Clifford-stabilizer group update for logical magic states that are used for diagonal logic gates, which we call \emph{diagonal logical magic states}.
For such magic states, their Clifford stabilizers can be expressed as a product of diagonal Clifford and Pauli-$X$ operators. This class is particularly simple from the perspective of group updates while also encompassing practically important examples.
Later, in Sec.~\ref{sec:clifford_stb_update}, we generalize our simulation method to a broader class of logical magic states.

\subsection{Example:  Surface code \texorpdfstring{$|\bar{T}\rangle$ state}{|Tbar> state}}
\label{subsec:clifford_stb_update_ex_surface}
We first illustrate Clifford-stabilizer simulation with an explicit example that demonstrates the update rules. Specifically, we consider the surface code $|\bar{T}\rangle$ state introduced in Sec.~\ref{subsec:clifford_stabilizer_ex_surface}.  

We first discuss how the Clifford-stabilizer group can be updated when Pauli errors occur on the data qubits.  Suppose an error $X_0X_5$ occurs on the $|\bar{T}\rangle$ state. It flips the sign of $G^Z_2$, and it also conjugates $\bar{H}_{XY}$ as 
\begin{equation}
\widetilde{\bar{H}}_{XY}:=X_0X_5\bar{H}_{XY}(X_0X_5)^\dagger=iZ_0Z_1\bar{H}_{XY}.
\label{eq:h_xy_conjugate}
\end{equation}
Thus, the Clifford-stabilizer group after the error is 
\begin{equation}
\langle G^X_1,\, G^X_2,\,\ldots, G^X_6,G^Z_1,\, -G^Z_2,\,\ldots, G^Z_6, \,\widetilde{\bar{H}}_{XY}\rangle,
\end{equation}
which uniquely specifies the post-error state.
Now, suppose we measure the Clifford stabilizer $\bar{H}_{XY}$ on this state. Note that this Clifford measurement is a type of measurement performed in magic state cultivation~\cite{gpvl-lg4c}.
This operator commutes with all $Z$-type stabilizers $G^Z_1,\, -G^Z_2,\,\ldots, G^Z_6$.
For the $X$ stabilizers, Eq.~\eqref{eq:surface_commutation_x} implies that $\bar{H}_{XY}$ effectively commutes with every $G^X_i$ except $G^X_2$. It effectively anticommutes with $G^X_2$ because
\begin{equation}
[\bar{H}_{XY},G^X_2]=G^Z_2,
\end{equation}
and $G^Z_2$ acts as $-I$ on the current state, since $-G^Z_2$ is one of the stabilizers.
For $\widetilde{\bar{H}}_{XY}$, the group commutator is given by
\begin{align}
[\bar{H}_{XY},\widetilde{\bar{H}}_{XY}]&=[\bar{H}_{XY},D\bar{H}_{XY}]\\
&=[\bar{H}_{XY},D]D[\bar{H}_{XY},\bar{H}_{XY}](D)^{-1}\\
&=[\bar{H}_{XY},D]\\
&=-I,
\label{eq:surface_comm_h}
\end{align}
where $D:=iZ_0Z_1$.
Thus, they anticommute with each other. 
Since the $\bar{H}_{XY}$ measurement effectively anticommutes with some of the current generators, the probabilities of obtaining the $+1$ or $-1$ outcomes are both $50\%$.

We now show how the Clifford-stabilizer group can be updated under this measurement. A key property is that the measurement either effectively commutes or effectively anticommutes with each generator of the Clifford-stabilizer group that uniquely specifies the pre-measurement state. Owing to this structure, the group can be updated analogously to the Pauli-stabilizer group updates in Eqs.~\eqref{eq:pauli_group_update_1} and \eqref{eq:pauli_group_update_2}.  (The correctness of this generalized update rule is established rigorously in Proposition~\ref{prop:clifford_stabilizer_update} below.)
Thus the post-measurement state is uniquely stabilized by the Clifford-stabilizer group
\begin{equation}
\langle G^X_1,\, G^X_2\widetilde{\bar{H}}_{XY},\,\ldots, G^X_6,G^Z_1,\, -G^Z_2,\,\ldots, G^Z_6, \,\pm \bar{H}_{XY}\rangle,
\end{equation}
where the sign of $\bar{H}_{XY}$ is random.

The key observation in Eq.~\eqref{eq:h_xy_conjugate} is that the effect of the error is to dress $\bar{H}_{XY}$ by only Pauli $Z$ terms with a phase, and these $Z$ terms commute with the diagonal Clifford operators $S$, $S^\dagger$, and $CZ$ appearing in $\bar{H}_{XY}$. Hence it is ensured that $\bar{H}_{XY}$ and $\widetilde{\bar{H}}_{XY}$ either commute or anticommute with each other as in Eq.~\eqref{eq:surface_comm_h}.
If there exist errors that dress a generator by Pauli $X$ or $Y$, the measured operator would neither commute nor anticommute with the Pauli-dressed generator. 
However, in the following, we show that every possible error that arises on the data qubits under circuit-level Pauli errors dresses the generators only by $I$ or Pauli $Z$ up to phases.
Possible errors on the data qubits are Pauli errors and also diagonal $S$, $S^\dagger$, and $CZ$ Clifford errors. These Clifford errors arise even under circuit-level Pauli errors due to hook errors when measuring the Clifford operator $\bar{H}_{XY}$ as shown in Fig.~\ref{fig:error_prop}. Here, we consider the standard circuit to measure $\bar{H}_{XY}$ by jointly measuring physical Hermitian operators appearing in Eq.~\eqref{eq:fold_h_operator}, where $SX$ and $S^\dagger X$ are measured as Hermitian operators $e^{-i\pi/4}SX$ and $e^{i\pi/4}S^\dagger X$, respectively.
\begin{figure}[t]
    \centering
    \includegraphics[width=0.95\linewidth]{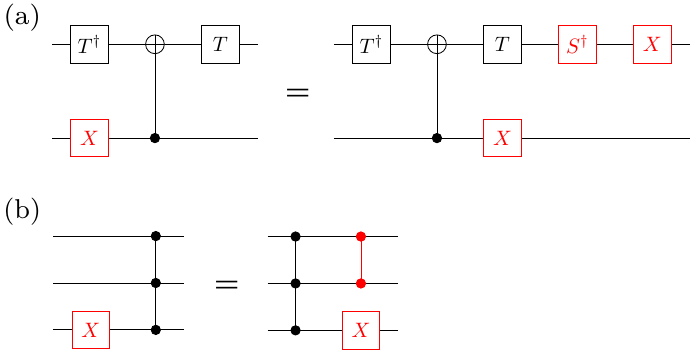}
    \caption{Clifford errors on the data qubits arising from hook errors when measuring the Clifford operator $\bar{H}_{XY}$ by jointly measuring physical Hermitian operators. (a) Circuit for measuring the physical Hermitian operator $T X T^\dagger=e^{-i\pi/4}SX=e^{i\pi/4}XS^\dagger$. The upper and lower wires represent a data and ancilla qubit, respectively. The $X$ error on the ancilla qubit produces a Clifford $S^\dagger$ error on the data qubit. (b) Circuit for measuring the physical Hermitian operator $CZ$. The lower wire corresponds to an ancilla qubit, and the other two wires are data qubits. The $X$ error on the ancilla qubit propagates to a Clifford $CZ$ error on the data qubits.}
    \label{fig:error_prop}
\end{figure}
Now we see how these errors conjugate the operators contained in the generators.
Pauli $Z$ errors commute or anticommute with Pauli operators and commute with the diagonal Clifford operators contained in $\bar{H}_{XY}$. Pauli $X$ errors also commute or anticommute with Pauli operators, but conjugate the diagonal Clifford operators in the following way:
\begin{align}
XSX&=iZS,
\label{eq:diagonal_dressing_ex_1}
\\
XS^\dagger X&=-iZS^\dagger,
\\
X_iCZ_{i,j}X_i&=Z_jCZ_{i,j}.
\end{align}
The diagonal Clifford errors $S$, $S^\dagger$, and $CZ$ commute with the diagonal operators and conjugate the $X$ terms as follows:
\begin{align}
SXS^\dagger&=-iZX,
\\
S^\dagger X S&=iZX,
\\
CZ_{i,j}X_iCZ_{i,j}&=Z_jX_i.
\label{eq:diagonal_dressing_ex_2}
\end{align}
Therefore, all possible errors only dress the generators by $I$ or Pauli $Z$ terms up to phases. Because a phase does not affect commutation, it is guaranteed that a Pauli- or Clifford-stabilizer measurement either commutes or anticommutes with every error-conjugated generator. 

To summarize, we have explicitly demonstrated that the Clifford-stabilizer group of the surface code $|\bar{T}\rangle$ state can be updated in a manner analogous to an ordinary Pauli-stabilizer group under one error and one Pauli- or Clifford-stabilizer measurement. 
This means that we can simulate the evolution of the state by this systematic update of the Clifford-stabilizer group under these error and measurement.

\subsection{Update rules for the Clifford-stabilizer group}
\label{subsec:rules_diagonal}
We now extend this discussion to other diagonal logical magic states and also describe the update rules in greater detail. In particular, we consider commutation and group updates under arbitrary sequences of circuit-level Pauli errors and Pauli- or Clifford-stabilizer measurements, as well as determining measurement outcomes for Pauli- or Clifford-stabilizer measurements that are deterministic. We do not provide rigorous proofs of these update rules in this section.
Later, in Sec.~\ref{sec:clifford_stb_update}, we establish the relevant conditions and properties in a more general setting that includes diagonal and non-diagonal magic states and provide rigorous proofs of the update rules.

We first describe the structure of a Clifford-stabilizer group that uniquely defines an ideal diagonal logical magic state.
In what follows, we use the term \emph{protocol measurements} to refer to measurements of Hermitian Pauli- or Clifford-stabilizer generators in the Clifford-stabilizer group that defines the target logical magic state.
For simplicity, in this section, we assume that each Pauli stabilizer is either an $X$-type or $Z$-type operator. For Pauli-stabilizer codes, this corresponds to logical magic states of Calderbank-Shor-Steane (CSS) codes~\cite{PhysRevLett.77.793,PhysRevA.54.1098}, including the surface code $|\bar{T}\rangle$ state and $\llbracket 8,3,2 \rrbracket$ $|\overline{CCZ}\rangle$ state given in Sec.~\ref{sec:clifford_stabilizer}.
Many previously studied Clifford-stabilizer codes~\cite{vrty-qs5h} also satisfy this requirement.
A pair of generators $G_i$ and $G_j$, where $G_i$ and $G_j$ are either Clifford stabilizers or Pauli $X$-type stabilizers, may not commute as operators, but in this case they effectively commute on the ideal logical magic state since their group commutator is equal to a Pauli $Z$-type stabilizer $G_k^Z$:
\begin{equation}
    [G_i, G_j]=G_k^Z,
    \label{eq:group_commutator_diag}
\end{equation}
where $G_k^Z$ can be a product of Pauli $Z$-type stabilizer generators.
Note that the right-hand side of Eq.~\eqref{eq:group_commutator_diag} must be a Pauli-$Z$-type operator because any Clifford stabilizer is the product of diagonal Clifford and Pauli-$X$ operators, and thus the group commutator can leave only Pauli $Z$ as nontrivial residuals. Hence, the iterated group commutator $[[G_i,G_j],G_k^Z]$ between a triple of stabilizers is trivial. 

In the following, we analyze the evolution of Clifford-stabilizer groups under errors and measurements. 
Throughout, we use the word \emph{original} to refer to entities associated with noiseless target logical magic states.
For simplicity, we suppose that every original Pauli- or Clifford-stabilizer generator $G$ is transversal~\cite{PhysRevLett.102.110502} or 2-fold transversal~\cite{Breuckmann2024foldtransversal,y14y-7kp3}.\footnote{In Sec.~\ref{sec:clifford_stb_update}, we describe how our analysis generalizes to gates that are not transversal.} 
That is, $G$ can be expressed as a tensor product of Pauli and Clifford operators with disjoint support:
\begin{equation}
  G = \bigotimes_j U_j,
    \label{eq:factor_diag}
\end{equation}
where $U_j \in \cC^{(2)}$ is a one- or two-qubit physical Pauli or Clifford operator.
If $U_j \in\cC^{(1)}$, $U_j$ is called a Pauli factor. If $U_j \in \cC^{(2)} \setminus \cC^{(1)}$, $U_j$ is called a Clifford factor.
In this section, for simplicity, we restrict each Clifford factor to a Hermitian diagonal operator or a Hermitian operator written as the product of a diagonal Clifford operator and a Pauli $X$, specifically, $e^{-i\pi/4}SX$, $e^{i\pi/4}S^\dagger X$, and $CZ$.
This covers the Clifford stabilizers of many typical diagonal logical magic states. 
When measuring $G$ in a protocol, we assume that the measurement is implemented as a sequence of controlled-$U_j$ operations, each realized by a controlled gate that may be conjugated by non-Clifford gates acting on the data qubits, as exemplified in Fig.~\ref{fig:error_prop}.
Let $\mathcal L_{\mathrm{P}}$ and $\mathcal L_{\mathrm{C}}$ denote the set of physical Pauli and Clifford factors appearing in any of the original  generators, respectively. 
We assume that all Clifford factors commute with one another for simplicity in this section.

We now describe how errors update the Pauli- or Clifford-stabilizer generators under conjugation. On data qubits, in addition to Pauli errors, Clifford errors $U_j$ in Eq.~\eqref{eq:factor_diag} can effectively occur due to hook errors under circuit-level Pauli errors as shown in Fig.~\ref{fig:error_prop_ancilla_herm}. 
We use the term \emph{Pauli-or-hook errors} to refer to errors on the data qubits that are either Pauli errors arising from the error model or Pauli or Clifford errors induced by hook errors under circuit-level Pauli errors. 
We note that circuit-level Pauli errors can lead to the measurement of an incorrect Clifford operator when the measured Clifford generator is not (fold-)transversal or one of the factors is not Hermitian as detailed in Sec.~\ref{subsubsec:effective}.
However, this does not happen (up to measurement readout errors) in the setting considered in this section, and Pauli-or-hook data errors can be regarded as phenomenological Pauli or Clifford data errors applied to the states between protocol measurements.

\begin{figure}[t]
    \centering
    \includegraphics[width=0.85\linewidth]{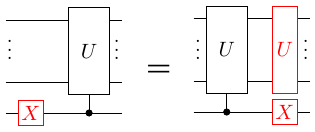}
    \caption{Propagation of a hook error through a controlled-$U$ gate, where $U$ is a Hermitian unitary. The bottom wire represents the ancilla qubit, while the remaining wires represent the data qubits.}
    \label{fig:error_prop_ancilla_herm}
\end{figure}

\subsubsection{Pauli dressings}

In our setting where Clifford factors in $\mathcal L_{\mathrm{C}}$ commute with one another, the only way an error can conjugate a Pauli or Clifford factor is through a \emph{Pauli dressing}. 
That is, the conjugated factor can be written as a product of a Pauli operator (including the identity) and the original factor multiplied by a phase in $\{\pm 1,\pm i\}$. 
In the following, we explain how Pauli and Clifford errors create Pauli dressings. 
When a single-qubit Pauli error $E_\mathrm{P}$ occurs, each Clifford factor $C\in\mathcal{L}_\mathrm{C}$ is conjugated by $E_\mathrm{P}$ and a Pauli dressing is obtained:
\begin{equation}
  E_\mathrm{P} CE_\mathrm{P}^\dagger = D_C(E_\mathrm{P})C,
  \label{eq:dressed_clifford_factor_diag}
\end{equation}
where $D_C(E_\mathrm{P})$ is a Pauli dressing, i.e., a trivial or non-trivial Pauli operator multiplied by a phase in $\{\pm 1,\pm i\}$, defined by
\begin{equation}
  \label{eq:dressing_Pauli_error_diag}
  D_C(E_\mathrm{P}):=E_\mathrm{P} CE_\mathrm{P}^\dagger C^\dagger=[E_\mathrm{P}, C].
\end{equation}
When a Clifford error $E_\mathrm{C}$ occurs, each Pauli factor $P \in \mathcal{L}_\mathrm{P}$ is conjugated by $E_\mathrm{C}$ and a Pauli dressing is acquired:
\begin{equation}
  E_\mathrm{C} P E_\mathrm{C}^\dagger=D_P(E_\mathrm{C})P, 
  \label{eq:dressed_pauli_factor_diag}
\end{equation}
where $D_P(E_\mathrm{C})$ is a Pauli dressing
\begin{equation}
  \label{eq:dressing_Clifford_error_diag}
  D_P(E_\mathrm{C})=E_\mathrm{C}PE_\mathrm{C}^\dagger P^\dagger=[E_\mathrm{C},P].
\end{equation}
Let $\mathcal D$ denote the set of possible Pauli dressings induced by any single-qubit Pauli error or any Clifford error on the data qubits under circuit-level Pauli errors.
A nontrivial Pauli operator contained in each Pauli dressing $D \in \mathcal D$ is a single-qubit Pauli $Z$, as explicitly shown in Eqs.~\eqref{eq:diagonal_dressing_ex_1}--\eqref{eq:diagonal_dressing_ex_2}.

Now it is guaranteed that every possible Pauli dressing either commutes or anticommutes with any original stabilizer generator $G$:
\begin{equation}
  [G,D]=\pm I, \qquad  \forall D\in\mathcal{D}.
  \label{eq:commute_anti_condition_stab_diag}
\end{equation}  
This is because the possible nontrivial Pauli dressings are only Pauli $Z$ up to phases, which commute with all diagonal Clifford operators contained in the Clifford factors. Since other components contained in the factors are Pauli, Eq.~\eqref{eq:commute_anti_condition_stab_diag} holds.

\subsubsection{Commutation relation for the first measurement}

Below, we explain that it is guaranteed that the first protocol measurement performed after Pauli-or-hook errors either effectively commutes or effectively anticommutes with any Pauli-dressed original stabilizer generators due to Eq.~\eqref{eq:commute_anti_condition_stab_diag}. 
Let $G$ be an original stabilizer generator and let $M$ be a protocol measurement, i.e., one of the original stabilizer generators. Suppose a current generator is Pauli-dressed due to Pauli-or-hook errors:
\begin{equation}
  \tilde{G}=DG,
\end{equation}
where $D$ is an accumulated Pauli dressing.
The group commutator between $M$ and $\tilde{G}$ can be written as
\begin{align}
  [M,\tilde{G}]
  &=
  [M,DG]\\
  &=
  [M,D]\,D[M,G]D^\dagger.
  \label{eq:commutator_decompose_diag}
\end{align}
In Eq.~\eqref{eq:commutator_decompose_diag}, $[M,G]$ is either the identity $[M,G]=I$ or a Pauli $Z$-type stabilizer $[M,G]=G^Z$. 
Since the Pauli dressing $D$ is $I$ or a $Z$-type Pauli up to a phase, $D$ and $[M,G]$ commute with each other, i.e., 
\begin{equation}
  D[M,G]D^\dagger=[M,G].
\end{equation}
Also, because $M$ is one of the original stabilizer generators, due to Eq.~\eqref{eq:commute_anti_condition_stab_diag}, $[M,D]=\pm I$.
Thus, it is guaranteed that $[M,\tilde{G}]$ is either $[M,\tilde{G}]=\pm I$ or $[M,\tilde{G}]=\pm G^Z$. 
Furthermore, possible Pauli-or-hook errors on the data qubits are either Pauli errors, which commute or anticommute with $G^Z$, or diagonal Clifford errors, which commute with $G^Z$. Accordingly, either $G^Z$ or $-G^Z$ is a current stabilizer.
Therefore, it is guaranteed that $M$ and $\tilde{G}$ either effectively commute or effectively anticommute on the current state $|\psi\rangle$:
\begin{equation}
  [M,\tilde{G}]|\psi\rangle=\pm|\psi\rangle.
  \label{eq:M-dressed-G_diag}
\end{equation}
Note that any Pauli-$Z$ stabilizer remains a stabilizer of the post-measurement state up to a sign since a protocol measurement and a Pauli-$Z$ stabilizer commute with each other by definition of the Clifford-stabilizer group.

\subsubsection{Anticommuting measurements}

When a measurement effectively commutes with every current stabilizer generator, the Clifford-stabilizer group remains unchanged. We now explain that, when a measurement effectively anticommutes with some current stabilizer generators, the group can be updated in a manner analogous to the update rules for an ordinary Pauli-stabilizer group.
Suppose a protocol measurement $M$ effectively anticommutes with at least one current generator.
Then, the group is updated by first choosing a pivot generator $G'_p$ that effectively anticommutes with $M$. Next, every other generator $G'_j$ that effectively anticommutes with $M$ is replaced by
\begin{equation}
G'_j\leftarrow G'_jG'_p.
\label{eq:anti_update_diag}
\end{equation}
Finally, replace the pivot $G'_p$ by the measured operator $M$ with a uniformly random sign:
\begin{equation}
G'_p\leftarrow (-1)^b M,
\qquad b\in\{0,1\}.
\label{eq:anti_update_diag_meas}
\end{equation}
The reason this update rule is correct is that, it is guaranteed that $G'_jG'_p$ effectively commutes with the measurement $M$ on the post-measurement state if both pre-measurement stabilizers $G'_j$ and $G'_p$ effectively anticommute with $M$ on the pre-measurement state. See Proposition~\ref{prop:clifford_stabilizer_update} for a rigorous proof of this.

\subsubsection{Commutation relations for later measurements}

We next explain that a protocol measurement either effectively commutes or effectively anticommutes with the current stabilizer generators on the current state after an arbitrary sequence of Pauli-or-hook errors and protocol measurements. This ensures that we can always update the current Clifford-stabilizer group in a way analogous to the update rules for an ordinary Pauli-stabilizer group.
After one protocol measurement, any tracked generator of the current Clifford-stabilizer group can be written as a product of at most two Pauli-dressed original stabilizer generators using Eqs.~\eqref{eq:anti_update_diag} and~\eqref{eq:anti_update_diag_meas}. Subsequent Pauli-or-Clifford errors only change the Pauli dressings.

We now consider how the group commutator between a protocol measurement $M$ and the product $\tilde{G}_1\cdots \tilde{G}_s$ of Pauli-dressed original stabilizer generators for some $s$ can be expressed. It is expressed as 
\begin{align}
  [M,\tilde{G}_1\cdots \tilde{G}_s]
  &=
  \prod_{k=1}^{s}
  (\tilde{G}_1 \cdots \tilde{G}_{k-1})
  [M,\tilde{G}_{k}]
  (\tilde{G}_1 \cdots \tilde{G}_{k-1})^\dagger \\
  &=
  \prod_{k=1}^{s}
  [M,\tilde{G}_{k}],
  \label{eq:general_product_comm_prod_diag}
\end{align}
since each $[M,\tilde{G}_{k}]$ is either the identity or a Pauli-$Z$ stabilizer up to a $\pm 1$ sign due to Eq.~\eqref{eq:commutator_decompose_diag}, which commutes with the Pauli-$Z$-dressed original generators $\tilde{G}_1, \ldots, \tilde{G}_{k-1}$. Note that a Pauli-$Z$ stabilizer commutes with every original generator.

By Eq.~\eqref{eq:general_product_comm_prod_diag} and the fact that every Pauli-$Z$ stabilizer remains a current stabilizer up to a sign, it follows that, for the next protocol measurement $M$ and any current stabilizer generator $G'$, 
\begin{equation}
  [M,G']|\psi\rangle=\pm|\psi\rangle.
  \label{eq:M-S-state_diag}
\end{equation}
By repeatedly applying this argument, Eq.~\eqref{eq:M-S-state_diag} holds for a protocol measurement and any current generator after an arbitrary sequence of Pauli-or-hook errors and protocol measurements.
Note that $M$ can be any original generator for each protocol measurement.

\begin{figure*}[t]
    \centering
    \includegraphics[width=0.9\linewidth]{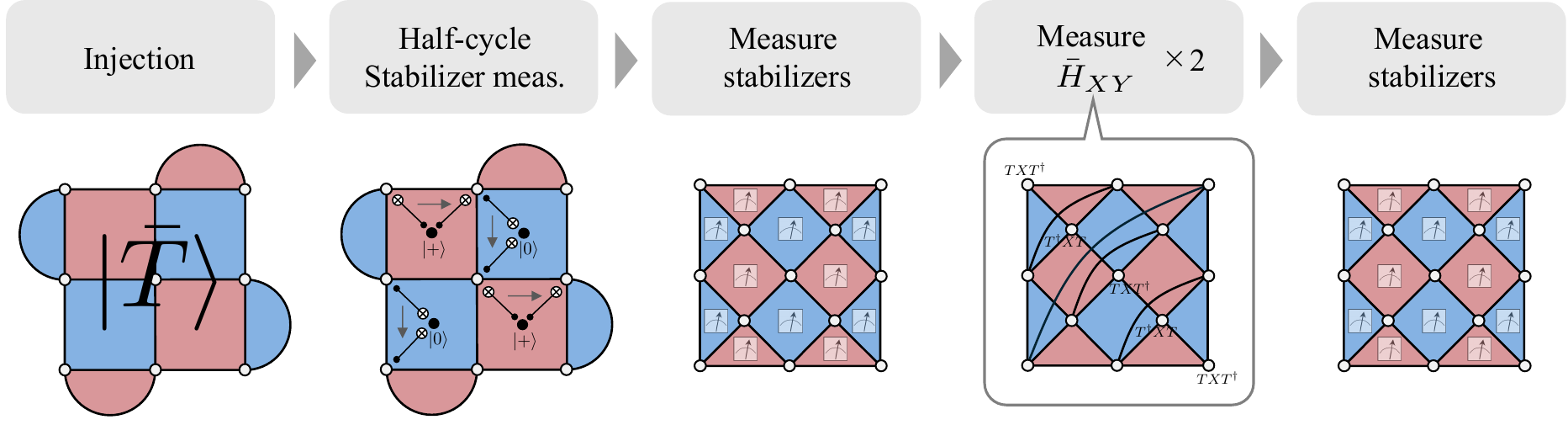}
    \caption{Fold-transversal magic state cultivation protocol with fault distance $3$ up to the end of the cultivation stage. ``Stabilizer'' in this figure refers to standard Pauli stabilizers. At the injection stage, we prepare the $|\bar{T}\rangle$ state encoded in the $d=3$ rotated surface code in a non-fault-tolerant manner. The code is then grown to the $d=3$ regular surface code by applying $\CNOT$ gates with appropriate ancilla qubits, which corresponds to a half-cycle of stabilizer measurements. On the regular surface code, all stabilizers indicated by the measurement symbols are measured once. Then $\bar{H}_{XY}$ is measured twice, followed by a round of stabilizer measurements. We post-select on runs in which all measurement outcomes are trivial.}
    \label{fig:cultivation}
\end{figure*}

\subsubsection{Deterministic measurements}

So far, we have seen that a protocol measurement either effectively commutes or effectively anticommutes with every generator of the current Clifford-stabilizer group after any sequence of Pauli-or-hook errors and protocol measurements. 
When the measurement effectively anticommutes with at least one current generator, the group is updated according to the rules given in Eqs.~\eqref{eq:anti_update_diag} and ~\eqref{eq:anti_update_diag_meas}. On the other hand, when the measurement effectively commutes with every current stabilizer, its outcome is deterministic and the Clifford-stabilizer group remains unchanged.

We have not yet addressed whether, in the deterministic case, the measurement outcome can be computed directly from the currently tracked generators.
For a standard Pauli-stabilizer group, the deterministic outcome can be computed from the current Pauli stabilizer generators. 
However, for a general Clifford-stabilizer group, this is no longer automatic as detailed in Sec.~\ref{subsec:deterministic}.
Nevertheless, we show that the deterministic outcome can be computed from the current generators for typical logical magic states including the examples in Sec.~\ref{sec:clifford_stabilizer} due to their favorable algebraic structure specified in Definition~\ref{def:primary_support_exactness}, which we assume to hold in this section.
Specifically, if the outcome of a protocol measurement $M$ is deterministic, a product of current generators is equal to $(-1)^m M$, where $(-1)^m$ is the deterministic measurement outcome.
The proof of this is given later in Proposition~\ref{prop:deterministic}, which establishes the corresponding result in a more general setting.

\subsubsection{From update rules to efficient simulation}

In conclusion, the Clifford-stabilizer group defining logical magic states used for diagonal logic gates can be updated in a manner analogous to the update rules for an ordinary Pauli-stabilizer group under circuit-level Pauli errors and protocol measurements.
This provides a systematic framework for tracking the evolving states. However, the update rules alone do not yield an efficient simulation algorithm, since an explicit algorithm for implementing the Clifford-stabilizer group update itself is required. In Sec.~\ref{sec:algorithm}, we provide such an algorithm, with time and space complexities that are polynomial in parameters relevant to logical magic state protocols, such as the number of physical qubits, the number of logical qubits, and the number of physical non-Clifford gates, with no direct dependence on the stabilizer or Pauli rank of the target logical magic state.
Specifically, in Sec.~\ref{sec:algorithm}, we show that there exists a Pauli-stabilizer group update whose measurement probabilities exactly reproduce those of the Clifford-stabilizer group update. Equivalently, a noisy non-Clifford circuit subject to circuit-level Pauli errors that implements protocol measurements for a logical magic state preparation can be mapped to a Clifford circuit that exactly reproduces its measurement probabilities. 
Our algorithm therefore reduces the simulation of a non-Clifford circuit to the simulation of a Clifford circuit. 
The details of the algorithm and its complexity analysis are provided in Sec.~\ref{sec:algorithm}.

\section{Application to magic state cultivation}\label{sec:simulation}
\label{sec:numerics}
\begin{figure}[t]
    \centering
    \includegraphics[width=0.9\linewidth]{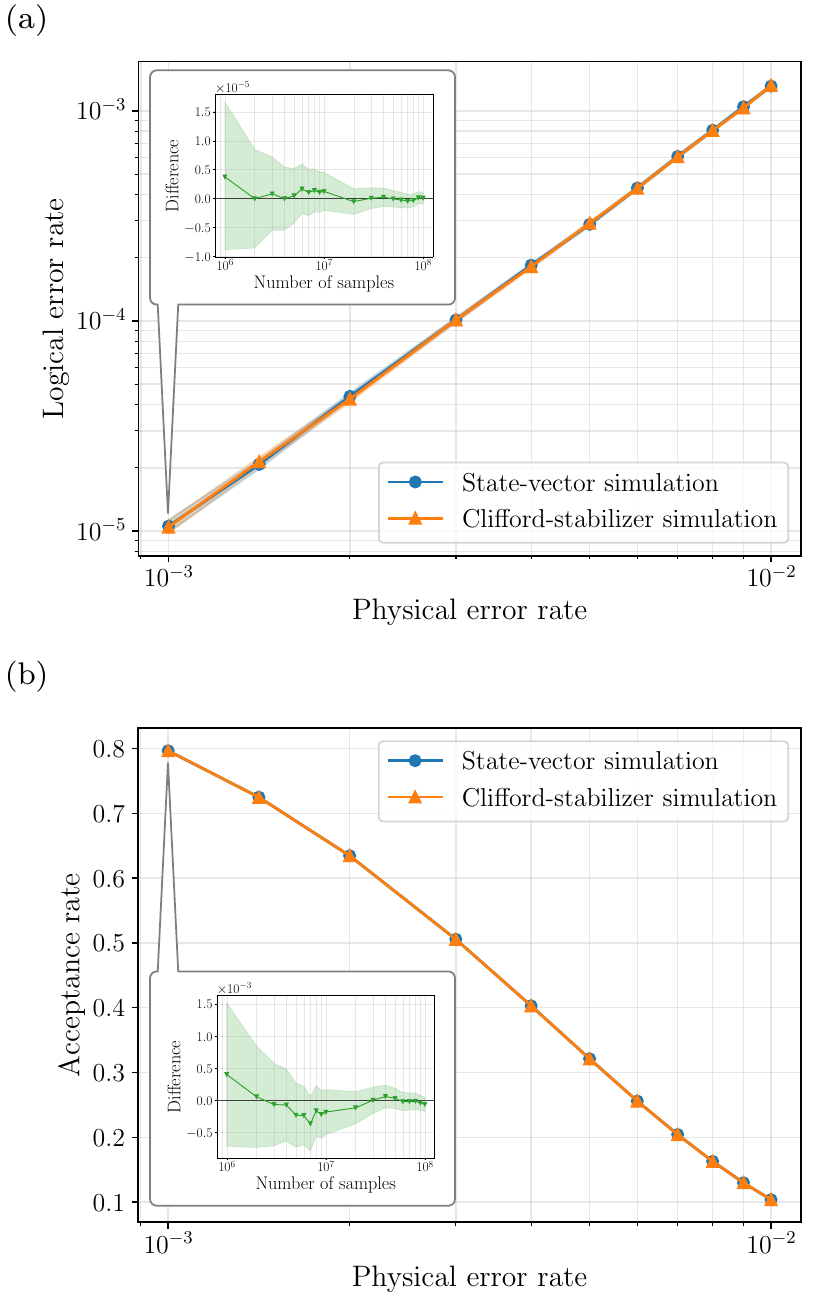}
    \caption{Comparison of the logical error rate and acceptance rate for the $f=3$ fold-transversal magic state cultivation obtained using state-vector and the proposed Clifford-stabilizer simulations. The number of Monte Carlo samples is $10^8$. (a) Logical error rate. (b) Acceptance rate. The insets show the corresponding differences of the values from the two simulation methods as functions of the number of Monte Carlo samples at the physical error rate $p=10^{-3}$. The differences are computed as state-vector simulation minus Clifford-stabilizer simulation.}
    \label{fig:numerics_statevector}
\end{figure}
\begin{figure}[t]
    \centering
    \includegraphics[width=0.9\linewidth]{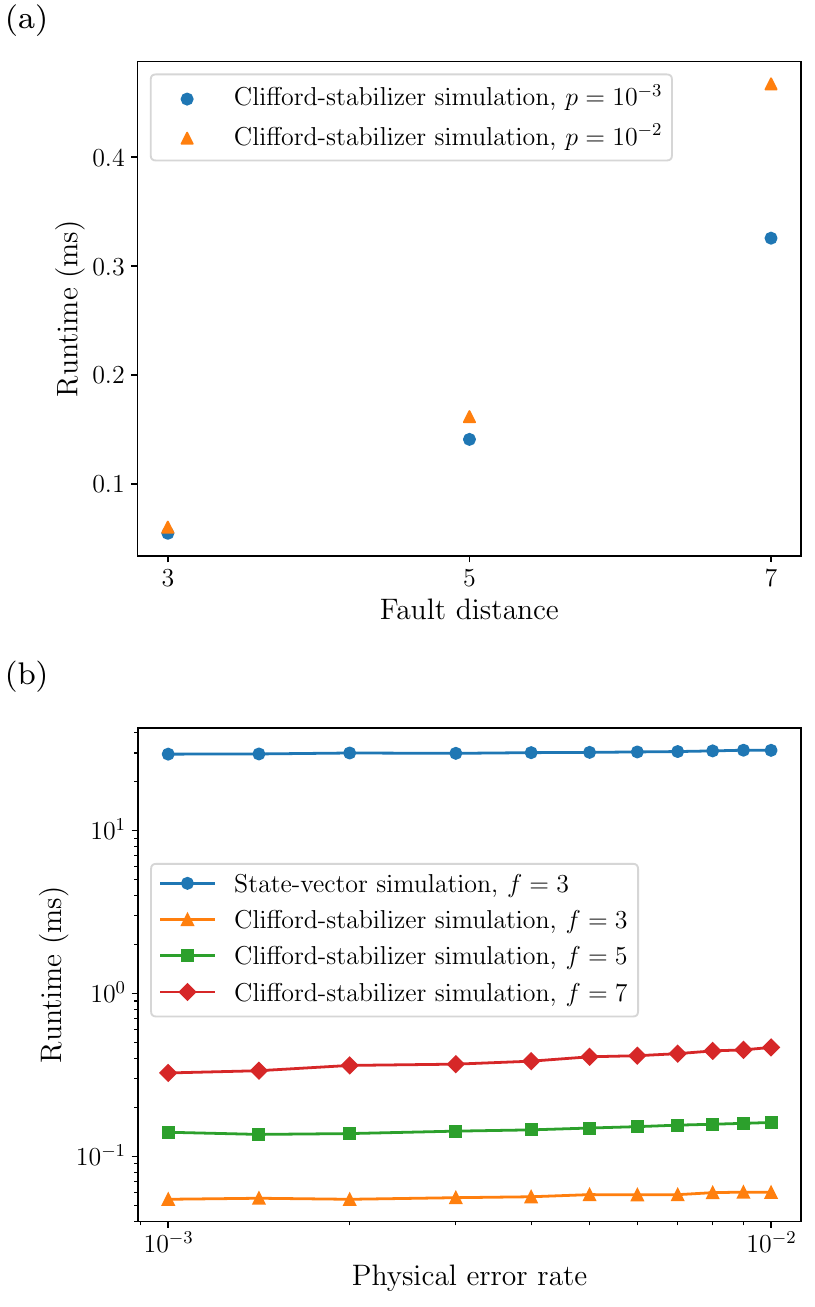}
    \caption{Runtime per Monte Carlo shot for the fold-transversal magic state cultivation. The runtime is measured as the average over $1000$ runs. (a) Runtime of the proposed Clifford-stabilizer simulation as a function of the fault distance $f$ with fixed physical error rates $p=10^{-3}$ and $p=10^{-2}$. (b) Runtime of the state-vector simulation with $f=3$ and of Clifford-stabilizer simulation with $f=3,5,$ and $7$, shown as a function of the physical error rate.}
    \label{fig:runtime}
\end{figure}
In this section, we demonstrate that the update rules for the diagonal case, together with the algorithm presented later in Sec.~\ref{sec:algorithm}, can be used for an exact and efficient simulation of important FTQC protocols. Specifically, we simulate the fold-transversal surface code cultivation~\cite{gpvl-lg4c}, one of the promising magic state cultivation protocols.
This protocol prepares the high-fidelity $|\bar{T}\rangle$ state based on the regular surface code described in Secs.~\ref{subsec:clifford_stabilizer_ex_surface} and~\ref{subsec:clifford_stb_update_ex_surface} with the measurements of the Pauli and Clifford stabilizers. 
Magic state cultivation~\cite{gidney2024magicstatecultivationgrowing,thxx-njr6,9kys-3whh,gpvl-lg4c,p8tw-6kq9,claes2025cultivatingtstatessurface,ch5r-cnfq,hirano2026efficientmagicstatecultivation,hetenyi2026constantdepthmagicstate} has attracted considerable interest because of its potential to substantially reduce the cost of logical magic state preparation in FTQC. However, exact simulation at larger, practically relevant fault distances, which is crucial for reliably assessing the performance and cost of the protocols, remains challenging. Due to this challenge, previous studies~\cite{gidney2024magicstatecultivationgrowing,thxx-njr6,9kys-3whh,gpvl-lg4c,p8tw-6kq9,claes2025cultivatingtstatessurface,ch5r-cnfq,hirano2026efficientmagicstatecultivation,hetenyi2026constantdepthmagicstate} have often relied on approximations or exact simulations restricted to small instances.

We simulate the fold-transversal surface code cultivation protocols with fault distances $f\in \{3, 5, 7\}$ up to the end of the cultivation stage under circuit-level Pauli noise.  While we do not simulate the escape stage, we note that our algorithm can also simulate the escape stage as detailed in Sec.~\ref{sec:scope}.
The protocols we simulated with fault distances $f=3$ and $f=5$ are based on those in Ref.~\cite{gpvl-lg4c} with slight modifications detailed in Appendix~\ref{app_subsec:cultivation_protocol}. 
While a fault distance $f=7$ protocol was not developed in previous work, we construct it by extending the $f=3$ and $f=5$ protocols.
We illustrate the procedures of the $f=3$ protocol in Fig.~\ref{fig:cultivation}, and the details of the procedures including the $f=5$ and $f=7$ protocols and the simulation setup are provided in Appendices~\ref{app_subsec:cultivation_protocol} and~\ref{app_subsec:cultivation_setup}.
Although we present the Clifford-stabilizer group update only for protocol measurements in Sec.~\ref{sec:clifford_stb_update_diagonal}, we note that our framework can also simulate the injection implemented via a unitary encoding circuit, as detailed in Sec.~\ref{sec:scope}.
To execute a Clifford circuit that exactly reproduces the measurement probabilities of a target non-Clifford circuit, we use the Clifford simulator \texttt{Stim}~\cite{Gidney2021stimfaststabilizer}. 

First, to validate the exactness of our efficient simulation method, we compare the logical error rates and acceptance rates of the $f=3$ protocol obtained using our Clifford-stabilizer simulation method with those obtained using state-vector simulation.
Those values are estimated by performing Monte Carlo simulation.
For state-vector simulation, we use \texttt{Qulacs}~\cite{Suzuki2021qulacsfast}.
Simulation results are shown in Fig.~\ref{fig:numerics_statevector}.  
The results agree within the error bars.
As further evidence of the exactness, the insets in Fig.~\ref{fig:numerics_statevector} show the corresponding differences between the state-vector and Clifford-stabilizer simulations at the physical error rate $p=10^{-3}$ as functions of the number of Monte Carlo samples. 
As expected, the results from the two simulation methods converge to the same values as the number of samples increases.

\begin{figure}[t]
    \centering
    \includegraphics[width=0.9\linewidth]{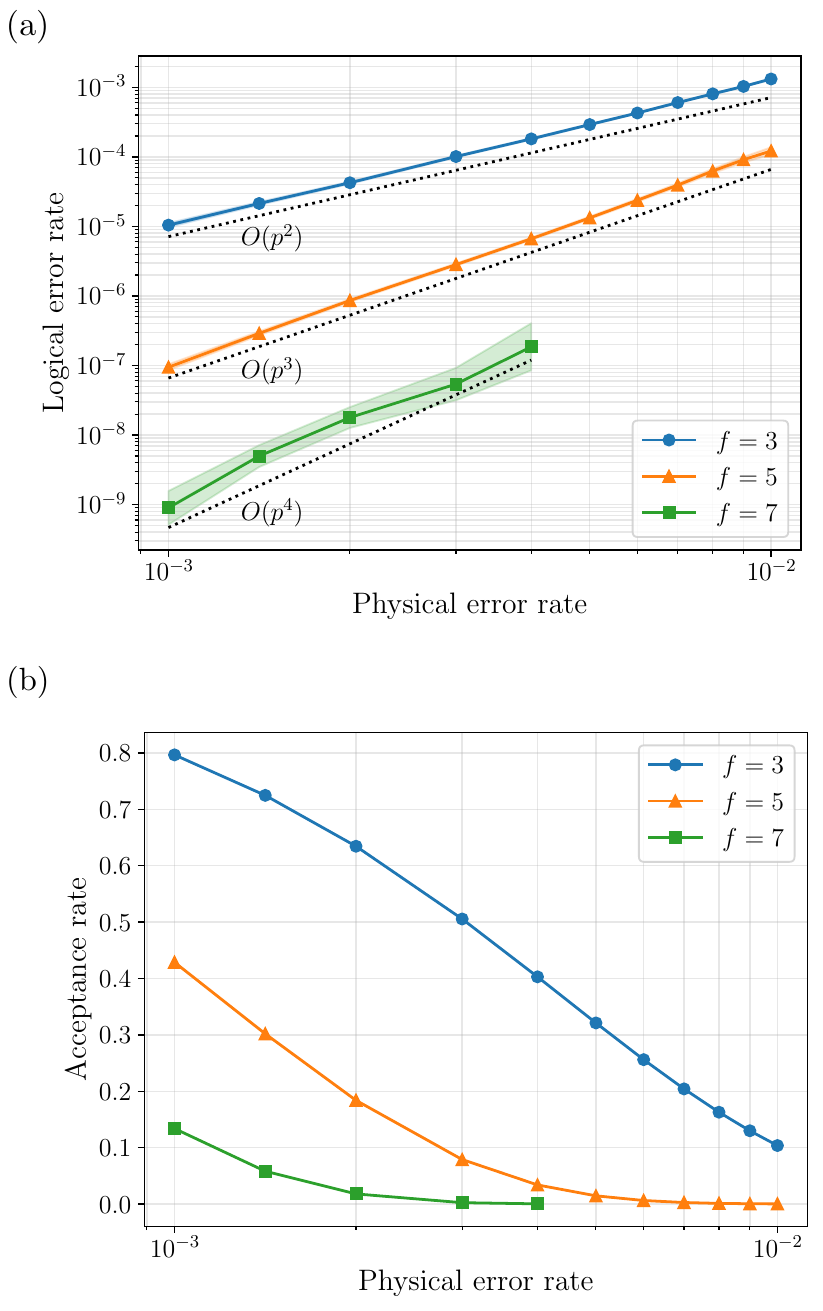}
    \caption{Logical error rate and acceptance rate for the fold-transversal magic state cultivation with FEC performed at the end of the protocol. The numbers of Monte Carlo samples are $10^{8}$ for $f=3$, $10^{10}$ for $f=5$, and $10^{11}$ for $f=7$. (a) Logical error rate. The dotted lines are reference guides indicating $O(p^{(f+1)/2})$ scaling for each $f$ and are not fits to the data. (b) Acceptance rate.}
    \label{fig:numerics_ec}
\end{figure}
\begin{figure}[t]
    \centering
    \includegraphics[width=0.9\linewidth]{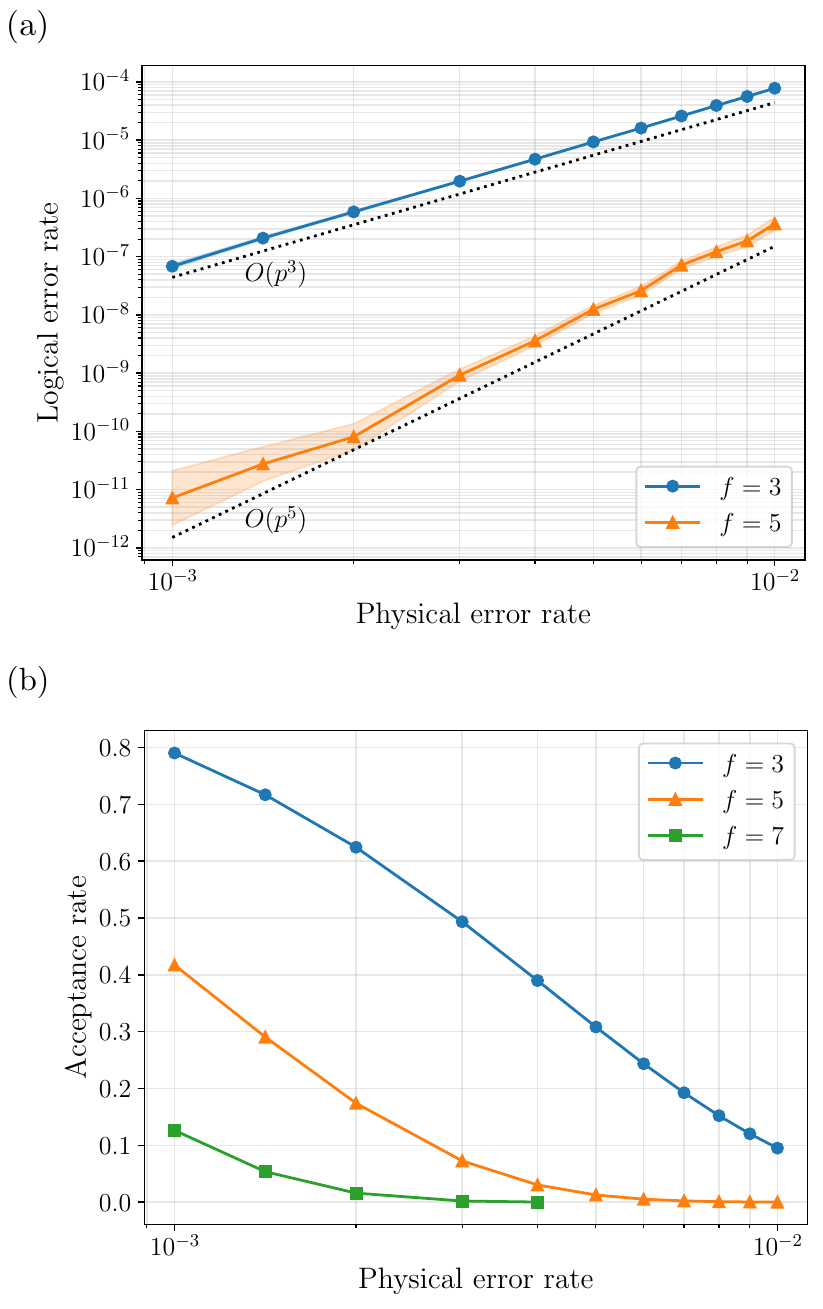}
    \caption{Logical error rate and acceptance rate for the fold-transversal magic state cultivation with FED performed at the end of the protocol. The numbers of Monte Carlo samples are $10^{10}$ for $f=3$, $10^{12}$ for $f=5$, and $10^{10}$ for $f=7$. (a) Logical error rate. The dotted lines are reference guides indicating $O(p^{f})$ scaling for each $f$ and are not fits to the data. (b) Acceptance rate. }
    \label{fig:numerics_ed}
\end{figure}
We highlight one subtle but important aspect of the simulations used for Fig.~\ref{fig:numerics_statevector}.
To evaluate whether the output state has a logical error for each Monte Carlo shot, we perform a fictitious (i.e., noiseless) round of stabilizer measurements at the end of the protocol to remove potentially corrected or detected errors.
The state is guaranteed to lie within the code space after discarding the state or applying a correction based on the noiseless syndrome.
This is a common procedure in numerical simulations of QEC protocols to evaluate whether a logical error has occurred.
For the simulation up to the end of the cultivation stage, it is common in the literature to post-select on the trivial syndrome of the fictitious measurement round. We refer to this as fictitious error detection (FED).
However, for the purpose of comparing the results obtained by Clifford-stabilizer simulation with those obtained by the state-vector simulation, this convention leads to such low logical error rates that the state-vector simulation cannot obtain enough samples with sufficiently small error bars even for the $f=3$ protocol. 
Thus, in Fig.~\ref{fig:numerics_statevector}, we instead apply a correction based on the noiseless syndrome, which we refer to as fictitious error correction (FEC). 
This allows the state-vector simulation to estimate the logical error rates with sufficient precision.
We emphasize that, as long as the same setting is used for both simulation methods, the choice between FEC and FED is unimportant for demonstrating the exactness of our efficient simulation method. 

Having validated the exactness of our simulation method, we can now analyze the performance of magic state cultivation under circuit-level Pauli noise through simulation in regimes that were not previously accessible. 
In Fig.~\ref{fig:runtime}(a), we first show the runtime of the simulation at fixed physical error rates as a function of fault distance of the protocol.
We observe that the runtime increases only modestly with the size of the protocol. 
Note that the numbers of qubits used are $16$ for $f=3$, $47$ for $f=5$, and $99$ for $f=7$. Also, the numbers of non-Clifford gates in the circuits are $29$ for $f=3$, $97$ for $f=5$, and $221$ for $f=7$. 
Indeed, as detailed in Sec.~\ref{subsec:complexity}, the runtime scales polynomially with relevant parameters, such as the number of qubits and the number of non-Clifford gates. 
We observe in Fig.~\ref{fig:runtime}(a) that the runtime also increases with the physical error rate. This is because the runtime depends on the number of sampled errors, but this dependence is only linear as explained in Sec.~\ref{subsec:complexity}.
Regarding the comparison with the state-vector simulation, as shown in Fig.~\ref{fig:runtime}(b), the runtime of Clifford-stabilizer simulation is over two orders of magnitude smaller than that of the state-vector simulation even for $f=3$. 
The runtimes of the state-vector simulation for $f=5$ and $f=7$ are not measurable since the numbers of qubits used in the protocols are beyond the memory feasibility of state-vector simulation in our classical computers. 

In Fig.~\ref{fig:numerics_ec}, we show the logical error rates and acceptance rates of the $f=3,5,$ and $7$ protocols with FEC. The numbers of Monte Carlo samples are $10^{8}$ for $f=3$, $10^{10}$ for $f=5$, and $10^{11}$ for $f=7$. For the $f=7$ protocol, we only simulate physical error rates up to a maximum value of $p=4\times10^{-3}$ because the acceptance rate is too low to estimate the logical error rate when $p$ is higher. 
In Fig.~\ref{fig:numerics_ed}, we also show the logical error rates and acceptance rates with FED, which is a more common convention in the setting of magic state cultivation up to the cultivation stage.
The numbers of Monte Carlo samples are $10^{10}$ for $f=3$, $10^{12}$ for $f=5$, and $10^{10}$ for $f=7$. 
For the $f=7$ protocol, the logical error rates are too low for logical error events to be observed within the sampled data, so the logical error rates are not shown. 
We emphasize that this limitation is separate from the cost of the proposed simulation method and is instead the limitation of Monte Carlo simulations in general for reliably estimating very low logical error rates even if each run is efficiently simulable.

These results demonstrate that we can efficiently and exactly simulate the magic state cultivation protocols up to fault distance $7$.
This is a regime that prior work could not simulate; in Ref.~\cite{gpvl-lg4c}, even an $f=5$ protocol is computationally too expensive to exactly simulate.

\section{Clifford-stabilizer simulation:  general case}\label{sec:clifford_stb_update}
In Sec.~\ref{sec:clifford_stb_update_diagonal}, we showed that, for protocol measurements, Clifford-stabilizer groups for diagonal logical magic states can be updated in an analogous way to ordinary Pauli-stabilizer groups under circuit-level Pauli errors. 
We also introduced several simplifying assumptions. In this section, we generalize the Clifford-stabilizer group update to logical magic states beyond the diagonal setting, while also relaxing these simplifying assumptions. 
In particular, we allow Clifford operators appearing in the Clifford factors to be non-diagonal, Clifford factors to not commute with one another, individual factors to be non-Hermitian, and protocol measurements to be non-transversal.

\subsection{Sufficient conditions on commutation relations}
We first provide sufficient conditions under which protocol measurements always either effectively commute or effectively anticommute with every generator of the current Clifford-stabilizer group under circuit-level Pauli errors.

\subsubsection{Pauli dressings}
\label{subsubsec:commutation}

A Pauli or Clifford generator $G$ of an original Clifford-stabilizer group is expressed as a product of physical Pauli and Clifford operators:
\begin{equation}
  \label{eq:factor}
  G = \prod_j U_j,
\end{equation}
where $U_j \in \cC^{(2)}$ is a physical Pauli or Clifford operator.
If $U_j \in\cC^{(1)}$, $U_j$ is called a Pauli factor. If $U_j \in \cC^{(2)} \setminus \cC^{(1)}$, $U_j$ is called a Clifford factor.
Each Pauli factor is a one-qubit operator, and each Clifford factor is a one- or two-qubit operator.
When measuring $G$ in a protocol, the measurement is implemented as a sequence of controlled-$U_j$ operations, each realized by a controlled gate that may be conjugated by non-Clifford gates acting on the data qubits.
Let $\mathcal L_{\mathrm{P}}$ and $\mathcal L_{\mathrm{C}}$ denote the set of physical Pauli and Clifford factors appearing in any of the original generators, respectively. 

In a circuit performing protocol measurements under circuit-level Pauli errors, possible errors on data qubits are Pauli operators or certain Clifford operators. Clifford errors on the data qubits arise from hook errors, as shown in Fig.~\ref{fig:error_prop_ancilla}. Thus, each possible Clifford error is equal to a Clifford factor having the same support.
Note that when $C$ is not Hermitian, a controlled-$(C^\dagger)^2$ gate is also produced, changing the \textit{effectively measured} operator. For the sake of clarity in our proofs, we set aside this effect for the moment and account for it in Sec.~\ref{subsubsec:effective}.

\begin{figure}[t]
    \centering
    \includegraphics[width=0.85\linewidth]{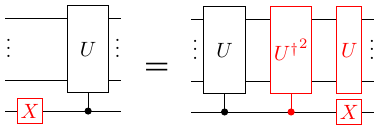}
    \caption{Hook-error propagation through a controlled-$U$ gate, with $U$ an arbitrary unitary. The bottom wire denotes the ancilla qubit, and the remaining wires denote the data qubits.}
    \label{fig:error_prop_ancilla}
\end{figure}

We first consider the cases where any two Clifford factors $C, C' \in\mathcal L_{\mathrm{C}}$ either commute or anticommute with each other:
\begin{equation}
  \label{eq:Clifford-Clifford-commutation}
  [C,C']=\pm I.
\end{equation}
In this setting, an error can conjugate a Pauli or Clifford factor only by introducing a \emph{Pauli dressing}.
That is, the conjugated factor can be expressed as the product of a Pauli operator (possibly the identity), the original factor, and a phase in $\{\pm 1, \pm i\}$.
This includes the case where the factor remains unchanged and the case where the sign is flipped, which occur when the error and the factor commute or anticommute with each other, respectively. 
In the following, we describe how Pauli and Clifford errors give rise to Pauli dressings specifically. When a single-qubit Pauli error $E_\mathrm{P}$ occurs, conjugating each Clifford factor $C\in\mathcal{L}_\mathrm{C}$ by $E_\mathrm{P}$ produces a Pauli dressing:
\begin{equation}
  E_\mathrm{P} CE_\mathrm{P}^\dagger = D_C(E_\mathrm{P})C,
  \label{eq:dressed_clifford_factor}
\end{equation}
where $D_C(E_\mathrm{P})$ denotes a Pauli dressing, namely, a trivial or non-trivial Pauli operator multiplied by a phase in $\{\pm 1,\pm i\}$, defined by
\begin{equation}
  \label{eq:dressing_Pauli_error}
  D_C(E_\mathrm{P}):=E_\mathrm{P} CE_\mathrm{P}^\dagger C^\dagger=[E_\mathrm{P}, C].
\end{equation}
We assume that the possible Pauli dressings $D_C(E_\mathrm{P})$ for any single-qubit Pauli error $E_\mathrm{P}$ are Pauli operators of weight at most one up to phases.
Similarly, when a Clifford error $E_\mathrm{C}$ occurs, conjugating each Pauli factor $P \in \mathcal{L}_\mathrm{P}$ by $E_\mathrm{C}$ produces a Pauli dressing:
\begin{equation}
  E_\mathrm{C} P E_\mathrm{C}^\dagger=D_P(E_\mathrm{C})P, 
  \label{eq:dressed_pauli_factor}
\end{equation}
where $D_P(E_\mathrm{C})$ is a Pauli dressing defined by
\begin{equation}
  \label{eq:dressing_Clifford_error}
  D_P(E_\mathrm{C})=E_\mathrm{C}PE_\mathrm{C}^\dagger P^\dagger=[E_\mathrm{C},P].
\end{equation}
A Pauli dressing is non-identity up to a phase when an error and a factor neither commute nor anticommute with each other.
Under the current assumption of Eq.~\eqref{eq:Clifford-Clifford-commutation}, dressings of Clifford factors induced by Clifford errors are $\pm I$.

Let $\mathcal D$ denote the set of possible Pauli dressings produced by any single-qubit Pauli error or any Clifford error on the data qubits under circuit-level Pauli errors.
For a protocol measurement to either effectively commute or effectively anticommute with every current generator, the corresponding measured original generator is required to either commute or anticommute with the Pauli dressings of the current generators.
To ensure this, we require that every possible Pauli dressing either commutes or anticommutes with every Clifford factor appearing in any of the original generators:
\begin{equation}
  [C,D]=\pm I, \qquad \forall C\in\mathcal{L}_\mathrm{C}, \quad \forall D\in\mathcal{D}.
  \label{eq:commute_anti_condition}
\end{equation} 
If Eq.~\eqref{eq:commute_anti_condition} is satisfied, every possible Pauli dressing either commutes or anticommutes with any original generator $G$:
\begin{equation}
  [G,D]=\pm I, \qquad  \forall D\in\mathcal{D}.
  \label{eq:commute_anti_condition_stab}
\end{equation} 
We now establish conditions under which Eq.~\eqref{eq:commute_anti_condition} is satisfied.
We first show that a Pauli dressing of a Clifford factor $C$ induced by a Pauli error $E_\mathrm{P}$ either commutes or anticommutes with the Clifford factor $C$ itself if the Clifford factor $C$ has involutive Pauli action on the Pauli $E_\mathrm{P}$, defined by
\begin{equation}
  C E_\mathrm{P} C^\dagger=\pm C^\dagger E_\mathrm{P} C,
  \label{eq:involutive_original}
\end{equation}
or equivalently, 
\begin{equation}
  C^2 E_\mathrm{P} (C^\dagger)^2=\pm E_\mathrm{P}.
  \label{eq:involutive}
\end{equation}
The following proposition proves this. 

\begin{proposition}
If a Clifford $C$ has involutive Pauli action on a Pauli $E_\mathrm{P}$ as defined in Eq.~\eqref{eq:involutive_original}, the Pauli dressing $D_C(E_\mathrm{P})$ defined in Eq.~\eqref{eq:dressing_Pauli_error} either commutes or anticommutes with $C$:
\begin{equation}
  C D_C(E_\mathrm{P})=\pm D_C(E_\mathrm{P}) C .
\end{equation}
\end{proposition}

\begin{proof}
Conjugating $D_C(E_\mathrm{P})$ by $C$ gives
\begin{equation}
  C D_C(E_\mathrm{P}) C^\dagger=(CE_\mathrm{P}C^\dagger)\bigl(C^2E_\mathrm{P}(C^\dagger)^2\bigr).
\end{equation}
Using Eq.~\eqref{eq:involutive}, 
\begin{equation}
  C D_C(E_\mathrm{P}) C^\dagger=\pm (CE_\mathrm{P}C^\dagger) E_\mathrm{P}.
\end{equation}
Since $CE_\mathrm{P}C^\dagger$ is a Pauli operator up to a sign and $E_\mathrm{P}$ is also a Pauli operator, they either commute or anticommute. Therefore
\begin{equation}
  C D_C(E_\mathrm{P}) C^\dagger=\pm E_\mathrm{P}(CE_\mathrm{P}C^\dagger)=\pm D_C(E_\mathrm{P}), 
  \label{eq:dressing_preserve}
\end{equation}
or equivalently, 
\begin{equation}
  C D_C(E_\mathrm{P})=\pm D_C(E_\mathrm{P})C.
\end{equation}
\end{proof}

Thus, if a Clifford factor has involutive Pauli action on every Pauli, a Pauli dressing on the Clifford factor induced by any Pauli error commutes or anticommutes with the Clifford factor itself. 
As a special case, a Hermitian Clifford operator has involutive Pauli action on every Pauli. We remark that a Clifford $C$ having involutive Pauli action on every Pauli is equivalent to its square being a Pauli $P\in\cP_2$:
\begin{equation}
  C^2=P,
  \label{eq:psc}
\end{equation}
which was a key property required for the method in Ref.~\cite{fby6-xjbm}.
So far, we have considered Pauli dressings of Clifford factors arising from Pauli errors. We must also account for the complementary case in which Clifford errors induce Pauli dressings of Pauli factors. Notably, all dressings arising in this latter case are already contained in the set obtained in the former, up to signs due to Eqs.~\eqref{eq:dressing_Pauli_error} and~\eqref{eq:dressing_Clifford_error}. 
Thus, it is enough to consider Pauli dressings by Pauli errors to ensure  Eq.~\eqref{eq:commute_anti_condition}.

Even if the Pauli dressing of a Clifford factor either commutes or anticommutes with the Clifford factor itself, this is not enough for the condition in Eq.~\eqref{eq:commute_anti_condition}. 
The Pauli dressing of a Clifford factor must also either commute or anticommute with any other Clifford factor whose support overlaps with that of the Pauli dressing.
Thus we require that, for each qubit, all Pauli dressings in $\mathcal{D}$ that act nontrivially on that qubit have the same Pauli axis.
Equivalently, all Pauli dressings in $\mathcal{D}$ commute with each other:
\begin{equation}
  [D,D']=I, \qquad \forall D,D'\in\mathcal{D}.
  \label{eq:dressings_commute_each}
\end{equation} 
Indeed, Eq.~\eqref{eq:dressings_commute_each} also implies Eq.~\eqref{eq:commute_anti_condition} as proven in Appendix~\ref{app:proofs}, so we do not need to impose Eq.~\eqref{eq:psc} as a separate condition to ensure Eq.~\eqref{eq:commute_anti_condition}.

Regarding Clifford errors acting on Clifford factors, so far we have assumed that two Clifford factors either commute or anticommute with each other in Eq.~\eqref{eq:Clifford-Clifford-commutation}. However, even without this assumption, the dressings created on Clifford factors by Clifford errors are indeed Pauli dressings that commute with all possible Pauli dressings by Pauli errors if Eq.~\eqref{eq:dressings_commute_each} holds. 
Specifically, the dressing of a Clifford factor $C'$ induced by a Clifford error $C$ is $[C,C']$ because
\begin{equation}
    CC'C^\dagger=[C,C']C'.
\end{equation}
If Eq.~\eqref{eq:dressings_commute_each} holds, it follows that
\begin{equation}
  [[C,C'],D]=I, \qquad \forall D\in\mathcal{D}.
  \label{clifford-clifford-dress}
\end{equation} 
The proof is provided in Appendix~\ref{app:proofs}. Thus, if Eq.~\eqref{eq:dressings_commute_each} is satisfied, Eq.~\eqref{eq:commute_anti_condition} holds even without the assumption in Eq.~\eqref{eq:Clifford-Clifford-commutation}.

\subsubsection{Nontrivial group commutators}

We have identified conditions under which Pauli dressings either commute or anticommute with Clifford factors. 
For the remainder of this section, we assume that Eq.~\eqref{eq:dressings_commute_each} holds, and thus Eqs.~\eqref{eq:commute_anti_condition} and~\eqref{eq:commute_anti_condition_stab} also hold.
To guarantee that a protocol measurement either commutes or anticommutes with every current generator, we also need the Pauli operators appearing in group commutators to remain stabilizers of the current state, up to signs.
Let $\mathcal{B}_{\mathrm{gc}}$ denote the Pauli subgroup generated by the Pauli operators appearing in group commutators between pairs of original stabilizer generators.
Below, we show that every element of $\mathcal{B}_{\mathrm{gc}}$ is guaranteed to remain a stabilizer up to a sign after any sequence of Pauli-or-hook errors and protocol measurements.

\begin{lemma}
\label{lem:B_remains_stabilizing}
After any sequence of Pauli-or-hook errors and protocol measurements,
for every $B\in\mathcal{B}_{\mathrm{gc}}$, either $B$ or $-B$ remains
an element of the current Clifford-stabilizer group.
\end{lemma}

\begin{proof}
We first show that every Pauli-or-hook error preserves every
$B\in\mathcal{B}_{\mathrm{gc}}$ up to a sign. For a Pauli error
$E_{\mathrm P}$,
\begin{equation}
  E_\mathrm{P} B E_\mathrm{P}^\dagger= \pm B,
  \quad \forall B\in\mathcal{B}_{\mathrm{gc}},
\end{equation}
because $B$ is a Pauli operator.
For a Clifford error $E_\mathrm{C}$,
\begin{equation}
  \label{eq:Clifford_error_vsB}
  E_\mathrm{C} B E_\mathrm{C}^\dagger= \pm B,
  \quad \forall B\in\mathcal{B}_{\mathrm{gc}}
\end{equation}
also holds due to the following reason.
Every Pauli operator appearing in a group commutator of original generators is a product of Pauli dressings in $\mathcal{D}$:
\begin{equation}
\mathcal B_{\rm gc}\subseteq \langle\mathcal D\rangle,
\label{eq:b_dressing_group}
\end{equation}
because both originate from group commutators between the same pairs of factors.
Since each Clifford error is equal to a Clifford factor having the same support, together with Eqs.~\eqref{eq:commute_anti_condition} and ~\eqref{eq:b_dressing_group}, Eq.~\eqref{eq:Clifford_error_vsB} is proven.

Finally, every protocol measurement always commutes with every element of $\mathcal{B}_{\mathrm{gc}}$.
Indeed, both are original stabilizers, and they commute as operators.
Therefore, every $B\in\mathcal{B}_{\mathrm{gc}}$ remains a current stabilizer up to a sign after any sequence of Pauli-or-hook errors and protocol measurements.
Note that, in general, an effectively measured operator is also Pauli-dressed as detailed in Sec.~\ref{subsubsec:effective}. Even in such a case, the effectively measured operator commutes with every element of $\mathcal{B}_{\mathrm{gc}}$ since every possible Pauli dressing commutes with every element of $\mathcal{B}_{\mathrm{gc}}$ due to Eqs.~\eqref{eq:dressings_commute_each} and~\eqref{eq:b_dressing_group}:
\begin{equation}
  DB=BD, \qquad \forall  D\in\mathcal{D},\quad \forall B\in\mathcal{B}_{\mathrm{gc}}.
  \label{eq:dressing_commutes_with_B}
\end{equation}
\end{proof}

\subsection{Update rules for the Clifford-stabilizer group}
We now show that the Clifford-stabilizer group can be updated in the same way as Pauli-stabilizer group cases, provided that every Pauli dressing in $\mathcal{D}$ has weight at most one and Eq.~\eqref{eq:dressings_commute_each} holds. 
First, we focus on the simpler case where circuit-level Pauli errors do not cause erroneous operators to be measured (up to measurement readout errors). 
We then describe the extension of our method to the general case where circuit-level Pauli errors can cause erroneous operators to be measured. 

\subsubsection{The case without erroneous measured operators}
First, we show that a protocol measurement and a Pauli-dressed original generator either effectively commute or effectively anticommute with each other on the current state.

\begin{lemma}
\label{lem:dressed_generator_commutator}
Let $M$ be a protocol measurement and let $\tilde{G}=DG$ be a Pauli-dressed version of an original generator $G$, where $D$ is an accumulated Pauli dressing.
Then, it is guaranteed that
\begin{equation}
  [M,\tilde{G}]\in\pm\mathcal{B}_{\mathrm{gc}}.
\end{equation}
Consequently, by Lemma~\ref{lem:B_remains_stabilizing}, $M$ and $\tilde{G}$ either effectively commute or effectively anticommute with each other on the current state $|\psi\rangle$:
\begin{equation}
  [M,\tilde{G}]|\psi\rangle=\pm|\psi\rangle.
\end{equation}
\end{lemma}

\begin{proof}
The group commutator is expressed as
\begin{align}
  [M,\tilde{G}]
  &=
  [M,DG]\\
  &=
  [M,D]\,D[M,G]D^\dagger .
\end{align}
From Eq.~\eqref{eq:commute_anti_condition_stab},
\begin{equation}
  [M,D]=\pm I.
\end{equation}
Also, since $[M,G]\in \mathcal{B}_{\mathrm{gc}}$, using Eq.~\eqref{eq:dressing_commutes_with_B},
\begin{equation}
  D[M,G]D^\dagger=[M,G].
\end{equation}
Thus
\begin{equation}
  [M,\tilde{G}]\in\pm\mathcal{B}_{\mathrm{gc}}.
\end{equation}
By Lemma~\ref{lem:B_remains_stabilizing}, every element of $\mathcal{B}_{\mathrm{gc}}$ is a current stabilizer up to a sign. 
Hence
\begin{equation}
  [M,\tilde{G}]|\psi\rangle=\pm|\psi\rangle.
\end{equation}
Therefore, $M$ and $\tilde{G}$ either effectively commute or effectively anticommute with each other on the current state.
\end{proof}

When a protocol measurement commutes with every current generator, the Clifford-stabilizer group remains unchanged. 
Now, we show that when a protocol measurement anticommutes with a current generator, we can update the Clifford-stabilizer group in an analogous way to the ordinary Pauli-stabilizer group cases.
\begin{proposition}
\label{prop:clifford_stabilizer_update}
If a protocol measurement $M$ effectively anticommutes with at least one current generator, the outcome is uniformly random. The group is updated by choosing a pivot generator $G'_p$ that effectively anticommutes with $M$, and for every other generator $G'_j$ that effectively anticommutes with $M$, replace it by
\begin{equation}
G'_j\leftarrow G'_jG'_p.
\label{eq:clifford_group_update_1}
\end{equation}
Then, replace the pivot $G'_p$ by the measured operator $M$ with the random sign:
\begin{equation}
G'_p\leftarrow (-1)^b M,
\qquad b\in\{0,1\}.
\label{eq:clifford_group_update_2}
\end{equation}
\end{proposition}

\begin{proof}
Consider the product $\tilde{G}_1\cdots \tilde{G}_s$ of some Pauli-dressed original generators $\tilde{G}_k$.
The group commutator between $M$ and the product is given by
\begin{equation}
  [M,\tilde{G}_1\cdots \tilde{G}_s]
  =
  \prod_{k=1}^{s}
  (\tilde{G}_1 \cdots \tilde{G}_{k-1})
  [M,\tilde{G}_{k}]
  (\tilde{G}_1 \cdots \tilde{G}_{k-1})^\dagger.
  \label{eq:general_product_comm}
\end{equation}
By Lemma~\ref{lem:dressed_generator_commutator}, $[M,\tilde{G}_{k}]\in\pm\mathcal{B}_{\mathrm{gc}}$, and each of $\tilde{G}_1, \ldots ,\tilde{G}_{k-1}$ commutes with every element of $\mathcal{B}_{\mathrm{gc}}$ since original generators commute with every element of $\mathcal{B}_{\mathrm{gc}}$, and Pauli dressings commute with every element of $\mathcal{B}_{\mathrm{gc}}$ by Eq.~\eqref{eq:dressing_commutes_with_B}. Hence
\begin{align}
  [M,\tilde{G}_1\cdots \tilde{G}_s]
  &=
  \prod_{k=1}^{s}
  (\tilde{G}_1 \cdots \tilde{G}_{k-1})
  [M,\tilde{G}_{k}]
  (\tilde{G}_1 \cdots \tilde{G}_{k-1})^\dagger \\
  &=
  \prod_{k=1}^{s}
  [M,\tilde{G}_{k}] \in \pm \mathcal{B}_{\mathrm{gc}}.
  \label{eq:general_product_comm_prod}
\end{align}

Since the elements of $\mathcal{B}_{\mathrm{gc}}$ commute with the protocol measurement $M$, the action of each $[M,\tilde{G}_{k}]$ on the state does not change due to the measurement of $M$. 
Thus, by Eq.~\eqref{eq:general_product_comm_prod}, the effective commutation sign between $M$ and a product of pre-measurement generators on the post-measurement state is equal to the product of the effective commutation signs between $M$ and the individual pre-measurement generators on the pre-measurement state. 
Consequently, if two generators effectively anticommute with $M$ on the pre-measurement state, their product effectively commutes with $M$ on the post-measurement state. 
This ensures that the Clifford-stabilizer group can be updated in the way described in Eqs.~\eqref{eq:clifford_group_update_1} and ~\eqref{eq:clifford_group_update_2}.
\end{proof}

We now prove that a protocol measurement effectively commutes or effectively anticommutes with current generators on the current state after any sequence of Pauli-or-hook errors and protocol measurements.

\begin{proposition}
\label{prop:commute_anticommute}
After any sequence of Pauli-or-hook errors and protocol measurements, for any protocol measurement $M$ and every current generator $G'$,
\begin{equation}
  [M,G']\in\pm\mathcal{B}_{\mathrm{gc}}.
  \label{eq:M-S-B}
\end{equation}
Consequently, by Lemma~\ref{lem:B_remains_stabilizing}, $M$ and $G'$ effectively commute or effectively anticommute with each other on the current state $|\psi\rangle$:
\begin{equation}
  [M,G']|\psi\rangle=\pm|\psi\rangle.
  \label{eq:M-S-state}
\end{equation}
\end{proposition}

\begin{proof}
By Lemma~\ref{lem:dressed_generator_commutator}, Eqs.~\eqref{eq:M-S-B} and~\eqref{eq:M-S-state} hold when $G'$ is a Pauli-dressed original generator. After one protocol measurement, any tracked generator of the current Clifford-stabilizer group can be written as the product of at most two Pauli-dressed original generators using Proposition~\ref{prop:clifford_stabilizer_update}.
Using Eq.~\eqref{eq:general_product_comm_prod} and Lemma~\ref{lem:B_remains_stabilizing}, Eqs.~\eqref{eq:M-S-B} and~\eqref{eq:M-S-state} also hold for the next protocol measurement. By repeatedly applying this argument, Eqs.~\eqref{eq:M-S-B} and~\eqref{eq:M-S-state} hold after any sequence of Pauli-or-hook errors and protocol measurements.
Note that, although it suffices for our purposes to show that Eq.~\eqref{eq:M-S-state} holds for the current state $|\psi\rangle$, more generally, the same relation holds for any state in the subspace stabilized, up to a sign, by the corresponding element of $\mathcal{B}_{\mathrm{gc}}$.
\end{proof}

To keep track of the state by updating the Clifford-stabilizer groups, we must ensure that the current Clifford-stabilizer group still uniquely defines the current state after any sequence of Pauli-or-hook errors and protocol measurements if the original group does so. In the following, we show that this holds in our update rules.

Suppose the current generators $G'_j, \forall j$ uniquely stabilize the state $|\psi\rangle$.
First, after a unitary error $E$, we replace each generator $G'_j$ by $EG'_jE^\dagger$. The updated generators stabilize the post-error state. Now suppose $|\phi\rangle$ is any state stabilized by all updated generators $EG'_jE^\dagger$, then $E^\dagger|\phi\rangle$ is stabilized by all $G'_j$. Thus, $E^\dagger|\phi\rangle \propto |\psi\rangle$, or equivalently, $|\phi\rangle \propto E|\psi\rangle$.
Therefore, any state stabilized by the updated generators is equal to the post-error state $E|\psi\rangle$ (up to a global phase), showing that the updated generators uniquely specify the post-error state.
A deterministic measurement does not change the state, so the post-measurement state is still unique.
For a random measurement, first multiply current generators so that exactly one current generator effectively anticommutes with the measured Hermitian operator $M$. Removing this anticommuting generator leaves a two-dimensional space, and the two eigenspaces of $M$ inside it are one-dimensional. Thus, after replacing the removed generator by $\pm M$ with the random obtained sign, the subspace is one-dimensional, which is unique.
Note that we here assume that the generators satisfy the property in Definition~\ref{def:primary_support_exactness} so that every nonidentity element of the group is traceless.

\subsubsection{The case with erroneous measured operators}
\label{subsubsec:effective}
So far we have considered the cases where errors are applied to the state while the measured operator remains a noiseless original generator.
This is true (up to measurement readout errors) for non-Clifford circuits performing protocol measurements under circuit-level Pauli errors if each original generator is (fold-)transversal and each factor is Hermitian as assumed in Sec.~\ref{subsec:rules_diagonal}.
The reason is that, in such circuits, errors can be propagated and represented as data Pauli or Clifford errors between protocol measurements, up to measurement readout errors.
However, this is not the case for general non-Clifford circuits performing protocol measurements.

For example, suppose we measure a non-transversal original generator by a sequence of controlled-$U_j$ gates. If an error occurs on the data qubits between the two controlled gates whose target qubits share the support, propagating the error to before or after the measurement circuit may need to Pauli-dress the measured data operator. The circuit therefore can be equivalent to measuring an incorrect Pauli-dressed data operator. 

The proofs so far are written for the cases in which the measured operator is a noiseless original generator.
Nevertheless, even when the effectively measured operator is Pauli-dressed, the proofs are almost directly applicable as shown below, assuming the effectively measured operator is still Hermitian.
Let $M$ and $G$ denote original generators, and let the Pauli-dressed versions of them be
\begin{equation}
  \tilde M = D_M M, \quad \tilde G = D_G G,
  \label{eq:dressed_version}
\end{equation}
where $D_M$ and $D_G$ are Pauli dressings.
The group commutator of them is expressed as
\begin{align}
  [\tilde M,\tilde G]
  &=
  [\tilde M,D_GG] \\
  &=
  [\tilde M,D_G]D_G[\tilde M,G]D_G^\dagger\\
  &=
  [M,D_G]D_G[\tilde M,G]D_G^\dagger.
  \label{eq:meas-frame-dressed-gen-comm-1}
\end{align}
Here, we used Eq.~\eqref{eq:dressings_commute_each}, i.e., the Pauli dressings $D_M$ and $D_G$ commute with each other.
The $[\tilde M,G]$ is expressed as
\begin{align}
  [\tilde M,G]
  &=
  [D_MM,G] \\
  &=
  D_M[M,G]D_M^\dagger[D_M,G]\\
  &=
  [M,G][D_M,G].
  \label{eq:meas-frame-dressed-gen-comm-2}
\end{align}
Here, we used Eq.~\eqref{eq:dressing_commutes_with_B}.
Since $[D_M,G]$ is $\pm I$, $[\tilde M,\tilde G]$ is then simplified to 
\begin{equation}
  [\tilde M,\tilde G]=[M,D_G][M,G][D_M,G] \in \pm \mathcal{B}_{\mathrm{gc}}.
  \label{eq:dressed-meas-dressed-gen-in-Bgc}
\end{equation}
For a current generator $\tilde G_1\cdots \tilde G_s$, by the same derivation as Eq.~\eqref{eq:general_product_comm_prod}, the group commutator is
\begin{align}
  [\tilde M,
    \tilde G_1\cdots \tilde G_s]
  &=
  \prod_{k=1}^{s}
  (\tilde G_1 \cdots \tilde G_{k-1})
  [\tilde M,\tilde G_k]
  (\tilde G_1 \cdots \tilde G_{k-1})^\dagger\\
  &=
  \prod_{k=1}^{s}
  [\tilde M,\tilde G_k]\in \pm\mathcal B_{\rm gc}.
  \label{eq:measured_conju_group}
\end{align}
Since Lemma~\ref{lem:B_remains_stabilizing} also holds even when a measured operator is Pauli-dressed due to Eq.~\eqref{eq:dressing_commutes_with_B}, we have shown that the Pauli-dressed measured operator either effectively commutes or effectively anticommutes with every current generator. Thus, the Clifford-stabilizer group can be updated in a manner analogous to the update of an ordinary Pauli-stabilizer group.

Another situation in which the effectively measured operator can become an incorrect Pauli-dressed operator is when one of the factors is not Hermitian. 
As shown in Fig.~\ref{fig:error_prop_ancilla}, when Pauli errors occur on the ancilla qubits, the controlled-$(C^\dagger)^2$ operation arises, changing the effectively measured operator. 
In the following, we show that $(C^\dagger)^2$ also either commutes or anticommutes with every Clifford factor if every other possible Pauli dressing does.
In our setting, $(C^\dagger)^2$ is a Pauli up to a phase by Eq.~\eqref{eq:psc}. 
Since $C$ either commutes or anticommutes with every possible Pauli dressing, for $D\in\mathcal{D}$, 
\begin{equation}
  C D C^\dagger=\pm D. 
\end{equation}
Applying $C$ from the left and $C^\dagger$ from the right gives 
\begin{equation}
  C^2 D (C^\dagger)^2=D. 
\end{equation}
Equivalently,
\begin{equation}
  D (C^\dagger)^2=(C^\dagger)^2 D, 
\end{equation}
which shows that $(C^\dagger)^2$ commutes with every possible Pauli dressing.
Thus, due to Eq.~\eqref{eq:dressings_commute_each}, it either commutes or anticommutes with every Clifford factor.
Therefore, the effectively measured Pauli-dressed operator due to the controlled-$(C^\dagger)^2$ operation either effectively commutes or effectively anticommutes with every current generator, assuming the effectively measured operator is Hermitian.

We now discuss conditions under which the effectively measured operator is Hermitian.
A Hermitian operator $M$ conjugated by an error $E$, i.e., $E^\dagger ME$, is always Hermitian. However, when only a subset of factors composing $M$ is Pauli-dressed due to errors, Hermiticity is not guaranteed.
Suppose a noiseless original generator is written as
\begin{equation}
  M = \prod_j U_j.
\end{equation}
If errors occur inside the measurement circuit of $M$ performing a sequence of controlled-$U_j$ operations, the data operator controlled by the ancilla qubits has the Pauli-dressed form
\begin{equation}
  \tilde{M} = \prod_j D_j^{t_j}U_j, \qquad t_j \in \{0,1 \},
  \label{eq:each_factor_dressing}
\end{equation}
where $D_j$ is a Pauli dressing of $U_j$.
In order to guarantee that $\tilde{M}$ is Hermitian, we require that $\tilde{M}^2=I$. Let $s_{i,j} \in \{ \pm1\}$ denote the commutation sign between $D_i$ and $U_j$:
\begin{equation}
  D_i U_j=s_{i,j} U_j D_i.
\end{equation}
We also define 
\begin{equation}
  D_i^2=\lambda_i I, 
\end{equation}
where $\lambda_i \in \{\pm1\}$ because $D_i$ is a Pauli multiplied by a phase in $\{ \pm 1, \pm i\}$.
Then, $\tilde{M}^2$ is given by
\begin{equation}
  \tilde{M}^2 = \left (\prod_{i:t_i=1} \lambda_i \prod_j s_{i,j} \right)I.
\end{equation}
Now, $\tilde{M}^2=I$ must hold for every combination of $t_i \in \{0,1\}$ and for every possible choice of Pauli dressing $D_i \in \mathcal{D}$ for each $i$ with $t_i=1$.
A sufficient condition for this is
\begin{equation}
  \lambda_i \prod_j s_{i,j}=1
  \label{eq:effective_herm}
\end{equation}
for every $i$ and every possible choice of $D_i \in \mathcal{D}$.
Equivalently, for any $D_i \in \mathcal{D}$ satisfying $D_i^2=I$, $D_i$ must anticommute with an even number of the factors $U_j$, i.e., $D_i$ must commute with $M$.
For any $D_i \in \mathcal{D}$ satisfying $D_i^2=-I$, $D_i$ must anticommute with an odd number of the factors $U_j$, i.e., $D_i$ must anticommute with $M$.
We remark that in preparation protocols for typical logical magic states including the surface code $|\bar{T}\rangle$ state, Steane code $|\bar{H}\rangle$ state, and $\llbracket 8,3,2 \rrbracket$ code $|\overline{CCZ}\rangle$ state where each factor is chosen to be Hermitian and the generators are transversal, the Hermiticity issue for effectively measured operators does not arise.

\subsection{Deterministic measurements}
\label{subsec:deterministic}
For an $n$-qubit Pauli-stabilizer group generated by $n$ independent commuting Pauli operators, if a Pauli measurement is deterministic, that is, the measured Pauli commutes with every tracked generator, it is guaranteed that the measured Pauli is equal, up to a sign, to a product of the tracked generators.
In contrast, for a general Clifford-stabilizer group that specifies a Clifford-stabilizer state, this implication is false. Even if the tracked $n$ generators uniquely define the state, a Pauli or Clifford measurement may commute with every tracked generator and thus have a deterministic outcome while not being equal, up to a sign, to a product of the tracked generators.
A simple example is a two-qubit state uniquely specified by the Clifford-stabilizer group
\begin{equation}
\langle CZ_{1,2},\; Z_1Z_2\rangle .
\end{equation}
The common $+1$ eigenstate of the two operators is determined to be $|00\rangle$. 
Suppose an $X_1$ error occurs. The generators are transformed to 
\begin{equation}
\langle Z_2CZ_{1,2},\; -Z_1Z_2\rangle .
\end{equation}
If $CZ_{1,2}$ is measured on this state, the outcome is deterministically $+1$, and it commutes with both tracked generators, while
\begin{equation}
CZ_{1,2} \notin \pm\langle Z_2CZ_{1,2},\; -Z_1Z_2\rangle .
\end{equation}
Thus, one-dimensionality of the stabilizer subspace is generally not sufficient to ensure that a measured operator is equal, up to a sign, to a product of generators of the Clifford-stabilizer group when the measurement is deterministic. 
Nevertheless, despite not holding in general, the implication is true for the Clifford-stabilizer groups of typical logical magic states under circuit-level Pauli errors and protocol measurements.
In this subsection, we establish a condition under which, for a deterministic measurement, the measured operator is equal to a product of generators of the current Clifford-stabilizer group, up to a sign.

We let $\mathcal{B}_\mathrm{init}$ denote a Pauli-stabilizer subgroup of an original Clifford-stabilizer group whose elements remain stabilizers, up to a sign, throughout the protocol. 
This means every possible Pauli dressing commutes with elements of $\mathcal{B}_\mathrm{init}$. The subgroup $\mathcal B_\mathrm{gc}$ is equal to or contained in $\mathcal{B}_\mathrm{init}$:
\begin{equation}
  \mathcal B_\mathrm{gc} \subseteq \mathcal{B}_\mathrm{init},
\end{equation}
because the subgroup $\mathcal{B}_\mathrm{init}$ may also contain Pauli stabilizers which commute with every possible Pauli dressing, even if they do not appear in group commutators. In the surface code $|\bar{T}\rangle$ state and Steane code $|\bar{H}\rangle$ state examples, $\mathcal B_\mathrm{gc} = \mathcal{B}_\mathrm{init}$, while the $\llbracket 8,3,2 \rrbracket$ $|\overline{CCZ}\rangle$ state example satisfies $\mathcal B_\mathrm{gc} \subset \mathcal{B}_\mathrm{init}$.

The \emph{Primary generators} are a maximal subset of the original generators such that no product of distinct members belongs to $\mathcal{B}_{\rm init}$.
We choose the $n$ generators of an $n$-qubit original Clifford-stabilizer group to consist of primary generators and independent elements of $\mathcal{B}_\mathrm{init}$.
The \emph{primary part} of an operator is the collection of operators appearing in the primary generators, excluding the Pauli operators of dressing axes.
A product is called \textit{zero-primary} if, after multiplying generators, their primary parts cancel and only non-primary Pauli operators remain.
The current Pauli stabilizer subgroup generated by all zero-primary products is denoted $\mathcal{B}$. Initially, one has $\mathcal{B}=\mathcal{B}_\mathrm{init}$, but during the protocol, $\mathcal{B}$ may contain additional Pauli stabilizers obtained from products of current generators. 

For each qubit $q$, we define a \emph{local dressing-axis Pauli}
\begin{equation}
  Q_q\in\{X_q,Y_q,Z_q\}
  \label{eq:local_dressing_axis_pauli}
\end{equation}
as the Pauli axis of the nontrivial Pauli dressings in $\mathcal{D}$ acting
on qubit $q$.
The $Q_q$ may depend on $q$.
In the diagonal surface code $|\bar{T}\rangle$ state and $\llbracket 8,3,2 \rrbracket$ $|\overline{CCZ}\rangle$ state examples, $Q_q=Z_q$ on all qubits. In the Steane code $|\bar{H}\rangle$ state example, $Q_q=Y_q$ on all qubits.
For $c\in\mathbb{F}_2^n$, we define a \emph{dressing-axis Pauli}
\begin{equation}
  Q(c):=\prod_{q:c_q=1} Q_q .
  \label{eq:dressing_axis_pauli}
\end{equation}

For each current generator $G'_j$, we define its \emph{dressing-anticommutation support} 
$v_j\in\mathbb{F}^n_2$ by
\begin{equation}
  Q(c)G'_j=(-1)^{c\cdot v_j}G'_jQ(c)
  \label{eq:dressing_primary_anti}
\end{equation}
for all $c\in\mathbb{F}_2^n$. Equivalently, $(v_j)_q=1$ exactly when $G'_j$ anticommutes with $Q_q$. The linear dependencies among the vectors $v_j$ correspond to the elements of the current zero-primary subgroup $\mathcal{B}$. For example, in the Steane code $|\bar{H}\rangle$ state, $G_i^X$ and $G_i^Z$ have the same $v$, and their product $G_i^X G_i^Z$ is in $\mathcal{B}$.
We assume the following property so that no product of current generators is equal to non-Pauli Clifford operators that commute with every Pauli dressing.

\begin{definition}[Primary-support exactness]
\label{def:primary_support_exactness}
Let $G'_1,\ldots,G'_n$ be the generators of a Clifford-stabilizer group, with corresponding dressing-anticommutation supports
$v_1,\ldots,v_n\in\mathbb{F}_2^n$ determined by the local dressing-axis Paulis in Eq.~\eqref{eq:local_dressing_axis_pauli}. We say that $G'_1,\ldots,G'_n$ satisfy primary-support exactness if, for every binary vector $m=(m_1,\ldots,m_n)$,
\begin{equation}
  \sum_{j=1}^n m_jv_j=0
\end{equation}
if and only if the product
\begin{equation}
  \prod_{j=1}^n G_j'^{m_j}
\end{equation}
does not have a primary part. Equivalently, the product is a Pauli operator in $\mathcal{B}$.
\end{definition}

Note that if Definition~\ref{def:primary_support_exactness} is satisfied for the original generators, it is also satisfied for the current generators because each current generator is the product of original generators with Pauli dressings.

\begin{proposition}
\label{prop:completeness}
Assume primary-support exactness, see Definition~\ref{def:primary_support_exactness}. The current subgroup $\mathcal{B}$ is exactly the subgroup generated by dressing-axis Paulis $Q(c)$ up to signs that commute with every current generator. Consequently, any dressing-axis Pauli with a deterministic eigenvalue on the current state is, up to a sign, a product of the current generators.
\end{proposition}

\begin{proof}
Let $k$ be the number of independent vectors among $v_1,\ldots,v_n$. Relabel the generators, if necessary, so that $v_1,\ldots,v_k$ form a basis for their span.
For each $j>k$, $v_j$ is an XOR of some subset of $v_1,\ldots,v_k$. Multiplying $G'_j$ by the corresponding product of generators gives a product whose dressing-anticommutation support is zero. 
By primary-support exactness, this product lies in $\mathcal{B}$.
Thus the linear dependencies among the $v_j$'s give $n-k$ independent zero-primary products. 

Now consider a dressing-axis Pauli $Q(c)$ that commutes with every current generator. By Eq.~\eqref{eq:dressing_primary_anti}, this is equivalent to
\begin{equation}
  c\cdot v_j=0
\end{equation}
for every $j \in \{1,2,\ldots,k\}$. These $k$ independent binary equations leave $n-k$ independent binary degrees of freedom for $c$. 

The $n-k$ independent generators of $\mathcal{B}$ are, up to signs, dressing-axis Paulis contained in the subgroup of all dressing-axis Paulis that commute with every current generators. Since the latter subgroup also has $n-k$ independent binary degrees of freedom, the two subgroups are equal up to signs. As a consequence, if a dressing-axis Pauli has a deterministic eigenvalue on the current state, it commutes with every current stabilizer and lies in $\mathcal{B}$ up to a sign. This means that the Pauli is, up to a sign, a product of the current generators.
\end{proof}

\begin{proposition}
\label{prop:deterministic}
Assume primary-support exactness, see Definition~\ref{def:primary_support_exactness}. Let $M$ be a Hermitian protocol measurement. If the measurement of $M$ is deterministic, then a product of current generators is equal to $(-1)^m M$, where $(-1)^m$ is the deterministic measurement outcome.
\end{proposition}

\begin{proof}
We first prove that when the measurement of $M$ is deterministic, a product of current generators has the same primary part as $M$. Assume the contrary. By primary-support exactness, this means that $v_M$, which is the dressing-anticommutation support of $M$, is independent of the current supports $v_1,\ldots,v_n$. 
Then there exists a vector $c\in\mathbb{F}_2^n$ such that
\begin{equation}
  c\cdot v_j=0 \quad \forall j \in \{1,2,\ldots,n\},
  \qquad
  c\cdot v_M=1,
\end{equation}
meaning that $Q(c)$ commutes with every current generator but anticommutes with $M$. By Proposition~\ref{prop:completeness}, that $Q(c)$ is, up to a sign, a product of current generators.
Since $M$ is deterministic, the current state should be an eigenstate of $M$, but this is impossible because $Q(c)$ and $M$ anticommute. Therefore, a product of current generators has the same primary part as $M$.

Choose such a product and call it $A$. Then the quotient $AM^{-1}$ does not have a primary part and has the form
\begin{equation}
  AM^{-1}=i^u Q(c)
  \label{eq:membership_residual}
\end{equation}
for some dressing-axis Pauli $Q(c)$. The product $A$ has a deterministic real eigenvalue on the current state because it is generated by current generators. The measured operator $M$ also has a deterministic real eigenvalue because $M$ is Hermitian. Therefore $i^uQ(c)$ must have a real deterministic eigenvalue. Since $Q(c)$ has a real eigenvalue, the scalar $i^u$ must be real:
\begin{equation}
  u=0\quad\text{or}\quad u=2\pmod4.
\end{equation}
Since $Q(c)$ has a deterministic eigenvalue, by Proposition~\ref{prop:completeness}, $Q(c)$ is, up to a sign, a product of current generators. Thus, by Eq.~\eqref{eq:membership_residual}, the measured operator $M$ is generated by the current generators up to a real sign. This sign is the deterministic measurement outcome.
Note that the proof also holds when the measured operator is a Hermitian Pauli-dressed original generator because the Pauli dressings do not affect the proof.
\end{proof}

Putting all sufficient conditions together, we define an \emph{admissible set of Clifford-stabilizer generators}.

\begin{definition}[Admissible set of Clifford-stabilizer generators]
\label{def:admissible}
Consider an $n$-qubit Clifford-stabilizer group with a fixed set of $n$ Hermitian Pauli or Clifford generators and a fixed factorization of each generator into Pauli and Clifford factors (Eq.~\eqref{eq:factor}). We call this generating set, together with these factorizations, an admissible set of Clifford-stabilizer generators if the following conditions are satisfied:
\begin{enumerate}[leftmargin=*]
\item The common $+1$ eigenspace of the generators is one-dimensional.
\item For every Clifford factor $C$ and every single-qubit Pauli error
$E_{\mathrm P}$, the Pauli dressing $D_C(E_{\mathrm P})$ (Eq.~\eqref{eq:dressing_Pauli_error}) has weight
at most one, and all such Pauli dressings mutually commute (Eq.~\eqref{eq:dressings_commute_each}).
\item The generators satisfy primary-support exactness (Definition~\ref{def:primary_support_exactness}).
\item For each generator, every possible effectively measured operator is Hermitian (Eq.~\eqref{eq:effective_herm}).
\end{enumerate}
\end{definition}

These conditions guarantee that the desired group update rules apply. We summarize this result in the following theorem.

\begin{theorem}[Clifford-stabilizer group update]
\label{theorem:update}
Suppose a state is uniquely stabilized by a Clifford-stabilizer group with an admissible set of Clifford-stabilizer generators, and each protocol measurement is implemented according to the chosen factorization of the generator. Then, when performing any sequence of protocol measurements under circuit-level Pauli errors, the Clifford-stabilizer group can be updated in a manner analogous to an ordinary Pauli-stabilizer group (Propositions~\ref{prop:clifford_stabilizer_update} and~\ref{prop:commute_anticommute}).
Furthermore, whenever a protocol measurement has a deterministic outcome, the group remains unchanged, and the measured operator is, up to the deterministic sign, a product of the current generators (Proposition~\ref{prop:deterministic}).
\end{theorem}

Note that a protocol measurement $G$ with the factorization in Eq.~\eqref{eq:factor} is implemented as a sequence of controlled-$U_j$ operations, each realized by a controlled gate that may be conjugated by non-Clifford gates acting on the data qubits, as explained in Sec.~\ref{subsubsec:commutation}.
The reason it is sufficient to impose the condition only on Pauli dressings of Clifford factors induced by Pauli errors in the condition $2$ of Definition~\ref{def:admissible} is that the Pauli dressings by Clifford errors and those due to the controlled-$(C^\dagger)^2$ operation in Fig.~\ref{fig:error_prop_ancilla} are contained, up to phases, in the set generated by Pauli dressings of Clifford factors induced by Pauli errors, as detailed in Secs.~\ref{subsubsec:commutation} and~\ref{subsubsec:effective}, respectively.

\subsection{Logical error evaluation}
\label{subsec:logical_error}
In simulations of QEC protocols, we aim to compute whether an output state has a logical error. To this end, we assume that a fictitious (i.e., noiseless) round of syndrome extraction is performed at the end of the protocol to remove potentially corrected or detected errors. After applying a recovery operation, or discarding the state based on the noiseless syndrome, the state is guaranteed to be in the code space. 
This is a common procedure to evaluate whether a logical error has occurred in numerical simulations of QEC protocols.

The commuting or anticommuting structure of the Clifford-stabilizer group allows us to evaluate the occurrence of a logical error on the output state in the following way.
After the fictitious round of syndrome extraction, we perform the fictitious measurements of Clifford stabilizers that fix the logical degrees of freedom. This corresponds to a fictitious measurement of $\bar{H}_{XY}$ in the surface code $|\bar{T}\rangle$ state example.
Due to the commuting or anticommuting structure of the Clifford-stabilizer group, the possibilities for each fictitious measurement of a Clifford stabilizer are: 
\begin{enumerate}[label=(\roman*)]
  \item It commutes with every current generator and has a deterministic $+1$ outcome.
  \item It commutes with every current generator and has a deterministic $-1$ outcome.
  \item It anticommutes with a current generator.
\end{enumerate}
Case ($\mathrm{i}$) is the only case in which the output state is a desired eigenstate of the Clifford stabilizer. Thus, if Cases ($\mathrm{ii}$) or ($\mathrm{iii}$) occur, the state has a logical error. 
Therefore, in our algorithm, we can evaluate the occurrence of a logical error from the information of both the measurement outcomes of the fictitious Clifford-stabilizer measurements and the commutations of the measurements with current generators.

\section{Polynomial-time algorithm to update the Clifford-stabilizer group}
\label{sec:algorithm}
We have shown that the Clifford-stabilizer group can be updated in a manner analogous to the update of an ordinary Pauli-stabilizer group when performing a sequence of protocol measurements under circuit-level Pauli errors, provided the conditions in Definition~\ref{def:admissible} are satisfied.
However, the update rules are not enough to perform the simulation efficiently. We need to construct an efficient algorithm that implements the update rules.
One insight from the effective commutation expression given in Eqs.~\eqref{eq:dressed-meas-dressed-gen-in-Bgc} and~\eqref{eq:measured_conju_group} is that the effective commutation between a Pauli-dressed measured operator and a current generator is computed from the commutation signs between Pauli dressings and original generators, mediated by the eigenvalues of the elements of $\mathcal{B}_{\mathrm{gc}}$.
This already suggests the idea of an efficient algorithm provided we have a way to efficiently compute the commutations between Pauli dressings and original generators. However, this approach requires one to track which elements of $\mathcal{B}_{\mathrm{gc}}$ arise in the group commutator, and the corresponding eigenvalues of the elements of $\mathcal{B}_{\mathrm{gc}}$ must be multiplied when computing the commutation signs.
This complicates the algorithm.
Therefore, here we view measurements and errors in a \emph{measurement-conjugated picture} to resolve this issue as explained in the following.

\subsection{Measurement-conjugated picture}\label{subsec:error_conjugated_meas}

In the picture in which errors are applied to a state, stabilizer generators are conjugated.
There is an equivalent picture in which the state is kept ideal and errors are absorbed into the measured operators instead, which we call the measurement-conjugated picture.
Suppose a unitary error $E$ has occurred on a state $|\psi\rangle$, so that the post-error state is $E|\psi\rangle$.
Measuring a Hermitian operator $M$ on this state is equivalent to measuring the conjugated operator
\begin{equation}
  \tilde M := E^\dagger M E
  \label{eq:measurement-frame-operator}
\end{equation}
on $|\psi\rangle$, followed by the application of $E$. Indeed, if
\begin{equation}
  \Pi_m(M)=\frac{I+mM}{2}
  \label{eq:projector}
\end{equation}
is the projector associated with a measurement outcome $m \in \{+1,-1\}$, then
\begin{equation}
  \Pi_m(M)E|\psi\rangle
  =
  E\,\Pi_m (\tilde M )|\psi\rangle.
\end{equation}
In this measurement-conjugated picture, the post-measurement state before applying the final $E$, i.e., 
\begin{equation}
  |\psi^\mathrm{fr}_m \rangle=\frac{\Pi_m (\tilde M )|\psi\rangle}{\|\Pi_m (\tilde M )|\psi\rangle \|}
  \label{eq:frame_state}
\end{equation}
is called a \emph{frame state}.
The stabilizers of the frame state are called the \emph{frame stabilizers}.
When we next measure another operator $M'$, we effectively measure $E^\dagger M'E$ on the frame state in Eq.~\eqref{eq:frame_state}, followed by the application of $E$.

Now, we explain how the commutation is expressed in this measurement-conjugated picture when we are effectively measuring an error-conjugated operator $\tilde M=E^\dagger ME$ on the current frame state $|\psi^\mathrm{fr} \rangle$.
First, $\tilde M$ is a Pauli-dressed operator $\tilde M=D_M M$, where $D_M$ is a Pauli dressing, since each factor contained in $M$ is Pauli-dressed. In particular, a factor $U$ contained in $M$ is Pauli-dressed due to the error $E$:
\begin{equation}
  E^\dagger UE = \left(E^\dagger UE U^\dagger \right) U:=\bar{D}_U(E)U,
  \label{eq:dressed_factor_meas}
\end{equation}
where $\bar{D}_U(E)=E^\dagger UE U^\dagger$ is a Pauli dressing in the measurement-conjugated picture.
Suppose ${ G'{^\mathrm{fr}}=\tilde G_1\tilde G_2\cdots \tilde G_s}$ is a current frame stabilizer generator, where each $\tilde G_k=D_{G_k}G_k$ is a Pauli-dressed original generator.
By the same derivation as Eq.~\eqref{eq:measured_conju_group}, the commutation between $\tilde M$ and $G'{^\mathrm{fr}}$ is determined by
\begin{align}
  [\tilde M,
    G'{^\mathrm{fr}}]|\psi^\mathrm{fr} \rangle
&=
  \prod_{k=1}^{s}
  [M,D_{G_k}][M,G_k][D_M,G_k]|\psi^\mathrm{fr} \rangle.
  \label{eq:final_commutator_meas}
\end{align}
Remarkably, in the measurement-conjugated picture, each $[M,G_k]$ has a $+1$ eigenvalue on the frame state $|\psi^\mathrm{fr} \rangle$ since errors that flip the eigenvalues of elements of $\mathcal{B}_{\mathrm{gc}}$ are not applied to $|\psi^\mathrm{fr} \rangle$.
Therefore, Eq.~\eqref{eq:final_commutator_meas} is simplified to 
\begin{align}
  [\tilde M,
    G'{^\mathrm{fr}}]|\psi^\mathrm{fr} \rangle
&=
  \prod_{k=1}^{s}
  [M,D_{G_k}][D_M,G_k]|\psi^\mathrm{fr} \rangle.
  \label{eq:final_commutator_meas_simple}
\end{align}

The advantage of this measurement-conjugated picture in terms of the Clifford-stabilizer group update is thus the group commutator between an error-conjugated measured operator and a current frame stabilizer generator can be computed without multiplying the eigenvalues of the elements of $\mathcal{B}_\mathrm{gc}$.
This leads to a simple efficient algorithm to implement the Clifford-stabilizer group update as explained below.

\subsection{Pauli proxy}
We now present a way to efficiently compute the commutations between Pauli dressings and original generators.
In particular, we develop a compact representation of original generators to achieve the efficient commutation computation.
The key observation is that, although original generators may contain Clifford operators, each original generator either commutes or anticommutes with every Pauli dressing.
Since Pauli operators likewise either commute or anticommute with one another, this suggests that the Clifford factors in the original generators may be replaced by suitable Pauli operators while preserving both their commutation relations with the Pauli dressings and the mutual commutation of the original generators.
If such a valid replacement exists, that means that we can reproduce the Clifford-stabilizer group update by a Pauli-stabilizer group update, which can be efficiently executed using a Clifford simulator. 
In the following, we show that such valid replacements indeed always exist. 

We call the Pauli operator used in this replacement a \emph{Pauli
proxy}. 
This replacement is neither an equality of operators nor an approximation to the exact simulation. Rather, it is a reduction rule designed to reproduce the commutation of the Clifford factors with possible Pauli dressings, while preserving the mutual commutation of the original generators.
More precisely, let $U_j$ be a factor contained in an original generator.
The Pauli proxy $\hat{P}_{U_j}$ of the factor $U_j$ is chosen so that $\hat{P}_{U_j}$ has the same commutation sign with every possible Pauli dressing in $\mathcal{D}$ as $U_j$:
\begin{equation}
  [U_j, D] = \lambda I, \quad [\hat{P}_{U_j}, D] = \lambda I, \qquad \forall D \in  \mathcal{D},
\end{equation}
where $\lambda \in \{\pm 1 \}$.
If this holds for every $U_j$, the Pauli proxy $\hat{P}_{G_j}$ of an original generator $G_j$ also satisfies the following for every $G_j$.
\begin{equation}
  [G_j, D] = \lambda I, \quad [\hat{P}_{G_j}, D] = \lambda I, \qquad \forall D \in  \mathcal{D},
  \label{eq:repre_dressing_relation}
\end{equation}
where $\lambda \in \{\pm 1 \}$.
To preserve the commutation of the original generators, the Pauli proxies of two given original generators $G_i$ and $G_j$ must commute with each other:
\begin{equation}
  [\hat{P}_{G_i},\hat{P}_{G_j}]=I.
  \label{eq:repre_commute}
\end{equation}
Furthermore, all Pauli proxies of original generators must be independent Pauli operators.
Hence, the set of Pauli proxies of original generators that uniquely define a target logical magic state is valid if Eqs.~\eqref{eq:repre_dressing_relation} and~\eqref{eq:repre_commute} are satisfied for all original generators and additionally all Pauli proxies are independent.

In the following, we show that a valid set of Pauli proxies for original generators is guaranteed to exist. 
We first show that there exists a set of Pauli proxies that satisfy Eq.~\eqref{eq:repre_dressing_relation}.
For non-primary generators, namely, elements of $\mathcal{B}_{\mathrm{init}}$, Pauli proxies can be chosen to be the non-primary generators themselves since they are already Pauli operators.
For primary generators, for each qubit $q$ in the dressing-anticommutation support, i.e., $(v_j)_q=1$ in Eq.~\eqref{eq:dressing_primary_anti}, choose one Pauli
$R_q$ that anticommutes with $Q_q$ defined in Eq.~\eqref{eq:local_dressing_axis_pauli}. 
Then, the Pauli proxy of a primary generator $G_j$ chosen in this way
\begin{equation}
  \hat{P}_{G_j}=\prod_{q:(v_j)_q=1} R_q
  \label{eq:repre_choice_1}
\end{equation}
reproduces the commutation of $G_j$ with every possible Pauli dressing, satisfying Eq.~\eqref{eq:repre_dressing_relation}.
Note that, due to primary support exactness, see Definition~\ref{def:primary_support_exactness}, every primary generator has dressing-anticommutation support.
Hence, each Pauli proxy of a primary generator chosen in this way is a nontrivial Pauli operator.

We now show that a set of Pauli proxies that satisfies both Eqs.~\eqref{eq:repre_dressing_relation} and~\eqref{eq:repre_commute} and that are independent is guaranteed to exist. 
For primary generators chosen in Eq.~\eqref{eq:repre_choice_1}, there are two choices for each $R_q$ because there are two Pauli operators that anticommute with $Q_q$.
Among the two, if the same choice $R_q$ is used for each qubit $q$ for all primary generators, then the Pauli proxies of the primary generators commute with one another because they use the same local Pauli axis on every shared qubit.
They also commute with the non-primary generators due to the fact that Pauli dressings and elements of $\mathcal{B}_{\mathrm{init}}$ have the same Pauli axis on each qubit and thus the commutation with elements of $\mathcal{B}_{\mathrm{init}}$ is preserved by the Pauli proxies. 
They are independent since the dressing-anticommutation support vectors are linearly independent due to primary support exactness from  Definition~\ref{def:primary_support_exactness}.
Therefore, the Pauli proxies chosen in this way satisfy both Eqs.~\eqref{eq:repre_dressing_relation} and~\eqref{eq:repre_commute} and are independent Pauli operators.
Note that the choice of a valid set of Pauli proxies is not necessarily unique.

We provide explicit examples of a valid set of Pauli proxies for the logical magic states presented in Sec.~\ref{sec:clifford_stabilizer}.
In the surface-code $|\bar{T}\rangle$ state example, the possible Pauli dressings are $Q_q=Z_q$ for every qubit $q$. 
The Pauli stabilizers serve as their own Pauli proxies.
The Pauli proxy of the fold-transversal $\bar{H}_{XY}$ can be chosen to be the logical $X$ operator running diagonally across the surface code. 
This corresponds to removing all diagonal Clifford operators $S$, $S^\dagger$, and $CZ$ from the fold-transversal $\bar{H}_{XY}$.
In the $\llbracket 8,3,2 \rrbracket$ code $|\overline{CCZ}\rangle$ state example, the possible Pauli dressings are also $Q_q=Z_q$ for every qubit $q$. The Pauli stabilizers serve as their own Pauli proxies.
The Pauli proxies of the Clifford stabilizers $G^\mathrm{C}_1$, $G^\mathrm{C}_2$, and $G^\mathrm{C}_3$ are $\bar{X}_1$, $\bar{X}_2$, and $\bar{X}_3$, respectively.
This corresponds to removing all diagonal Clifford operators $S$ and $S^\dagger$ from the Clifford stabilizers.
In the Steane-code $|\bar{H}\rangle$ state example, the possible Pauli dressings are $Q_q=Y_q$ for every qubit $q$. The Pauli stabilizers serve as their own Pauli proxies. The Pauli proxy of the transversal $H$ can be chosen to be either transversal $X$ or transversal $Z$ acting on all qubits. 

Note that one must use the same choice of the Pauli proxies for each generator throughout the protocol.
Also, the Pauli proxy of the product of a primary generator and an element of $\mathcal{B}_{\mathrm{init}}$ must be the product of the Pauli proxy of the primary generator and the element of $\mathcal{B}_{\mathrm{init}}$ so that a different choice of generators has consistent Pauli proxies. 

We also show that whenever a Pauli-dressed original generator $\tilde{M} = D_M M$ is Hermitian, the corresponding Pauli-proxy version $\hat{P}_{\tilde{M}} = D_M\hat{P}_M$ is also Hermitian.
By Eq.~\eqref{eq:repre_dressing_relation}, 
\begin{equation}
    MD_MM=\hat{P}_MD_M\hat{P}_M.
    \label{eq:repre_dressing_relation_1}
\end{equation}
Now, assuming $\tilde{M}$ is Hermitian, we have
\begin{equation}
    D_MMD_MM=I.
    \label{eq:m_hermitian}
\end{equation}
Combining Eq.~\eqref{eq:repre_dressing_relation_1} and Eq.~\eqref{eq:m_hermitian} gives
\begin{equation}
    D_M\hat{P}_MD_M\hat{P}_M=I,
\end{equation}
showing that $\hat{P}_{\tilde{M}}$ is Hermitian.

Given valid Pauli proxies as defined above, it is guaranteed that the effective commutation between an effectively measured original generator $\tilde{M}$ and any current generator $G'{^\mathrm{fr}}$ of the frame state $|\psi^\mathrm{fr} \rangle$, given in Eq.~\eqref{eq:final_commutator_meas_simple}, is the same as the commutation between their corresponding Pauli-proxy versions $\hat{P}_{\tilde{M}}$ and $\hat{P}_{G'{^\mathrm{fr}}}$:
\begin{equation}
  [\tilde{M} ,G'{^\mathrm{fr}}]|\psi^\mathrm{fr} \rangle = \lambda |\psi^\mathrm{fr} \rangle, \quad [\hat{P}_{\tilde{M}} ,\hat{P}_{G'{^\mathrm{fr}}}] = \lambda I,
  \label{eq:compatible}
\end{equation}
where $\lambda \in \{\pm 1 \}$.
This holds since 
\begin{align}
  [\hat{P}_{\tilde{M}},
    \hat{P}_{G'{^\mathrm{fr}}}]
&=
  \prod_{k=1}^{s}
  [\hat{P}_{M},D_{G_k}][\hat{P}_{M},\hat{P}_{G_k}][D_M,\hat{P}_{G_k}]\\
&=
  \prod_{k=1}^{s}
  [\hat{P}_{M},D_{G_k}][D_M,\hat{P}_{G_k}],
  \label{eq:commu_sign_repre}
\end{align}
where we used Eq.~\eqref{eq:repre_commute}, and the commutation sign given in Eq.~\eqref{eq:commu_sign_repre} is the same as the true effective commutation sign given by Eq.~\eqref{eq:final_commutator_meas_simple} due to Eq.~\eqref{eq:repre_dressing_relation}.

\subsection{Algorithm: operator-level implementation}

Now we construct an efficient algorithm to implement the Clifford-stabilizer group update using the Pauli proxies.
In Algorithm~\ref{alg:operator}, we first introduce an algorithm in terms of data operators under Pauli-or-hook errors on data qubits that arise from circuit-level Pauli errors.
Note that although this version is sufficient for illustrating the essential steps, the circuit-level implementation presented later in Algorithm~\ref{alg:circuit} is more practical.

A Pauli-stabilizer group consisting of the Pauli proxies and Pauli dressings is called a \emph{Pauli-proxy stabilizer group}. 
In Algorithm~\ref{alg:operator}, we did not specify the input state on which the Pauli measurements are performed in Step $3$, but the input state must be a Pauli-stabilizer state. If the protocol performs the protocol measurements on an ideal target logical magic state, the input Pauli-stabilizer state is a state stabilized by the Pauli proxies.
We can also incorporate the effects of noisy injection by mapping it to ideal injection followed by errors, as detailed in Sec.~\ref{sec:scope}, although Algorithm~\ref{alg:circuit}, presented later, provides a more natural framework for doing so.

Since Step $3$ performs Pauli measurements on a Pauli-stabilizer state, we can execute it using an ordinary Clifford simulator such as \texttt{Stim}~\cite{Gidney2021stimfaststabilizer}. 
Whether the state contains a logical error is evaluated by the approach in Sec.~\ref{subsec:logical_error} using the information that we obtain in Step $3$.
\begin{figure}[tb]
\normalsize
\begin{algorithm}[Operator-level]\label{alg:operator}
\leavevmode\par
\medskip
\noindent\textbf{Input:}\par
\noindent

\begin{itemize}
    \item A sequence of Hermitian protocol measurements with a sampled configuration of Pauli-or-hook errors on data qubits.
\end{itemize}

\noindent\textbf{Output:}\par
\noindent
\begin{itemize}
    \item Measurement outcomes sampled from exactly the same probability distribution as in the input noisy protocol measurements.
    \item Occurrence of a logical error in the output state.
\end{itemize}

\noindent\textbf{Procedure:}\par
\begin{enumerate}[label=\textbf{\arabic*.}]
  \item \textbf{Transform into Pauli-dressed form:} Transform the input sequence of Hermitian protocol measurements followed by a fictitious round of syndrome extraction and logical measurements into a sequence of error-conjugated measurements (Eq.~\eqref{eq:measurement-frame-operator}) to work in the measurement-conjugated picture. For each Pauli or Clifford factor $U_i$ conjugated by an error $E$, transform $E^\dagger U_iE$ into the Pauli-dressed form $\bar{D}_{U_i}(E)U_i$ as in 
  Eq.~\eqref{eq:dressed_factor_meas}.
  \item \textbf{Construct Pauli-proxy measurements:} Replace each $U_i$ by its valid Pauli proxy $\hat{P}_{U_i}$. This step makes each protocol measurement a Pauli measurement. 
  \item \textbf{Perform Pauli-proxy measurements:} Perform the resulting sequence of Pauli measurements.
\end{enumerate}
\end{algorithm}
\end{figure}

\subsection{Algorithm: circuit-level implementation}
\label{subsec:algorithm_circuit}
So far, we have discussed the algorithm using the Pauli proxies from the perspective of operators on data qubits.
In this subsection, we explain how the procedure can be implemented in circuits. 
In circuits, a Pauli or Clifford factor $U$ is measured using a controlled-$U$ operation. 
Let
\begin{equation}
  \Lambda(U)
  =
  |0\rangle\langle0|_a\otimes I
  +
  |1\rangle\langle1|_a\otimes U
\end{equation}
be a controlled-$U$ operation with the data qubits as targets and an ancilla as the control. Suppose a data Pauli error $E$ occurs immediately before $\Lambda(U)$. Then
\begin{align}
  \Lambda(U)(I\otimes E)
  &=
  (I\otimes E) \Lambda(E^\dagger)\Lambda(U)\Lambda(E)\\
  &=(I\otimes E) \Lambda(E^\dagger U E),
  \label{eq:circuit_conju}
\end{align}
or
\begin{equation}
  \Lambda(U)(I\otimes E)
  =
  (I\otimes E) \Lambda(\bar{D}_U(E))\Lambda(U),
  \label{eq:circuit_dressing}
\end{equation}
where $\bar{D}_U(E)$ is a Pauli dressing in Eq.~\eqref{eq:dressed_factor_meas}. 
In Fig.~\ref{fig:error_prop_conju}, we show a circuit corresponding to Eq.~\eqref{eq:circuit_conju}.
Equation~\eqref{eq:circuit_dressing} shows that transforming the error-conjugated form $E^\dagger UE$ into the Pauli-dressed form $\bar{D}_U(E)U$ corresponds to propagating the error $E$ forward with the controlled-$U$ operation fixed in a circuit. 
A circuit corresponding to this Pauli-dressing form transformation by propagating errors is shown in Fig.~\ref{fig:error_prop_dress}.
\begin{figure}[t]
    \centering
    \includegraphics[width=0.85\linewidth]{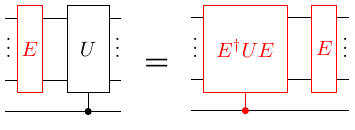}
    \caption{Error conjugation of a controlled gate in a circuit. The lowest wire is an ancilla qubit and the other qubits are data qubits. Conjugating by the error $E$ transforms the controlled-$U$ operation into a controlled-$E^\dagger UE$ operation.}
    \label{fig:error_prop_conju}
\end{figure}
\begin{figure}[t]
    \centering
    \includegraphics[width=0.95\linewidth]{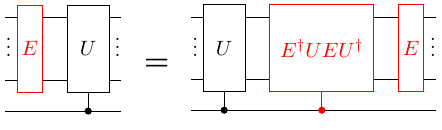}
    \caption{Error propagation through a controlled gate in a circuit. The lowest wire is an ancilla qubit and the other qubits are data qubits. The error $E$ is propagated forward through the controlled-$U$ operation. The additional controlled-$E^\dagger UE U^\dagger$, i.e., controlled-$\bar{D}_U(E)$ operation arises.}
    \label{fig:error_prop_dress}
\end{figure}
 
When $U$ is a Clifford factor, the circuit implementing the controlled-$U$ operation contains non-Clifford gates such as a non-Clifford controlled-$U$ gate or single-qubit non-Clifford gates. For example, in the surface code $|\bar{T}\rangle$ state case, the circuit for measuring $\bar{H}_{XY}$ contains $CCZ$, $T$, and $T^\dagger$ gates as shown in Fig.~\ref{fig:error_prop}. 
Thus, the procedure to replace factors by Pauli proxies in circuits is as follows. 
We first propagate errors forward with the controlled-$U$ operation fixed as shown in Fig.~\ref{fig:error_prop_dress}. 
After that, we replace gates in the circuit in such a way that the replacement corresponds to replacing $\Lambda(U)$ by $\Lambda(\hat{P}_U)$ as shown in Fig.~\ref{fig:error_prop_repr}. 
This replacement removes all non-Clifford gates from the circuit because the resulting controlled-Pauli gates $\Lambda(\hat{P}_U)$ are Clifford gates.
We call the resulting Clifford circuit the \emph{Pauli-proxy Clifford circuit}.
\begin{figure}[t]
    \centering
    \includegraphics[width=0.95\linewidth]{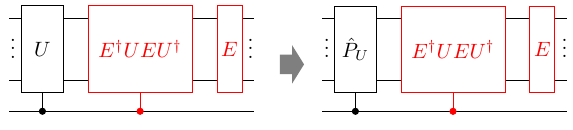}
    \caption{Replacement of a factor by a Pauli proxy in a circuit implementation. The lowest wire is an ancilla qubit and the other qubits are data qubits. After an error $E$ before the controlled-$U$ is propagated as in Fig.~\ref{fig:error_prop_dress}, gates inside the controlled-$U$ operation are replaced in such a way that the resulting Clifford circuit performs the controlled-$\hat{P}_U$ operation, where $\hat{P}_U$ is a valid Pauli proxy of $U$.}
    \label{fig:error_prop_repr}
\end{figure}

Although our proofs of effective commutation assume that an error-conjugated operator is written as the original generator multiplied by a global Pauli dressing as in Eq.~\eqref{eq:dressed_version}, the arising controlled-$\bar{D}_U(E)$ need not be propagated further before the individual factors are replaced by their Pauli proxies even when multiple factors in the measured operator $M$ act on overlapping qubit supports.
This is because replacing each $U_j$ by its Pauli proxy in the expression for $\tilde{M}$ in Eq.~\eqref{eq:each_factor_dressing} is equivalent to first transforming $\tilde{M}$ into the form of Eq.~\eqref{eq:dressed_version} and then performing the replacement.
The equivalence follows because each factor $U_j$ has the same commutation relation with a Pauli dressing as its Pauli proxy $\hat{P}_{U_j}$. Consequently, moving the dressings $D_j$ to the left through the factors $U_j$ produces the same effect as moving them to the left through the Pauli proxies $\hat{P}_{U_j}$.

Regarding practical implementation, there are some types of errors that do not need to be propagated forward in a non-Clifford circuit as our preprocessing step. 
Specifically, an error on a data qubit whose action on a factor having support on that data qubit is the same as the action on the corresponding Pauli proxy does not need to be propagated further in the non-Clifford circuit. 
This is because the Clifford circuit obtained by first propagating such errors forward in the non-Clifford circuit and then replacing certain gates is equivalent to the one obtained by applying the same replacements without first propagating the errors.
Dressing-axis Pauli errors satisfy this property, because a Pauli proxy has the same commutation relations as the corresponding factor with dressing-axis Paulis by its definition.
Thus, for example, in the surface-code $|\bar{T}\rangle$ state example, $Z$ errors on data qubits do not need to be propagated forward before the gate replacement of Fig.~\ref{fig:error_prop_repr}.
In the Steane-code $|\bar{H}\rangle$ state example, such errors are $Y$ errors on data qubits.

Note that such errors on ancilla qubits need to be propagated in general because they can propagate to errors on data qubits, and the data errors can be different between when propagating in the original non-Clifford circuit and when propagating in the Pauli-proxy Clifford circuit. 
In an implementation of the algorithm, it is enough to propagate the errors until they reach ancilla readouts in the non-Clifford circuit in the preprocessing step.
Nevertheless, depending on the circuit, some errors on ancilla qubits may not need to be propagated. For example, if the control qubit of a multi-qubit gate is an ancilla qubit, the gate commutes with $Z$ errors on the ancilla qubit just before the gate and thus the effect of the $Z$ errors does not change even if we replace gates in the non-Clifford circuit without propagating them. Thus, such $Z$ errors do not need to be propagated before the gate replacement.

There can also be Clifford errors that do not need to be propagated before the gate replacement. For example, in the surface-code $|\bar{T}\rangle$ state case, all diagonal Clifford errors do not need to be propagated before the gate replacement because they commute with the diagonal Clifford factors.
Also, the error-dependent controlled operations $\Lambda(\bar{D}_U(E))$ shown in Fig.~\ref{fig:error_prop_dress} need not be propagated further.
Data errors that do need to be propagated are referred to as \emph{continuing data errors}.
For example, in the fold-transversal surface code $|\bar{T}\rangle$ state cultivation in Sec.~\ref{sec:numerics}, only $X$ and $Y$ errors on data qubits are continuing data errors.

In summary, the algorithm to simulate a non-Clifford circuit performing protocol measurements is provided in Algorithm~\ref{alg:circuit}. 
This algorithm corresponds to a circuit-level implementation of Algorithm~\ref{alg:operator}.
\begin{figure}[tb]
\normalsize
\begin{algorithm}[Circuit-level]\label{alg:circuit}
\leavevmode\par
\medskip
\noindent\textbf{Input:}\par
\noindent

\begin{itemize}
    \item A noisy non-Clifford circuit performing protocol measurements with a sampled configuration of circuit-level Pauli errors.
\end{itemize}

\noindent\textbf{Output:}\par
\noindent
\begin{itemize}
    \item Measurement outcomes sampled from exactly the same probability distribution as in the input noisy non-Clifford circuit.
    \item Occurrence of a logical error in the output state.
\end{itemize}

\noindent\textbf{Procedure:}\par
\begin{enumerate}[label=\textbf{\arabic*.}]
  \item \textbf{Propagate errors:} In the input noisy non-Clifford circuit followed by a circuit for a fictitious round of syndrome extraction and logical measurements, propagate all sampled Pauli errors on ancilla qubits forward until they reach ancilla readouts. Also propagate all continuing data errors (defined above) until they reach the end of the circuit. 
  \item \textbf{Construct Pauli-proxy Clifford circuit:} For every controlled-$U_i$ operation corresponding to measuring a factor $U_i$ in the circuit, replace certain gates so that the controlled-$U_i$ becomes a controlled-$\hat{P}_{U_i}$ operation, where $\hat{P}_{U_i}$ is a valid Pauli proxy of $U_i$. This step removes all non-Clifford gates from the circuit, giving a Clifford circuit.
  \item \textbf{Run Pauli-proxy Clifford circuit:} Run the obtained Clifford circuit by a Clifford simulator such as \texttt{Stim}~\cite{Gidney2021stimfaststabilizer}.
\end{enumerate}
\end{algorithm}
\end{figure}

Step $1$ ensures that all controlled operations for protocol measurements are transformed to the Pauli-dressed form as shown in Fig.~\ref{fig:error_prop_dress} and the effects of ancilla errors are correctly accounted for. Note that if $\hat{P}_P=P$ in Step $2$, it is enough that in Step $1$, $\Lambda(P)$ is simply conjugated to $ \Lambda(E^\dagger PE)$ by a preceding error $E$ on the target qubit as shown in Fig.~\ref{fig:error_prop_conju} without transforming it to the Pauli-dressed form as in Fig.~\ref{fig:error_prop_dress}. This may reduce the number of gates in the circuit. 
Also, in Algorithm~\ref{alg:circuit}, we did not specify the input state of the Pauli-proxy Clifford circuit in Step $3$, but we need to input a Pauli-stabilizer state. If the protocol performs the non-Clifford circuit on an ideal target logical magic state, the input Pauli-stabilizer state is a state stabilized by the Pauli proxies.
As detailed in Sec.~\ref{sec:scope}, we can also input noisy non-Clifford circuits for injection, such as unitary injection circuits or certain measurement-based injection circuits, which are used in practical protocols.
Finally, whether the output state contains a logical error is determined using the approach described in Sec.~\ref{subsec:logical_error}, based on the information obtained in Step 3.

Note that we can also apply a correction in the middle of a circuit based on midcircuit measurement outcomes in the following way.
First, perform Step $1$ for the full non-Clifford circuit. 
On the portion of the full non-Clifford circuit up to a midcircuit measurement, perform Steps $2$ and $3$, obtaining the midcircuit measurement outcome. 
Next, apply a correction operator as an error immediately before the rest of the full non-Clifford circuit and propagate it in the same way as Step $1$. 
Finally, perform Steps $2$ and $3$ to continue the simulation of the full non-Clifford circuit.
Multiple midcircuit measurements and corrections can be handled by repeating the same procedure.
As we can see, the corrections in the middle of the circuit thus correspond to dynamically changing the Pauli-proxy Clifford circuit.

\begin{figure*}[!t]
    \centering
    \includegraphics[width=0.9\linewidth]{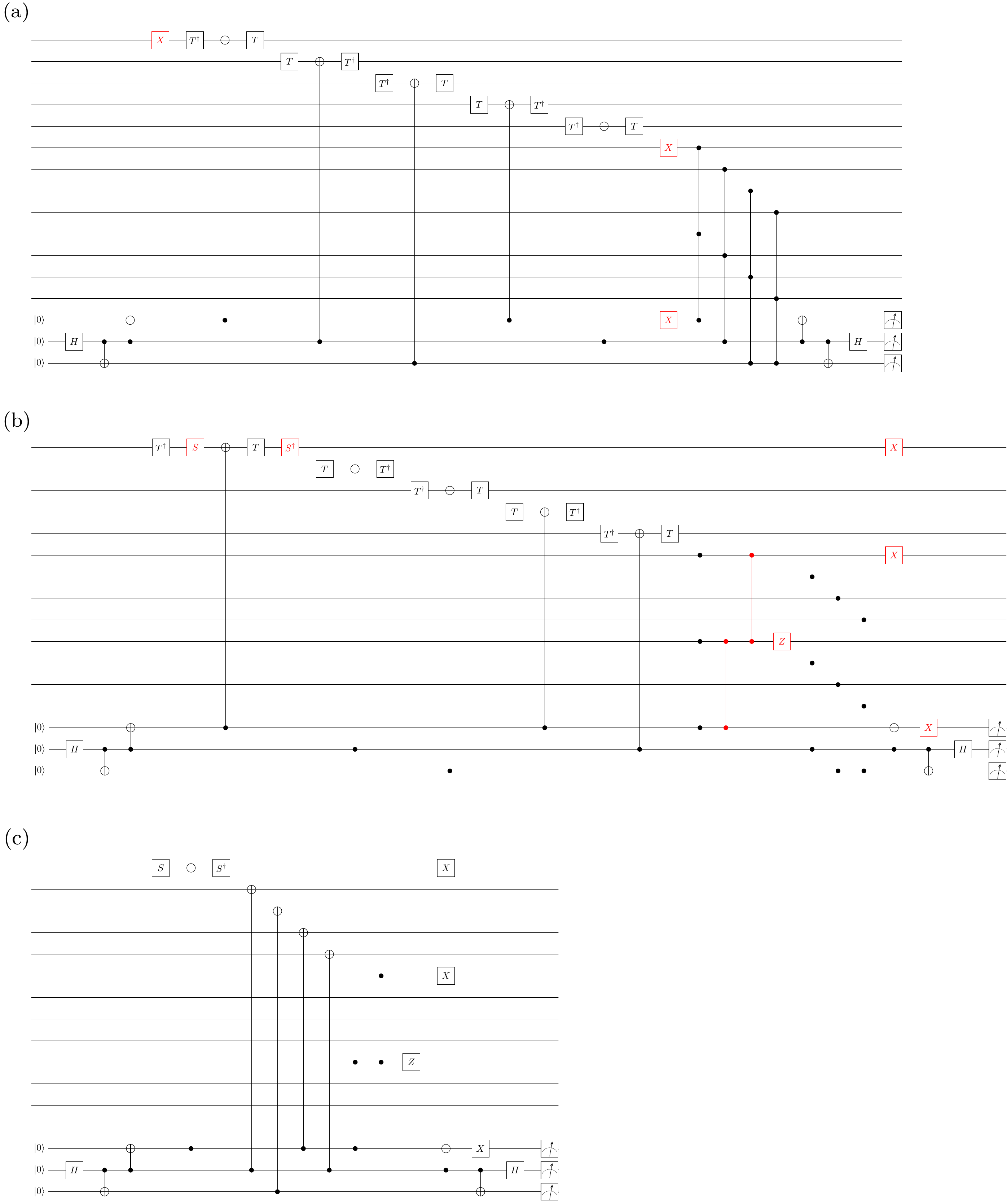}
    \caption{A concrete example of Algorithm~\ref{alg:circuit}. (a) The noisy circuit for measuring $\bar{H}_{XY}$ of the distance-3 surface code $|\bar{T}\rangle$ state using a 3-qubit cat state as an ancilla. The gates shown in red represent errors. (b) Step 1. $X$ errors are propagated forward. All diagonal errors need not be propagated further. (c) Step 2. All non-Clifford gates are removed from the circuit of (b) so that the ideal $\bar{H}_{XY}$ becomes the logical $X$ operator of the code, which is a valid Pauli proxy of the $\bar{H}_{XY}$ operator. This obtained Clifford circuit is executed in Step 3.}
    \label{fig:error_prop_surface}
\end{figure*}

In Fig.~\ref{fig:error_prop_surface}, we show a concrete example of Algorithm~\ref{alg:circuit} in the case of the measurement of $\bar{H}_{XY}$ for the distance-3 surface code $|\bar{T}\rangle$ state, which is a measurement performed in the fold-transversal surface code cultivation in Sec.~\ref{sec:numerics}.
Suppose the three $X$ errors occur in the circuit as shown in Fig.~\ref{fig:error_prop_surface}(a). 
When the errors are propagated forward, we obtain the circuit shown in Fig.~\ref{fig:error_prop_surface}(b). 
Note that all diagonal errors do not need to be propagated further in this example. 
Then, by removing all non-Clifford gates from the circuit so that the ideal $\bar{H}_{XY}$ becomes the corresponding valid Pauli proxy, i.e., the logical $X$ operator of the surface code, we obtain a Clifford circuit shown in Fig.~\ref{fig:error_prop_surface}(c). 
Then, by running the Clifford circuit by a Clifford simulator with an input state specified by an intended Pauli-proxy stabilizer group, we obtain the measurement outcomes sampled from exactly the same probability distribution as in the target noisy non-Clifford circuit.

\subsection{Time and space complexities}
\label{subsec:complexity}

In this section, we analyze the time and space complexities of Algorithm~\ref{alg:circuit}.
We first establish notation. Let $n$ be the total number of physical qubits being used, including data qubits and ancilla qubits. Let $r$ be the number of protocol measurements being performed in the circuit.
Let $r_\mathrm{q}$ denote the maximum number of protocol measurements, among the $r$ measurements, whose supports contain the same data qubit.
Let $g$ be the number of unitary gates in the original non-Clifford circuit. 
Let $w$ be the maximum number of continuing data errors created by one sampled Pauli error. 
This parameter includes hook errors from a Pauli error on ancilla qubits.
Let $\alpha$ be the maximum number of controlled gates in the circuit for a single protocol measurement whose supports contain the same data qubit.
Let $m$ be the maximum number of ancilla readouts in the circuit for a single protocol measurement.
Note that the number of readouts is the same as that in the Pauli-proxy Clifford circuit.
If we use a single ancilla qubit for each protocol measurement, $m=1$. On the other hand, if we perform each protocol measurement using multiple ancilla qubits, such as cat states~\cite{548464} or flag qubits~\cite{PhysRevLett.121.050502,chao2018fault,Chamberland2018flagfaulttolerant,PRXQuantum.1.010302}, then $m>1$. The total number of readouts is at most $rm$. The notations are summarized in Table~\ref{tab:complexity_notations}.

{\renewcommand{\arraystretch}{1.5}
\begin{table*}
\centering
\caption{Parameters used in our complexity analysis in Secs.~\ref{subsubsec:time} and~\ref{subsubsec:space}.}
\begin{tabular}{ll}
\hline \hline
\parbox[t]{2.5cm}{\raggedright Notation}
& \parbox[t]{14cm}{\raggedright Meaning} \\
\hline

\parbox[t]{2.5cm}{\raggedright $n$}
& \parbox[t]{14cm}{\raggedright Total number of physical qubits being used, including data qubits and ancilla qubits.} \\

\parbox[t]{2.5cm}{\raggedright $r$}
& \parbox[t]{14cm}{\raggedright Number of protocol measurements performed in the circuit.} \\

\parbox[t]{2.5cm}{\raggedright $r_{\mathrm{q}}$}
& \parbox[t]{14cm}{\raggedright Maximum number of protocol measurements, among the $r$ measurements, whose supports contain the same data qubit.} \\

\parbox[t]{2.5cm}{\raggedright $g$}
& \parbox[t]{14cm}{\raggedright Number of unitary gates in the original non-Clifford circuit.} \\

\parbox[t]{2.5cm}{\raggedright $w$}
& \parbox[t]{14cm}{\raggedright Maximum number of continuing data errors created by one sampled Pauli error.} \\

\parbox[t]{2.5cm}{\raggedright $\alpha$}
& \parbox[t]{14cm}{\raggedright Maximum number of controlled gates in the circuit for a single protocol measurement whose supports contain the same data qubit.} \\

\parbox[t]{2.5cm}{\raggedright $m$}
& \parbox[t]{14cm}{\raggedright Maximum number of ancilla readouts in the circuit for a single protocol measurement.} \\

\parbox[t]{2.5cm}{\raggedright $e$}
& \parbox[t]{14cm}{\raggedright Number of nontrivial sampled circuit-level Pauli errors in the original non-Clifford circuit.} \\

\parbox[t]{2.5cm}{\raggedright $\hat{g}$}
& \parbox[t]{14cm}{\raggedright Total number of unitary gates in the Pauli-proxy Clifford circuit.} \\

\parbox[t]{2.5cm}{\raggedright $N_{\mathrm{loc}}$}
& \parbox[t]{14cm}{\raggedright Number of possible error locations.} \\

\parbox[t]{2.5cm}{\raggedright $N_{\mathrm{idle}}$}
& \parbox[t]{14cm}{\raggedright Number of idling error locations.} \\

\parbox[t]{2.5cm}{\raggedright $d$}
& \parbox[t]{14cm}{\raggedright Code distance.} \\

\parbox[t]{2.5cm}{\raggedright $k$}
& \parbox[t]{14cm}{\raggedright Number of logical qubits of the target logical magic state.} \\

\hline \hline
\end{tabular}
\label{tab:complexity_notations}
\end{table*}
}

\subsubsection{Time complexity}
\label{subsubsec:time}

We first consider the time complexity of Steps $1$ and $2$ in Algorithm~\ref{alg:circuit}. Let us consider one continuing data error before a protocol measurement. A continuing data error is a one- or two-qubit Pauli or Clifford error.
Since each factor $U$ is also a one- or two-qubit operator, when one continuing data error is propagated through one controlled-$U$ gate, the generated $\Lambda(\bar{D}_U(E))$ operation consists of only $O(1)$ Pauli or Clifford gates.
A continuing data error acts on at most two data qubits, and each data qubit is touched by at most $O(\alpha)$ gates in one protocol measurement circuit. Each interaction costs $O(1)$ time and generates $O(1)$ additional Pauli or Clifford gates. Hence, one continuing data error costs $O(\alpha)$ per protocol measurement, and it generates $O(\alpha)$ additional gates in the Pauli-proxy Clifford circuit.

Let $e$ be the number of nontrivial sampled circuit-level Pauli errors in the original non-Clifford circuit. Since one sampled circuit-level Pauli error creates at most $w$ continuing data errors, $e$ sampled Pauli errors create at most $O(ew)$ continuing data errors, and thus the total propagation cost is $O(ewr_\mathrm{q}\alpha)$.
The total number of additional gates inserted into the Pauli-proxy Clifford circuit is also $O(ewr_\mathrm{q}\alpha)$.
Note that the propagation cost of ancilla errors and the resulting number of gates on ancilla qubits are not dominant throughout the analysis.
From the discussion above, we can formalize the number of gates in the Pauli-proxy Clifford circuit. Let $\hat{g}$ be the total number of unitary gates in the Pauli-proxy Clifford circuit. In the absence of errors, the Clifford circuit has at most $O(g)$ gates, because each gate is replaced by a Pauli or Clifford gate, and sometimes by the identity gate. Sampled circuit-level Pauli errors generate at most $O(ewr_\mathrm{q}\alpha)$ additional gates. Therefore, $\hat{g}$ is given by
\begin{equation}
  \hat{g}=O(g+ewr_\mathrm{q}\alpha).
\end{equation}
Now we discuss the entire time complexity of Algorithm~\ref{alg:circuit} excluding the process of evaluating a logical error.
In an ordinary Clifford simulator using a tableau algorithm, a Pauli or Clifford gate is applied in $O(n)$ time, and a Pauli measurement is performed in $O(n^2)$ time. Note that if a measurement has a deterministic outcome, it is simulated in $O(n)$~\cite{Gidney2021stimfaststabilizer}, but in general a measurement takes $O(n^2)$ when the outcome is random. Therefore, since the Pauli-proxy Clifford circuit has $\hat{g}$ Pauli or Clifford gates and $rm$ measurement readouts, simulating the circuit costs
\begin{align}
  T&=O(\hat{g}n+rmn^2)\\
  &=O\!\left((g+ewr_\mathrm{q}\alpha )n+rmn^2\right).
  \label{eq:time_complexity}
\end{align}
Since the preprocessing step to build the Pauli-proxy Clifford circuit is $O(ewr_\mathrm{q}\alpha)$ as discussed, the entire time complexity is given by Eq.~\eqref{eq:time_complexity}.
The complexity is polynomial in the number of physical qubits. It is also polynomial in the number of physical non-Clifford gates for the following reasons. The total number of gates $g$ scales linearly with the number of non-Clifford gates. The number of gates touched by a single continuing data error, which is bounded by $O(r_q \alpha)$, is smaller than $g$. Also, $w$ and $r$ are smaller than $g$. Finally, $e$ scales at most linearly with $g$, and $m$ is independent of $g$.
Furthermore, the complexity has no separate dependence on the number of logical qubits and on the stabilizer or Pauli rank of the target logical magic state.

Now we compare the time complexity of our algorithm for simulating a noisy non-Clifford circuit with that of an ordinary Clifford simulator for simulating a noisy Clifford circuit with the same circuit size and the same number of sampled circuit-level Pauli errors.
While we have already shown that our algorithm has polynomial-time complexity, this comparison indicates how closely the complexity of our algorithm approaches that of Clifford simulation.
Clifford simulators such as Stim~\cite{Gidney2021stimfaststabilizer} can process each Pauli error in $O(1)$ time by utilizing a Pauli frame.
Thus, if the Clifford circuit has $g$ gates, $O(rm)$ measurement readouts, and $e$ sampled circuit-level Pauli errors, its simulation time is
\begin{equation}
    T_{\mathrm{Cliff}}=O\!\left(gn+e+rmn^2\right).
    \label{eq:time_complexity_Clifford}
\end{equation}
Note that Eq.~\eqref{eq:time_complexity_Clifford} can often be simplified to $O\!\left(gn+rmn^2\right)$ because $O(e)$ is often not a dominant factor due to the following reason.
The number of possible error locations $N_{\mathrm{loc}}$ in our setting is
\begin{equation}
  N_{\mathrm{loc}}=O(g+n+rm+N_{\mathrm{idle}}),
  \label{eq:n_tot}
\end{equation}
where $N_{\mathrm{idle}}$ denotes the number of idling error locations.
Here, the term $g$ corresponds to gate errors, the term $n$ accounts for errors on data-qubit initializations, and the term $rm$ accounts for errors on ancilla-qubit initializations and measurement readouts. 
Typically, each protocol measurement circuit has $O(n)$ controlled gates.
For $r$ protocol measurements, even if the circuit depth of each protocol measurement is $O(n)$, the number of idling error locations is $N_{\mathrm{idle}}=O(rn^2)$. 
Also, since $e$ is smaller than $N_{\mathrm{loc}}$, we have
\begin{equation}
  e=O(N_{\mathrm{loc}}).
  \label{eq:e_n_tot}
\end{equation}
Therefore, by Eqs.~\eqref{eq:n_tot} and~\eqref{eq:e_n_tot}, $e$ is not dominant in Eq.~\eqref{eq:time_complexity_Clifford}.

Therefore, the difference between Eqs.~\eqref{eq:time_complexity} and~\eqref{eq:time_complexity_Clifford} lies in the term $O(ewr_\mathrm{q}\alpha n)$. 
In general, Eq.~\eqref{eq:time_complexity} can be larger than Eq.~\eqref{eq:time_complexity_Clifford} due to this difference. We now discuss how much the two complexities differ in typical protocols. 
In typical logical magic state protocols, we have $\alpha=O(1)$, which includes the cases where Clifford stabilizers are transversal. 
Moreover, if the code is a topological quantum low-density parity-check (qLDPC) code, for 1 round of stabilizer and logical Clifford measurements, we have $r_\mathrm{q}=O(1)$. For $O(d)$ rounds, which is a typical number required for fault tolerance, we have $r_\mathrm{q}=O(d)$, where $d$ is the code distance. 
Thus, the term $O(ewr_\mathrm{q}\alpha n)$ in Eq.~\eqref{eq:time_complexity} is typically equal to $O(ewdn)$.

As a concrete example, we consider magic state cultivation performing protocol measurements using cat states of the same order in size as the data qubits. In this protocol, $w=O(1)$ and $N_{\mathrm{idle}}=O(rn)$ for protocol measurements. Hence, $O(ewdn)$ is equal to $O(drn^2)$ due to Eqs.~\eqref{eq:n_tot} and~\eqref{eq:e_n_tot}. Here, we assume $O(g)=O(rn)$.
Since $m=O(n)$ and thus $O(rmn^2)=O(rn^3)$, the term $O(ewr_\mathrm{q}\alpha n)$ is not dominant in Eq.~\eqref{eq:time_complexity}.
Therefore, the time complexity of our algorithm for simulating magic state cultivation in this setup is the same as that of ordinary Clifford simulation for simulating a noisy Clifford circuit with the same circuit size and the same number of sampled errors. 

The complexity analysis discussed above uses a conservative raw error-propagation bound. In this model, the continuing data errors generated by sampled errors are represented as an ordered list of elementary Pauli or Clifford errors and are propagated independently through later protocol measurements. This is
analogous to explicit gate-list propagation, as in Ref.~\cite{fby6-xjbm}, where a propagated Clifford error is maintained as an ordered list and pushed through subsequent measurement rounds.
A more efficient implementation could instead maintain a canonical error frame. In this representation, continuing data errors are inserted into the current frame, and errors of the same type and support are combined or cancelled before later propagation. This can reduce the effective number of continuing errors that must be propagated. The time complexity of this implementation is given by
\begin{equation}
    T=O\!\left((g+r_\mathrm{q}\alpha n^2)n+ew+rmn^2\right).
    \label{eq:time_complexity_canonical}
\end{equation}
Here, given that the continuing data errors consist of a constant number of types of one- and two-qubit Pauli or Clifford errors, the size of the canonical frame is $O(n^2)$, since there are $O(n)$ possible one-qubit supports and $O(n^2)$ possible two-qubit supports.
Thus, the propagation cost is bounded by $O(r_\mathrm{q}\alpha n^2)$.
Inserting $O(ew)$ continuing data errors into a current canonical frame takes $O(ew)$ time.

Finally, we discuss the evaluation of a logical error. Here we use the conservative raw error-propagation algorithm. As described in Sec.~\ref{subsec:logical_error}, we perform one round of fictitious syndrome extraction and logical measurements at the end of the protocol to evaluate whether the output state has a logical error for each Monte Carlo shot.
Thus, we have additional circuits for measuring $O(n-k)$ stabilizers of the code and $k$ global Clifford operators at the end of the circuit, where $k$ is the number of logical qubits of the target logical magic state.
The time complexity of our algorithm including the evaluation of a logical error is then modified from Eq.~\eqref{eq:time_complexity} as follows.
\begin{align}
  T'
  &=O\!\left(
      (g+(n-k)n+kn^2+(r_\mathrm{q}+k)ew\alpha)n
      \right. \notag\\
  &\qquad\left.
      +(rm+n)n^2
      \right)\\
  &=O\!\left((g+kn^2+(r_\mathrm{q}+k)ew\alpha)n
      +rmn^2\right).
  \label{eq:time_complexity_logical_err}
\end{align}
Here, we assume that each of the $O(n-k)$ stabilizers of the code has $O(n)$ factors and we measure each of them at least once in the protocol. We also used a bound that each of the $k$ global Clifford operators has at most $O(n^2)$ factors.
About ancilla readouts, we use the fact that one ancilla qubit is enough for each fictitious logical measurement because the measurement gadgets are noiseless, thus there is no need to use multiple ancilla qubits.
Note that we need to determine whether the single-qubit Pauli measurement commutes with all current stabilizer generators or anticommutes with at least one of them to evaluate a logical error, and Clifford simulators do this in $O(n)$ time.
From Eq.~\eqref{eq:time_complexity_logical_err}, we see that the time complexity of our algorithm for each Monte Carlo shot including the process of evaluating a logical error is linear in the number of logical qubits $k$.

\subsubsection{Space complexity}
\label{subsubsec:space}

We now analyze the space complexity. 
In Algorithm~\ref{alg:circuit}, naively we store the Pauli-proxy Clifford circuit as a complete gate list using $O(g+ewr_\mathrm{q}\alpha)$ memory before executing it. Alternatively, the algorithm can also be implemented more memory-efficiently with the same time complexity by processing the circuit sequentially. In this implementation, at each step, we propagate the current continuing data errors through a gate in the original non-Clifford circuit, replace the gate by a valid Clifford gate and generate the additional Pauli or Clifford gates, immediately apply these gates to the tableau, and then discard their descriptions.
After the gates have been applied, their effects are stored in the tableau itself. Therefore, even though the propagation procedure generates $O(ewr_\mathrm{q}\alpha)$ additional gates, these gates do not have to be stored simultaneously.
Thus the working memory consists of the current stabilizer tableau of the Pauli-proxy Clifford circuit and the current continuing data errors. 
A stabilizer tableau on $n$ physical qubits requires $O(n^2)$ memory. 
In the conservative raw error-propagation algorithm, the memory needed to store the current continuing data errors is $O(ew)$. 
Thus, the total required memory $R$ is 
\begin{equation}
  R=O(n^2+ew).
\end{equation}
We remark that this memory bound can be improved by storing continuing data errors in a canonical error frame in a similar way to the time complexity in Eq.~\eqref{eq:time_complexity_canonical}. Since a canonical frame for continuing data errors can be stored using at most $O(n^2)$ memory, with this implementation, we have 
\begin{equation}
  R=O(n^2).
\end{equation}
This is the same memory scaling as ordinary Clifford simulation for a Clifford circuit with the same number of physical qubits.

\subsubsection{Comparison with prior work}
\label{subsubsec:comparison}

Now we discuss the difference in the complexities between our algorithm and the one proposed in Ref.~\cite{fby6-xjbm} that can efficiently simulate a similar class of the logical magic states.
Both algorithms propagate sampled circuit-level errors forward through an original non-Clifford circuit as a preprocessing step, but they differ essentially in how far the errors must be propagated, and also the subsequent algorithms are fundamentally different. 
As a result of these differences, our algorithm has lower time and space complexities.
In particular, Ref.~\cite{fby6-xjbm} needs to propagate all errors, including the ones produced while propagating errors up to the end of the circuit, or in practice, up to the last layer of non-Clifford gates.
On the other hand, in our algorithm, we do not propagate all errors forward up to the end of the circuit. For example, the error-dependent controlled operation $\Lambda(\bar{D}_U(E))$ in Fig.~\ref{fig:error_prop_dress} does not need to be propagated further, and there are also other types of errors that do not have to be propagated further as detailed in Sec.~\ref{subsec:algorithm_circuit}. 
This distinction arises from the different objectives of error propagation. In Ref.~\cite{fby6-xjbm}, errors are propagated to obtain an ideal circuit followed by Clifford errors. By contrast, our algorithm propagates them to ensure that all elements of $\mathcal{B}_{\mathrm{gc}}$ have eigenvalues $+1$ on the frame states and also each factor has the Pauli-dressed form.
In what follows, we discuss the difference more specifically using the same parameters as in the analysis of Ref.~\cite{fby6-xjbm}.
{\renewcommand{\arraystretch}{1.5}
\begin{table*}
\centering
\caption{Parameters used in our comparison with prior work in Sec.~\ref{subsubsec:comparison}.}
\begin{tabular}{ll}
\hline \hline
\parbox[t]{2.5cm}{\raggedright Notation}
& \parbox[t]{14cm}{\raggedright Meaning} \\
\hline

\parbox[t]{2.5cm}{\raggedright $w_{\mathrm{q}}$}
& \parbox[t]{14cm}{\raggedright Maximum number of Pauli stabilizers in a round whose supports contain the same data qubit.} \\

\parbox[t]{2.5cm}{\raggedright $w_{\mathrm{c}}$}
& \parbox[t]{14cm}{\raggedright Maximum number of data qubits involved in a Pauli stabilizer.} \\

\parbox[t]{2.5cm}{\raggedright $\ell$}
& \parbox[t]{14cm}{\raggedright Number of factors in a Clifford stabilizer.} \\

\parbox[t]{2.5cm}{\raggedright $r_{\mathrm{S}}$}
& \parbox[t]{14cm}{\raggedright Number of Pauli-stabilizer measurements rounds.} \\

\parbox[t]{2.5cm}{\raggedright $r_{\mathrm{L}}$}
& \parbox[t]{14cm}{\raggedright Number of Clifford-stabilizer measurements rounds.} \\

\parbox[t]{2.5cm}{\raggedright $p$}
& \parbox[t]{14cm}{\raggedright Pauli rank of the target logical magic state.} \\

\parbox[t]{2.5cm}{\raggedright $q$}
& \parbox[t]{14cm}{\raggedright Stabilizer rank of the target logical magic state.} \\
\hline \hline
\end{tabular}
\label{tab:comparison_notations}
\end{table*}
}

In Ref.~\cite{fby6-xjbm}, the complexities are analyzed using the following parameters. A data qubit is involved in at most $w_\mathrm{q}$ Pauli stabilizers in a round. A Pauli stabilizer involves at most $w_\mathrm{c}$ data qubits. Let $\ell$ be the number of factors in a Clifford stabilizer. Let $r_\mathrm{S}$ and $r_\mathrm{L}$ be the number of rounds of Pauli- and Clifford-stabilizer measurements, respectively. Here, one round among $r_\mathrm{L}$ rounds consists of one transversal Clifford-stabilizer measurement.
The notations are summarized in Table~\ref{tab:comparison_notations}.

For a Pauli error on the data qubits, the number of gates generated by propagating it is $O(r_\mathrm{L}(r_\mathrm{S}w_\mathrm{q}+r_\mathrm{L}))$ in Ref.~\cite{fby6-xjbm}, compared with $O(r_\mathrm{S}w_\mathrm{q}+r_\mathrm{L})$ in our algorithm. 
For a Pauli error on the ancilla qubits used in a Pauli stabilizer measurement, the corresponding gate counts are $O(r_\mathrm{L}w_\mathrm{c}(r_\mathrm{S}w_\mathrm{q}+r_\mathrm{L}))$ and $O(w_\mathrm{c}(r_\mathrm{S}w_\mathrm{q}+r_\mathrm{L}))$, respectively.
For a Pauli error on the ancilla qubits used in a Clifford stabilizer measurement, they are $O(\ell r_\mathrm{S}w_\mathrm{q}(r_\mathrm{S}w_\mathrm{q}+r_\mathrm{L}))$ and $O(\ell r_\mathrm{S}w_\mathrm{q})$, respectively. 
Thus, our algorithm has a smaller number of gates additionally generated by propagating sampled circuit-level errors forward. Since the number of gates corresponds to the cost of either combining them or explicitly applying them, the difference leads to our smaller time and also space complexities.

Our algorithm is more efficient not only due to the number of generated additional gates but also due to the subsequent procedures.
In Ref.~\cite{fby6-xjbm}, two procedures after error propagation are proposed: phase-insensitive and phase-sensitive methods. 
In the phase-insensitive method, the subsequent procedure after error propagation can be performed by Clifford simulators. However, the method needs to run $O(p^2 2^k)$ different Clifford circuits, where $p$ is the Pauli rank of the target logical magic state and $k$ is the number of logical qubits. 
In the phase-sensitive method, Clifford simulators cannot be used for computing measurement probabilities, and also the method needs to update $q$ stabilizer states and evaluate interference between them by computing $O(q^2)$ stabilizer overlaps, where $q$ is the stabilizer rank of the target logical magic state.
Both $p$ and $q$ can scale exponentially with the number of logical qubits in the target logical magic state.
Additionally, in Ref.~\cite{fby6-xjbm}, midcircuit measurements and ancilla resets are not considered despite there existing a method to simulate them. 

By contrast, in our algorithm, executing only \emph{a single} Clifford circuit after error propagation is enough to obtain the measurement outcomes of a noisy non-Clifford circuit and evaluate a logical error for each shot. Also, the time and space complexities do not explicitly depend on the stabilizer or Pauli rank of the target logical magic state.
The time complexity to evaluate a logical error increases only linearly in the number of logical qubits as in Eq.~\eqref{eq:time_complexity_logical_err}.
Furthermore, our algorithm can naturally simulate midcircuit measurements, ancilla resets, and active Pauli or Clifford corrections during the circuit.

\section{Scope of the algorithm}\label{sec:scope}
So far, we have restricted our attention to protocol measurements. While protocol measurements are common operations in logical magic state protocols, for other fault-tolerant protocols we may wish to consider additional operations. In this section, we describe the class of operations that the Clifford-stabilizer simulation method can simulate beyond protocol measurements.
Here, we focus on diagonal logical magic states.
We base our discussion on Algorithm~\ref{alg:circuit}.

\subsection{Unitary gates}
We consider a unitary gate $U$ on data qubits that, in the absence of errors, transforms a noiseless state $|\psi\rangle$ into a noiseless diagonal logical magic state $|\bar{\phi}\rangle$:
\begin{equation}
    |\bar{\phi}\rangle=U|\psi\rangle.
\end{equation}
Here, $|\psi\rangle$ can be any state including the tensor product of physical states, a single- or multi-qubit logical Pauli-stabilizer state, and a single- or multi-qubit logical magic state.
The state $|\bar{\phi}\rangle$ may involve one or more logical qubits.
The unitary gate $U$ can be composed of several elementary unitary gates, including Pauli, Clifford, and third-level non-Clifford gates.
Protocol measurements are performed on $|\bar{\phi}\rangle$, that is, Pauli or Clifford stabilizers in the Clifford-stabilizer group that defines $|\bar{\phi}\rangle$ are measured. 
In Secs.~\ref{sec:clifford_stb_update_diagonal} and~\ref{sec:clifford_stb_update}, we showed that our Clifford-stabilizer simulation method can simulate an arbitrary sequence of Pauli or diagonal Clifford data errors arising from circuit-level Pauli errors during protocol measurements on $|\bar{\phi}\rangle$.
In fact, as long as the errors are restricted to Pauli or diagonal Clifford data errors, they can be simulated for arbitrary error supports.
Thus, if errors in the circuits preparing $|\psi\rangle$ and implementing $U$ are equivalent to a product $E$ of Pauli or diagonal Clifford errors such that 
\begin{equation}U_{\mathrm{noisy}}|\psi_{\mathrm{noisy}}\rangle=EU|\psi\rangle=E|\bar{\phi}\rangle, 
    \label{eq:scope_error_condition}
\end{equation}
then our Clifford-stabilizer simulation algorithm is directly applicable.
Here, $U_{\mathrm{noisy}}$ and $|\psi_{\mathrm{noisy}}\rangle$ are noisy versions of $U$ and $|\psi\rangle$, respectively, containing errors in their corresponding circuits.
Note that given a circuit satisfying the property in Eq.~\eqref{eq:scope_error_condition}, Step 1 of Algorithm~\ref{alg:circuit} transforms the state to the right-hand side of Eq.~\eqref{eq:scope_error_condition}.
Here, diagonal errors generated when propagating errors need not be propagated further.
In Step 2 of Algorithm~\ref{alg:circuit}, the logical magic state $|\bar{\phi}\rangle$ can be replaced by a Pauli-stabilizer state stabilized by a valid set of Pauli proxies typically by just removing the diagonal non-Clifford gates in $U$ and in the preparation circuit for $|\psi\rangle$. 

Indeed, many typical unitary operations in logical magic state protocols satisfy Eq.~\eqref{eq:scope_error_condition} as we show below. 
The first examples are unitary encoding circuits that prepare logical magic states.
Usually, a unitary encoding circuit consists of physical $|0\rangle$ and $|+\rangle$ state preparations, $\CNOT$ gates, and physical non-Clifford gates. 
Thus, in this case $|\psi\rangle$ is the tensor product of $|0\rangle$ and $|+\rangle$ states, and $|\psi_{\mathrm{noisy}}\rangle$ may include a Pauli error on each of $|0\rangle$ and $|+\rangle$.
The unitary $U$ is a sequence of $\CNOT$ and non-Clifford gates, and $U_{\mathrm{noisy}}$ may include Pauli errors after each gate. 
If we propagate the sampled circuit-level Pauli errors in $|\psi_{\mathrm{noisy}}\rangle$ and $U_{\mathrm{noisy}}$ forward through the circuit, they can produce diagonal Clifford errors in the middle of the circuit due to the non-Clifford gates in $U$. 
Notably, diagonal errors remain diagonal errors when propagating through $\CNOT$ gates, and diagonal errors commute with diagonal non-Clifford gates.
Hence, a noisy encoding circuit is always equivalent to the corresponding noiseless circuit followed by Pauli or diagonal Clifford errors, satisfying Eq.~\eqref{eq:scope_error_condition}.
Therefore, our Clifford-stabilizer simulation is directly applicable to this case. 

As a second example, we consider logical non-Clifford gates, including transversal non-Clifford gates. 
Consider a typical case involving a logical diagonal non-Clifford gate $U$ that consists of diagonal physical gates in the third level of the Clifford hierarchy.  Let $U$ be applied to a Pauli-stabilizer state $|\psi\rangle$ in order to prepare a logical magic state. 
Note that $U$ does not have to be transversal in general.
In this situation, $|\psi_{\mathrm{noisy}}\rangle$ is equivalent to $E_{\psi}|\psi\rangle$, where $E_{\psi}$ is a combination of Pauli errors.
If a Pauli error is conjugated by a diagonal third-level gate, the result is a diagonal Clifford multiplied by a Pauli, up to a phase.
Hence, the property in Eq.~\eqref{eq:scope_error_condition} is satisfied, and thus our Clifford-stabilizer simulation directly applies to this case. 

Yet another example is the use of $\CNOT$ gates to transform a logical magic state encoded in one code into a logical magic state encoded in another code. 
In this case, $|\psi\rangle$ corresponds to the logical magic state before the transformation and $|\bar{\phi}\rangle$ corresponds to the logical magic state after the transformation. The unitary $U$ is made up of $\CNOT$ gates.
Since diagonal errors remain diagonal errors when propagating through $\CNOT$ gates as discussed above, the condition Eq.~\eqref{eq:scope_error_condition} is satisfied, thereby our Clifford-stabilizer simulation is directly applicable.
Such operations are used in the growing operations of the fold-transversal surface code cultivation in Sec.~\ref{sec:numerics}, in the logical teleportation for code-switching between 3D and 2D color codes~\cite{PhysRevResearch.7.023080,Heussen2025efficientfault}, and in magic state distillation~\cite{PhysRevA.71.022316,PhysRevA.86.032324,PhysRevA.86.052329,Litinski2019magicstate}. 

In some cases where we perform protocol measurements on $|\psi\rangle$, the state before applying $U$ may not be written as $E_{\psi}|\psi\rangle$ for some error $E_{\psi}$ even after the errors have been propagated. 
Instead, it can be written as $E_{\psi'}|\psi'\rangle$ for some error $E_{\psi'}$, where $|\psi'\rangle$ is stabilized by a Clifford-stabilizer group that is different from the one defining the ideal $|\psi\rangle$. 
This can happen when a protocol measurement anticommutes with a Clifford-stabilizer group due to errors as explicitly exemplified in Sec.~\ref{subsec:clifford_stb_update_ex_surface}. 
Accordingly, the resulting state is expressed as 
\begin{equation}
    E|\bar{\phi}'\rangle=EU|\psi'\rangle, 
    \label{eq:scope_error_diff}
\end{equation}
where $|\bar{\phi}'\rangle:=U|\psi'\rangle$, and $E$ is a product of Pauli or diagonal Clifford errors.
Although Eq.~\eqref{eq:scope_error_diff} is different from Eq.~\eqref{eq:scope_error_condition}, Clifford-stabilizer simulation is still directly applicable in such a situation for the following reason.
Every generator of the Clifford-stabilizer group of $|\psi'\rangle$ is a product of generators of the Clifford-stabilizer group of $|\psi\rangle$ with each of them possibly $Z$-type Pauli-dressed. 
After applying $U$, each generator of $|\psi'\rangle$ is transformed to a product of generators of $|\bar{\phi}\rangle$ possibly with $Z$-type Pauli dressings. 
Since the state $|\bar{\phi}'\rangle$ having such a form of generators is exactly the case considered in Secs.~\ref{sec:clifford_stb_update_diagonal} and~\ref{sec:clifford_stb_update}, our Clifford-stabilizer simulation directly applies.
The fold-transversal surface code cultivation in Sec.~\ref{sec:numerics} corresponds to this situation, where $|\psi\rangle$ corresponds to the $d=3$ regular surface code $|\bar{T}\rangle$ state with physical $|0\rangle$ and $|+\rangle$ states, and $|\bar{\phi}\rangle$ corresponds to the $d=5$ regular surface code $|\bar{T}\rangle$ state. The unitary $U$ corresponds to the $\CNOT$ growing operations. The same argument holds in the case of growing from $d=5$ to $d=7$.

\subsection{Measurements}
\label{subsec:new_measurement}
We now consider measuring operators that are not protocol measurements.
Even if a measurement is not a protocol measurement, as long as the measurement anticommutes with at least one generator and either commutes or anticommutes with every other generator of the current Clifford-stabilizer group, the same update rule applies. 
Furthermore, there are cases where the Clifford-stabilizer simulation can correctly simulate measurements that neither commute nor anticommute with some current tracked generators. 
In the following, we explain how such non-protocol measurements can be handled using Algorithm~\ref{alg:circuit}.

We first consider non-protocol measurements that anticommute with at least one generator and either commute or anticommute with every other generator of the current Clifford-stabilizer group. 
As detailed below, there are more complications for anticommuting measurements than for protocol measurements.
Before analyzing non-protocol measurements, we revisit anticommuting protocol measurements in Algorithm~\ref{alg:circuit}.
Suppose we perform a Hermitian protocol measurement $M_1$ on the current state $|\psi\rangle$, and $M_1$ effectively anticommutes with a stabilizer $G$:
\begin{equation}
    M_1G|\psi\rangle=-GM_1|\psi\rangle.
\end{equation}
The state $|\psi\rangle$ is projected onto the $+1$ or $-1$ branches of the measurement $M_1$ with equal probabilities. 
The two branches are related by 
\begin{equation}
    \Pi_-(M_1)|\psi\rangle=G\Pi_+(M_1)|\psi\rangle,
\end{equation}
where $\Pi_m(M_1)$ are the projectors defined in Eq.~\eqref{eq:projector}.
Thus, the $-1$ branch is related to the $+1$ branch by applying $G$.
The operator $G$ is referred to as a \emph{branch-flipping operator}.
In the measurement-conjugated picture explained in Sec.~\ref{subsec:error_conjugated_meas}, a measured operator is a Pauli-dressed version $\tilde{M}_1$, and a branch-flipping operator is a stabilizer of the frame state $|\psi^\mathrm{fr} \rangle$ on which $\tilde{M}_1$ is effectively measured. We denote this branch-flipping operator by $G'{^\mathrm{fr}}$. 
In Algorithm~\ref{alg:circuit} using the measurement-conjugated picture, whenever a measurement yields the $-1$ branch, the corresponding branch-flipping operator must generally be propagated in Step 1 if it is a continuing data error as defined in Sec.~\ref{subsec:complexity}.
Remarkably, in the setting we considered in Sec.~\ref{sec:algorithm} where we only perform protocol measurements, we do not need to consider this as part of the procedures in the algorithm, and just flipping the sign of the corresponding measured Pauli operator in the \emph{Pauli-proxy stabilizer tableau}, i.e., the stabilizer tableau in the Pauli-proxy Clifford circuit, is valid. 
The reason is explained below.

In this section, we say two operators are \emph{compatible} on a state when they satisfy Eq.~\eqref{eq:compatible}, that is, the commutation of their Pauli-proxy counterparts reproduces their true effective commutation on that state.
Since the branch-flipping operator $G'{^\mathrm{fr}}$ is a stabilizer of the pre-measurement frame state $|\psi^\mathrm{fr} \rangle$, due to Eq.~\eqref{eq:compatible}, $G'{^\mathrm{fr}}$ and a future Pauli-dressed protocol measurement $\tilde{M}_2$ are compatible on the state on which $\tilde{M}_2$ is measured.
Thus, propagating $G'{^\mathrm{fr}}$ through the non-Clifford circuit measuring $\tilde{M}_2$ has the same effect on measurement probabilities as propagating $\hat{P}_{G'{^\mathrm{fr}}}$ through the Clifford circuit measuring $\hat{P}_{\tilde{M}_2}$. 
This means that we do not need to propagate $G'{^\mathrm{fr}}$ in Step 1 of Algorithm~\ref{alg:circuit}, and replacing it by $\hat{P}_{G'{^\mathrm{fr}}}$ is enough when converting the non-Clifford circuit to a Pauli-proxy Clifford circuit. 
Furthermore, we notice that the effect of applying $\hat{P}_{G'{^\mathrm{fr}}}$ immediately after the measurement of $\hat{P}_{\tilde{M}_1}$ is equivalent to flipping the sign of the measured $\hat{P}_{\tilde{M}_1}$ in the Pauli-proxy stabilizer tableau. 
Therefore, for protocol measurements, Algorithm~\ref{alg:circuit} correctly handles the $-1$ branch.
Note that even when a circuit contains unitary operations to grow the code such as those implemented using $\CNOT$ gates, each possible branch-flipping operator is compatible with every future Pauli-dressed protocol measurement. This is because a branch-flipping operator generated in the pre-growth code is transformed to a product of the stabilizers that define the post-growth code possibly with Pauli dressings, which is compatible with every Pauli-dressed protocol measurement of the post-growth code.

However, for non-protocol measurements, the compatibility of a branch-flipping operator and a future measurement does not necessarily hold.
That is, the propagation of a branch-flipping operator through a future non-protocol measurement circuit can have a different effect on measurement probabilities compared to their Pauli-proxy counterparts.
A canonical example of such a situation is the expansion of a code into a larger-distance code by measuring new operators, as occurs in the escape stage of magic-state cultivation protocols.
Here we explicitly consider the growth of the $d=3$ regular surface code $|\bar{T}\rangle$ state into the $d=4$ $|\bar{T}\rangle$ state as illustrated in Fig.~\ref{fig:escape}(a).
\begin{figure}[t]
    \centering
    \includegraphics[width=0.9\linewidth]{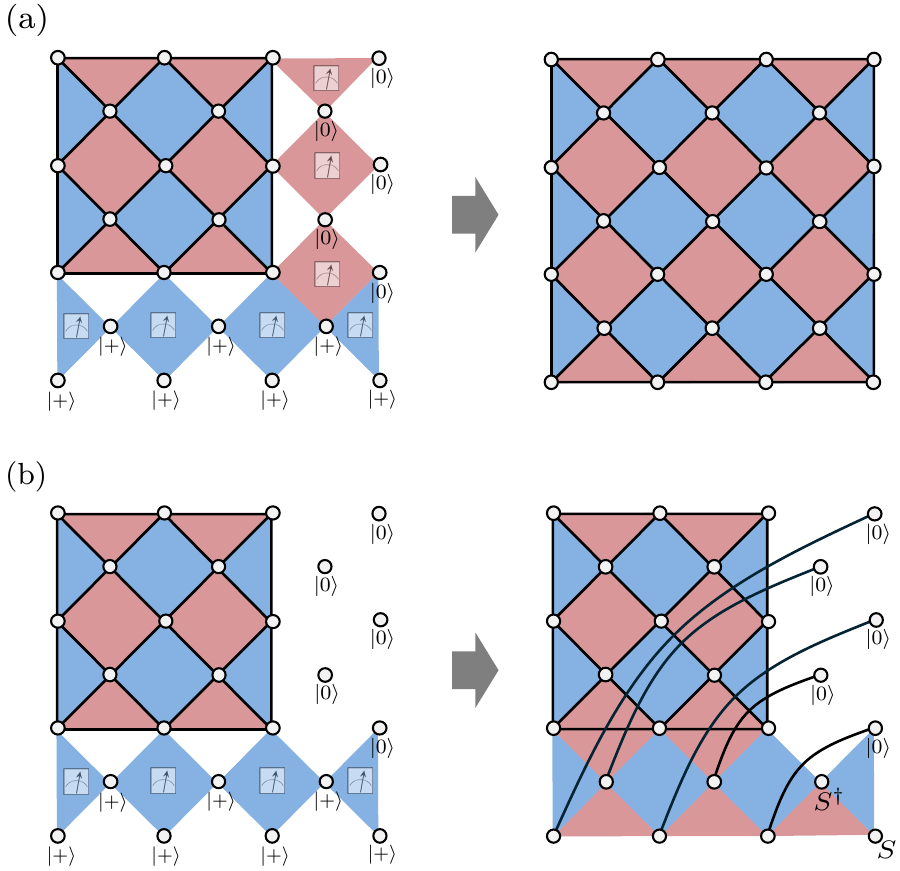}
    \caption{Growing the regular surface code from $d=3$ to $d=4$. (a) The full growing measurements. Once we measure the three new $X$ and four new $Z$ stabilizers indicated by the measurement symbols, the $d=3$ code is grown to the $d=4$ code. (b) The state after measuring all new $Z$ stabilizers. After the measurements, the frame state is stabilized by $\bar{S}{^{(3)}}^\dagger \bar{S}^{(4)}$, which is displayed in the right-hand panel.}
    \label{fig:escape}
\end{figure}
In Fig.~\ref{fig:escape}(a), we measure the four new $Z$ stabilizers and three new $X$ stabilizers to grow the code.
Each new $Z$ stabilizer commutes or anticommutes with every element of the current Clifford-stabilizer group, including the $\bar{H}_{XY}$ of the $d=3$ code, which is denoted by $\bar{H}^{(3)}_{XY}$. 
Since they anticommute with the single-qubit $X$ operators that stabilize the physical $|+\rangle$ states, the branch-flipping operator for each new $Z$ stabilizer measurement is a single-qubit $X$. 
After all new stabilizers have been measured, a subsequent stage of the protocol may measure $\bar{H}^{(4)}_{XY}$, the fold-transversal $\bar{H}_{XY}$ of the $d=4$ code. 
Notably, the (Pauli-dressed) fold-transversal $\bar{H}^{(4)}_{XY}$ neither commutes nor anticommutes with the single-qubit $X$ branch-flipping operators, so they are not compatible.
Therefore, the branch-flipping operators must be propagated in the non-Clifford circuit instead of just flipping the sign of the newly measured $Z$ stabilizers in the Pauli-proxy stabilizer tableau when the $-1$ branches are obtained. 
To account for this issue, we need to modify Algorithm~\ref{alg:circuit} as shown in Algorithm~\ref{alg:branch}.

Note that valid non-protocol measurements in Algorithm~\ref{alg:branch} are those that anticommute with at least one generator and either commute or anticommute with every other generator in the current Clifford-stabilizer group.
The procedure to handle a correction based on a midcircuit measurement outcome is detailed in Sec.~\ref{subsec:algorithm_circuit}.
The purpose of Step 2 is to cancel the incorrect sign flip performed by a Clifford simulator for the $-1$ branch. 
It can be accomplished by applying a Pauli branch-flipping operator of the Pauli-proxy stabilizer tableau. 
The branch-flipping operator applied in Step 3 is propagated further in the non-Clifford circuit in the same way as a correction operator.

Above, we did not specify how to simulate the measurements of the new $X$ stabilizers that neither commute nor anticommute with the diagonal $CZ$ and $S$ included in $\bar{H}^{(3)}_{XY}$.
Below, we show that Algorithm~\ref{alg:circuit} can directly simulate them.
\begin{figure}[tb]
\normalsize
\begin{algorithm}[Non-protocol measurements]\label{alg:branch}

\leavevmode\par
\medskip
\noindent\textbf{Input:}\par
\noindent

\begin{itemize}
    \item A noisy non-Clifford circuit performing protocol and non-protocol measurements with a sampled configuration of circuit-level Pauli errors.
\end{itemize}

\noindent\textbf{Output:}\par
\noindent
\begin{itemize}
    \item Measurement outcomes sampled from exactly the same probability distribution as in the input noisy non-Clifford circuit.
    \item Occurrence of a logical error in the output state.
\end{itemize}

\noindent\textbf{Procedure:}\par
\begin{enumerate}[label=\textbf{\arabic*.}]
  \item Execute Algorithm~\ref{alg:circuit}, treating each measurement whose branch-flipping operator can be incompatible with at least one subsequent measurement as a midcircuit measurement used to apply a correction.
  \item If each of such measurements yields the $-1$ branch with probability 50\%, flip the sign back to $+1$ in the Pauli-proxy stabilizer tableau. 
  \item Apply the corresponding branch-flipping operator as a correction immediately after the measurement in the non-Clifford circuit and continue the algorithm.
\end{enumerate}
\end{algorithm}
\end{figure}
After measuring all new $Z$ stabilizers, the frame state has the stabilizers shown on the right-hand panel of Fig.~\ref{fig:escape}(b). Here, the eigenvalues of the new $Z$ stabilizers are $+1$ on the frame state. 
The key point is that this state is also stabilized by $\bar{S}{^{(3)}}^\dagger \bar{S}^{(4)}$ as displayed in Fig.~\ref{fig:escape}(b), where $\bar{S}^{(3)}$ and $\bar{S}^{(4)}$ are the fold-transversal $\bar{S}$ gates of the $d=3$ and $d=4$ codes, respectively.
This is because each $CZ$ operator stabilizes the frame state since one of its supporting qubits is the $|0 \rangle$ state, and the $S^\dagger_i S_j$ operator also stabilizes the frame state since there is a weight-two $Z$ stabilizer supported on the two qubits on which the $S^\dagger_i$ and $S_j$ act.
Note that the weight-two $Z$ stabilizer has a $+1$ eigenvalue since the branch-flipping operator is treated as an error and thus is not acting on the frame state. 
Therefore, the current generator of the frame state that includes $\bar{H}^{(3)}_{XY}$ as its component can be multiplied by $\bar{S}{^{(3)}}^\dagger \bar{S}^{(4)}$ so that the subsequent new $X$ stabilizer measurements effectively commute with the component $\bar{S}{^{(3)}}^\dagger \bar{S}^{(4)}\bar{H}^{(3)}_{XY}$.
We remark that each new $X$ stabilizer measurement anticommutes with the single-qubit $Z$ stabilizer of a $|0\rangle$ state, and the corresponding branch-flipping operator is compatible with any future measurements. Since the Pauli proxies of $\bar{H}^{(3)}_{XY}$ and $\bar{S}{^{(3)}}^\dagger \bar{S}^{(4)}\bar{H}^{(3)}_{XY}$ are the same, Algorithm~\ref{alg:circuit} simulates the new $X$ stabilizer measurements correctly without any modification. 

Finally, we provide more examples of non-protocol measurements that Algorithm~\ref{alg:circuit} or Algorithm~\ref{alg:branch} can simulate.
Single-qubit $Z$-basis measurements on data qubits for diagonal logical magic states are simulable since they anticommute with at least one $X$ stabilizer and either commute or anticommute with every other current generator.
Another example is measurement-based injection. We can prepare a logical magic state by measuring Clifford stabilizers on a Pauli-stabilizer state. 
If the Clifford stabilizers anticommute with at least one generator and either commute or anticommute with every other generator of the Pauli-stabilizer state, Algorithm~\ref{alg:circuit} or Algorithm~\ref{alg:branch} works. For example, measuring $\bar{H}_{XY}$ on the $|\bar{0}\rangle$ state of the regular surface code satisfies this condition. 
In particular, this applies to gauging logical measurements of Clifford operators~\cite{vrty-qs5h}. 

\section{Discussion}\label{sec:discuss}

In this work, we have introduced Clifford-stabilizer simulation, an efficient method for the exact simulation of a broad family of noisy logical magic state preparation protocols. 
We have provided a framework for tracking the evolution of Clifford-stabilizer states by systematically updating the corresponding Clifford-stabilizer groups.
We have also developed an efficient algorithm to implement the group update, which results in efficient and exact simulation of non-Clifford circuits implementing logical magic state preparation protocols under circuit-level Pauli noise.
This algorithm compiles each sample of the noisy non-Clifford circuit into a Clifford circuit with an identical measurement outcome distribution that can be run on a standard Clifford simulator. 
We applied the proposed Clifford-stabilizer simulation to exactly and efficiently simulate magic state cultivation across a range of system sizes that go beyond what was accessible to existing techniques. 
Since our Clifford-stabilizer simulation method is scalable, it has the potential to simulate larger fault-distance protocols.

Our Clifford-stabilizer simulation sets the stage for exact simulation of large-scale logical magic state preparation protocols that are relevant for useful FTQC.
This fills an important gap in the literature and complements well-established methods for exact simulation of fault-tolerant Clifford logic in Pauli-stabilizer codes. 
As such, we expect Clifford-stabilizer simulation to become a standard and ubiquitous tool for the simulation of the primitive operations required for universal FTQC.
In particular, our simulation method is well suited to perform exact circuit-level noise simulations of the recently introduced scalable twisted quantum double (TQD) magic state preparation protocol~\cite{vrty-qs5h} and its generalizations~\cite{Williamson_2026Low,williamson2026fastmagicstatepreparation,christos2026nonabelianquantumlowdensityparity,j2q4-bnc6,zhu2026nonabelianqldpctqftformalism,huang2025generatinglogicalmagicstates,huang2026hybridlatticesurgerynonclifford,kfxl-rv7c,manjunath2026universalquantumcomputationgroup}. 
This should allow a fair comparison and benchmarking of magic state cultivation, the TQD protocol, and magic state distillation over a range of physically relevant error rates. 
See also related work on simulating non-Clifford gates on topological codes~\cite{bombin20182dquantumcomputation3d,brown2020fault,Scruby2022numerical,PhysRevResearch.4.043052,scruby2025fault}.

Our results raise a number of directions for future work, including the extension of our simulation method to allow less restrictive assumptions and to explore the full domain of applicability of the method. 
For example, it would be interesting to explore generalizations of the Clifford-stabilizer simulation to more cases where a measurement neither commutes nor anticommutes with current stabilizers.
In some cases, such as the example in Sec.~\ref{subsec:new_measurement}, even if the current tracked generators neither commute nor anticommute with a measurement, we can change the generators to those that either commute or anticommute with  the measurement. Hence, our Clifford-stabilizer simulation applies to such cases.
For general non-protocol measurements, it must be ensured that a measurement either commutes or anticommutes with every generator. A systematic method for checking the existence of such a set of generators should be explored. 
For certain non-protocol measurements, such as a Clifford projection on a Pauli-stabilizer state, the measurement outcomes can be neither uniformly random nor deterministic. This implies that the tracked group must contain a generator that neither commutes nor anticommutes with the measured operator.
It would be interesting to extend our framework to such cases.
Furthermore, it would be valuable to extend the method beyond qubits to qudits including non-Abelian group algebras that underlie recent topological magic state preparation protocols~\cite{huang2025generatinglogicalmagicstates,huang2026hybridlatticesurgerynonclifford,manjunath2026universalquantumcomputationgroup}. 
Finally, the generalization of our formalism to stabilizer groups generated by operators from higher levels of the Clifford hierarchy should be explored. 
Such a formalism could accommodate a potential generalization of the Bravyi-K\"onig bound on achievable gates on Clifford hierarchy stabilizer groups~\cite{PhysRevLett.110.170503}. 

We note independently developed forthcoming  work~\cite{julio_upcoming,tom_upcoming}: Ref.~\cite{julio_upcoming} develops an efficient method for simulating non-Clifford circuits with diagonal third-level gates in fault-tolerant protocols, whereas Ref.~\cite{tom_upcoming} combines this approach with existing methods to study logical failure rates in large-scale universal circuits.

\begin{acknowledgements}
We thank Isaac Kim, Lucas Daguerre, and Samyak Surti for inspiring discussions at the outset of this project and for helpful comments on an earlier version of the manuscript. 
We also appreciate Julio Carlos Magdalena De La Fuente, Thomas R. Scruby, and Thomas B. Smith for helpful comments on an earlier version of the manuscript.
YT is grateful to Keisuke Fujii, Nicholas Fazio, Campbell McLauchlan, Seok-Hyung Lee, Timo Hillmann, Georgia Nixon, Lucas English, and Andrew Li for helpful discussions.

This research was sponsored by the State of Maryland through the QSANDRA program at the University of Maryland’s Applied Research Laboratory for Intelligence and Security (ARLIS).
YT is supported by MEXT Quantum Leap Flagship Program (MEXT Q-LEAP) Grant No. JPMXS0120319794, JST COI-NEXT Grant No. JPMJPF2014, JST Moonshot R\&D Grant No. JPMJMS256E, JST CREST Grant No. JPMJCR24I3, and Japan Society for the Promotion of Science (JSPS) KAKENHI Grant No. 25KJ1756.
DJW is supported by the Australian Research Council Discovery Early Career Research Award (DE220100625).
\end{acknowledgements}

\bibliography{main}

\newpage

\appendix

\section{Details of simulated magic state cultivation}
\label{app:cultivation}
In this appendix, we provide the details of the fold-transversal surface code cultivation simulation we performed in Sec.~\ref{sec:simulation}.

\subsection{Protocols}
\label{app_subsec:cultivation_protocol}
We provide the procedures of the protocols for each fault distance $f$ in Table~\ref{tab:cultivation_procedure}.
The protocols post-select on runs in which all measurement outcomes are trivial.
The protocols for $f=3$ and $f=5$ are based on those in Ref.~\cite{gpvl-lg4c}, but the procedures are slightly different. In particular, we measure Pauli stabilizers only on the regular surface codes, while some Pauli stabilizers are measured on the rotated surface codes in Ref.~\cite{gpvl-lg4c}. 
Moreover, we develop the $f=7$ protocol, a regime not explored in Ref.~\cite{gpvl-lg4c}.

For the initial $|\bar{T}\rangle$ injection into the $d=3$ rotated surface code, we use the unitary encoding circuit~\cite{tsai2025unitaryencodersurfacecodes} shown in Fig.~\ref{fig:injection}.
To grow the $d=3$ rotated code into the $d=3$ regular code, we apply a unitary circuit which corresponds to a half-cycle of Pauli-stabilizer measurements for the weight-$4$ stabilizers as shown in Fig.~\ref{fig:cultivation}. 
For growing the $d=3$ regular code into the $d=5$ regular code, we first grow the $d=3$ regular code into the $d=5$ rotated code using the unitary circuit used in Refs.~\cite{gpvl-lg4c,tsai2025unitaryencodersurfacecodes}. Then, we grow the $d=5$ rotated code into the $d=5$ regular code by performing a half-cycle of Pauli-stabilizer measurements. 
When growing the $d=5$ regular code into the $d=7$ regular code, we directly grow the code using the unitary circuit proposed in Ref.~\cite{Higgott2021optimallocalunitary}.
Each Pauli stabilizer is measured using a single ancilla qubit~\cite{PhysRevResearch.7.033074}.

Below, we explain the circuits for measuring $\bar{H}_{XY}$.
Consider a regular surface code of odd code distance $d$. We define the fold line as the diagonal connecting the upper-left and lower-right corners of the lattice. The data qubits on the fold line are denoted by the ordered set
\begin{equation}
    \mathcal{Q}=(q_0,q_1,\ldots,q_{2d-2}),
\end{equation}
indexed from the upper-left corner to the lower-right corner. Let $\mathcal{T}$ denote the set of all data qubits on the lower-left side of the fold line, excluding the qubits on the fold line. Reflection across the fold line is denoted by $\mu$, so that every $t\in\mathcal{T}$ is paired with the data qubit $\mu(t)$ at a mirrored position above the fold line.
The fold-transversal $\bar{H}_{XY}$ operator is expressed as
\begin{align}
    \bar{H}_{XY}
    &=e^{-i\pi/4}\bar{S}\,\bar{X}
    \nonumber\\
    &=e^{-i\pi/4}
    \prod_{j\ \mathrm{even}}(SX)_{q_j}
    \prod_{j\ \mathrm{odd}}(S^{\dagger}X)_{q_j}
    \prod_{t \in\mathcal{T}}CZ_{t,\mu(t)}
    \nonumber\\
    &=\prod_{j\ \mathrm{even}}(TXT^\dagger)_{q_j}
    \prod_{j\ \mathrm{odd}}(T^{\dagger}XT)_{q_j}
    \prod_{t \in\mathcal{T}}CZ_{t,\mu(t)}.
    \label{eq:hxy_hermitian_decomp}
\end{align}
In Eq.~\eqref{eq:hxy_hermitian_decomp}, the factors are all Hermitian, so $\bar{H}_{XY}$ is measured by jointly measuring the physical Hermitian operators. The single-qubit Hermitian factors $TXT^\dagger$ and $T^\dagger XT$ are measured using the controlled operations
\begin{align}
    \Lambda(TXT^\dagger)_{a,q}
    &=(I_a\otimes T_q)\,
      \CNOT_{a,q}\,
      (I_a\otimes T_q^{\dagger}),
    \\
    \Lambda(T^\dagger XT)_{a,q}
    &=(I_a\otimes T_q^{\dagger})\,
      \CNOT_{a,q}\,
      (I_a\otimes T_q),
\end{align}
respectively, where $a$ denotes an ancilla qubit and $q$ denotes a data qubit.
The $CZ$ factor is measured using a $CCZ$ gate. 

\begin{table}[t]
    \centering
    \caption{Procedures for the fold-transversal surface code cultivation protocols. In this table, rotated and regular surface codes are abbreviated as rotated and regular, respectively.}
    \small
    \renewcommand{\arraystretch}{1.5}
    \begin{tabular}{ll}
        \hline\hline
        \parbox[t]{0.12\linewidth}{\centering Protocol}
        & \parbox[t]{0.76\linewidth}{Procedure} \\
        \midrule
        \parbox[t]{0.12\linewidth}{\centering $f=3$}
        & \parbox[t]{0.76\linewidth}{%
              \begin{enumerate}
                  \setlength{\itemsep}{0pt}
                  \setlength{\parskip}{0pt}
                  \setlength{\parsep}{0pt}
                  \setlength{\topsep}{0pt}
                  \item $|\bar{T}\rangle$ injection into $d=3$ rotated.
                  \item Grow $d=3$ rotated into $d=3$ regular.
                  \item Measure stabilizers of $d=3$ regular.
                  \item Measure $\bar H_{XY}$ twice on $d=3$ regular.
                  \item Measure stabilizers of $d=3$ regular.
              \end{enumerate}
          } \\
        \addlinespace 
        \parbox[t]{0.12\linewidth}{\centering $f=5$}
        & \parbox[t]{0.76\linewidth}{%
              \begin{enumerate}
                  \setlength{\itemsep}{0pt}
                  \setlength{\parskip}{0pt}
                  \setlength{\parsep}{0pt}
                  \setlength{\topsep}{0pt}
                  \item $|\bar{T}\rangle$ injection into $d=3$ rotated.
                  \item Grow $d=3$ rotated into $d=3$ regular.
                  \item Measure stabilizers of $d=3$ regular.
                  \item Measure $\bar H_{XY}$ twice on $d=3$ regular.
                  \item Grow $d=3$ regular into $d=5$ regular.
                  \item Measure stabilizers of $d=5$ regular.
                  \item Measure $\bar H_{XY}$ twice on $d=5$ regular.
                  \item Measure stabilizers of $d=5$ regular.
              \end{enumerate}
          } \\
        \addlinespace
        \parbox[t]{0.12\linewidth}{\centering $f=7$}
        & \parbox[t]{0.76\linewidth}{%
              \begin{enumerate}
                  \setlength{\itemsep}{0pt}
                  \setlength{\parskip}{0pt}
                  \setlength{\parsep}{0pt}
                  \setlength{\topsep}{0pt}
                  \item $|\bar{T}\rangle$ injection into $d=3$ rotated.
                  \item Grow $d=3$ rotated into $d=3$ regular.
                  \item Measure stabilizers of $d=3$ regular.
                  \item Measure $\bar H_{XY}$ twice on $d=3$ regular.
                  \item Grow $d=3$ regular into $d=5$ regular.
                  \item Measure stabilizers of $d=5$ regular.
                  \item Measure $\bar H_{XY}$ twice on $d=5$ regular.
                  \item Grow $d=5$ regular into $d=7$ regular.
                  \item Measure stabilizers of $d=7$ regular.
                  \item Measure $\bar H_{XY}$ twice on $d=7$ regular.
                  \item Measure stabilizers of $d=7$ regular.
              \end{enumerate}
          } \\
        \hline\hline
    \end{tabular}
    \label{tab:cultivation_procedure}
\end{table}

\begin{figure}[t]
    \centering
    \includegraphics[width=0.9\linewidth]{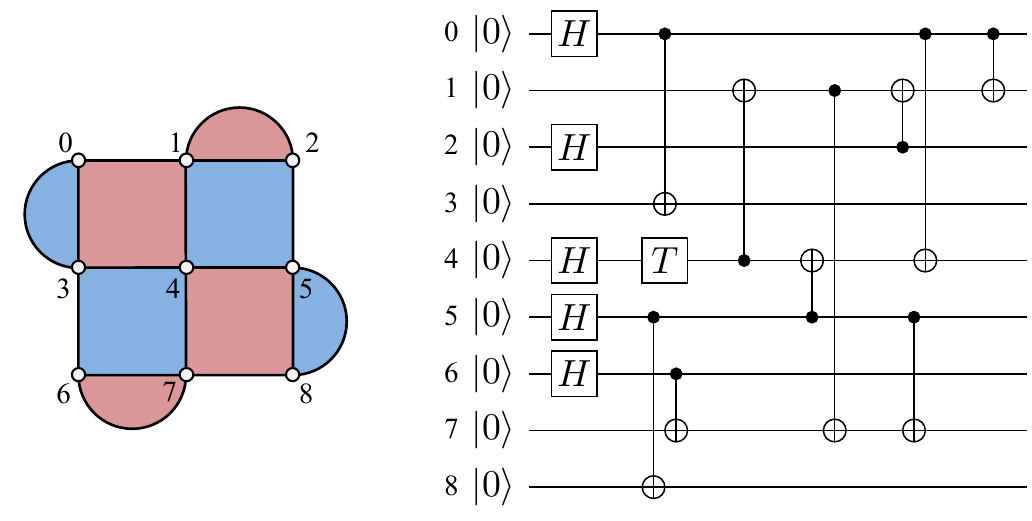}
    \caption{Unitary encoding circuit for the $|\bar{T}\rangle$ state of the $d=3$ rotated surface code~\cite{tsai2025unitaryencodersurfacecodes}. This circuit was used at the injection stage in our simulation.}
    \label{fig:injection}
\end{figure}

We now describe the detailed circuits we used in our simulation to measure $\bar{H}_{XY}$ on the $d=3$, $d=5$, and $d=7$ regular surface codes.
The circuits are constructed systematically based on the construction in Ref.~\cite{gpvl-lg4c}. For $d=3$ and $d=5$, we use the same gate scheduling as in Ref.~\cite{gpvl-lg4c}. 
Note that the circuits primarily considered in Ref.~\cite{gpvl-lg4c} use controlled-$S$ gates, and a $T^\dagger$ gate is applied to the cat state right before readout to acquire the required phase. In our implementation, each factor is instead measured as a Hermitian operator by applying $T$ and $T^\dagger$ gates to the data qubits, and thus no $T^\dagger$ gate is applied to the cat state.
For $d=7$, we use a circuit that naturally extends this scheduling. 
The systematic construction is provided below.
First, as ancilla qubits, we use an $l$-qubit cat state
\begin{equation}
    \lvert \mathrm{cat}_{l}\rangle
    =\frac{\lvert0\rangle^{\otimes l}+\lvert1\rangle^{\otimes l}}{\sqrt{2}},
\end{equation}
whose physical qubits are labeled $(a_0,a_1,\ldots,a_{l-1})$. 
Our circuits use the cat states with $l$ smaller than the number of the data qubits, so the assignment and ordering of the controlled gates must be chosen carefully to preserve the desired fault distance. 
To explain the detailed gate scheduling, we introduce more parameters. The qubits in $\mathcal{T}$ form $d-1$ diagonal lines parallel to the fold line.  Starting with the line closest to the fold line and proceeding toward the lower-left corner, the qubits on these lines are labeled $\mathcal{T}_1,\mathcal{T}_2,\ldots,\mathcal{T}_{d-1}$.  Thus, the set $\mathcal{T}$ is partitioned into subsets
\begin{equation}
    \mathcal{T}=\bigsqcup_{k=1}^{d-1}\mathcal{T}_k.
    \label{eq:half-lattice-partition}
\end{equation}
The qubits in each subset are denoted by
\begin{equation}
    \mathcal{T}_k
    =\left\{t_{k,j}:j=0,1,\ldots,2(d-k)-2\right\},
\end{equation}
where increasing $j$ follows the line from upper left to lower right.  Thus, the line adjacent to the fold line has $2d-3$ qubits, and each subsequent line has two fewer qubits, ending with one qubit in $\mathcal{T}_{d-1}$. We extend the notation to $k=0$ by setting $\mathcal{T}_0=\mathcal{Q}$ and $t_{0,j}=q_j$.

In our circuits, the controlled gates associated with $\mathcal{T}_0=\mathcal{Q}$ are applied first, followed by those associated with $\mathcal{T}_1$, and so on until $\mathcal{T}_{d-1}$. Within each $\mathcal{T}_k$, the gates are applied in increasing order of $j$, starting from $j=0$.
About the assignment of each controlled gate, within each $\mathcal{T}_k$, the controlled gates are assigned cyclically to the qubits constituting the cat state. The cat-state qubit assigned to the controlled gate associated with $t_{k,0}$ depends on $k$. To specify this assignment, we introduce
\begin{equation}
    s_k\in\{0,1,\ldots,l-1\},
\end{equation}
where $s_k$ identifies the cat-state qubit that controls the first, upper-left factor associated with $\mathcal{T}_k$. In particular, the cat-state qubit $a_{s_0}$ controls the factor $TXT^\dagger$ on $q_0$, while $a_{s_k}$ controls the factor $CZ_{t_{k,0},\mu(t_{k,0})}$ for $k\geq1$.
Since the remaining controlled gates are assigned cyclically, the index of the cat-state qubit controlling the $j$th factor associated with $\mathcal{T}_k$ is 
\begin{equation}
    c(k,j)=(s_k+j)\bmod l.
\end{equation}
Accordingly, the controlled gate $G_{k,j}$ acting on the $j$th qubit, counted from the upper left in $\mathcal{T}_k$, is given by
\begin{align}
    G_{0,j}
    &=
    \begin{cases}
        \Lambda(TXT^\dagger)_{a_{c(0,j)},q_j}, & j\ \text{even},\\
        \Lambda(T^\dagger XT)_{a_{c(0,j)},q_j}, & j\ \text{odd},
    \end{cases}
    \qquad 0\leq j\leq2d-2,
    \\
    G_{k,j}
    &=CCZ_{a_{c(k,j)},t_{k,j},\mu(t_{k,j})},
    \qquad
    \begin{matrix}
        1\leq k\leq d-1,\\
        0\leq j\leq2(d-k)-2.
    \end{matrix}
\end{align}
In summary, the circuit first applies 
\begin{equation}
    G_{0,0},G_{0,1},\ldots,G_{0,2d-2}
\end{equation}
 in this order, followed by
 \begin{equation}
    G_{1,0},G_{1,1},\ldots,G_{1,2d-4}.
\end{equation}
It then proceeds analogously for increasing $k$ until $k=d-1$.

In Table~\ref{tab:cat_index}, we show the sizes of the cat states and the indices of the cat-state qubits that control the upper-leftmost factor associated with each diagonal line $\mathcal{T}_k$ used in our simulation.
\begin{table}[t]
    \centering
    \caption{The sizes $l$ of the cat states and the indices $s_k$ of the cat-state qubits supporting the controlled gates that act on the upper-leftmost data qubit of each diagonal line $\mathcal{T}_k$.}
    \begin{tabular}{@{}c c c@{}}
        \toprule
        protocol & $l$ & $(s_0,s_1,\ldots,s_{d-1})$ \\
        \midrule
        $f=3$ & $3$ & $(0,0,2)$ \\
        $f=5$ & $5$ & $(0,3,0,1,2)$ \\
        $f=7$ & $8$ & $(0,3,7,2,5,0,3)$ \\
        \bottomrule
    \end{tabular}
    \label{tab:cat_index}
\end{table}
For fault-tolerant preparation of the cat states, we used the circuits shown in Fig.~\ref{fig:cat}. The circuits for the $f=3$ and $f=5$ protocols are the same as those used in Ref.~\cite{gpvl-lg4c}. The preparation circuit of the $5$-qubit cat state for the $f=5$ protocol uses one flag qubit. For the $f=7$ protocol, we use the $8$-qubit cat state prepared with six verification qubits discovered in Ref.~\cite{13v7-n843}. The qubits verify fault tolerance through the transversal $\CNOT$ gates. The preparation is accepted only when the six verification outcomes are all equal.
In all cases, the cat states are unencoded after the controlled gates are applied. 
Note that a smaller cat state or a simpler preparation circuit could be used to achieve fault tolerance. However, we emphasize that optimizing the performance of the protocols is not the objective of this work. Rather, our goal is to demonstrate that our simulation method can efficiently and exactly simulate large logical magic state preparation circuits. Thus, optimization of the protocols themselves is outside the scope of this work.
 
\begin{figure}[t]
    \centering
    \includegraphics[width=0.9\linewidth]{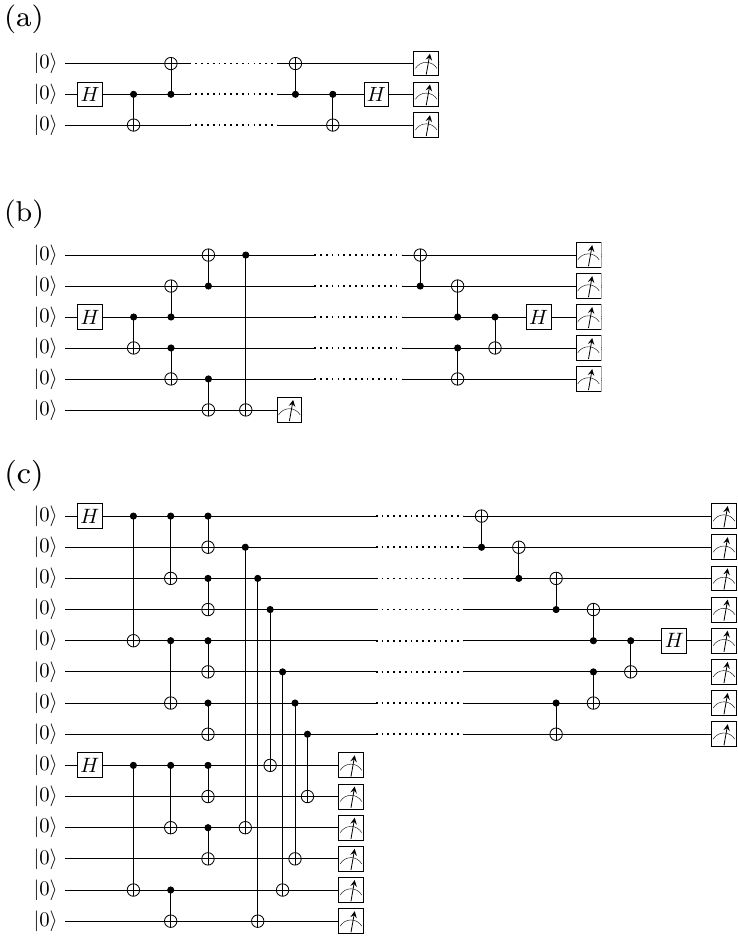}
    \caption{Fault-tolerant cat-state preparation circuits used for the fold-transversal magic state cultivation protocols with different fault distances. (a) $f=3$. (b) $f=5$. We post-select on the trivial measurement outcome of the bottom qubit. (c) $f=7$. We post-select on the all-$0$ or all-$1$ measurement outcomes of the bottom $6$ qubits. In all cases, at the dotted lines, controlled gates are applied between the cat states and the data qubits. After that, the cat states are unencoded.}
    \label{fig:cat}
\end{figure}

\subsection{Simulation setup}
\label{app_subsec:cultivation_setup}
We now describe the simulation setup and relevant implementation details.
In our simulation, ideal non-Clifford circuits are expressed using the elementary gates $H$, $\CNOT$, $T$, $T^\dagger$, and $CCZ$, together with $|0\rangle$ state initializations and $Z$-basis measurements. 
Note that the Pauli-proxy Clifford circuits contain Clifford errors generated by the error-propagation step in our algorithm. Thus, the Clifford circuits are expressed using the elementary gates $H$, $\CNOT$, $S$, $S^\dagger$, $CZ$, and Pauli $X$ and $Z$ gates, along with $|0\rangle$ initializations and $Z$-basis measurements.

For the noise model, we use the following standard circuit-level Pauli error model. After every ideal $k$-qubit unitary gate, we apply a $k$-qubit depolarizing channel with total error probability $p$. The channel is
\begin{equation}
    \mathcal{E}^{(k)}_p(\rho)
    =(1-p)\rho
    +\frac{p}{4^k-1}
      \sum_{P\in\{I,X,Y,Z \}^{\otimes k}\setminus\{I^{\otimes k}\}}
      P\rho P,
\end{equation}
where $\rho$ is a density operator.
No error is applied with probability $1-p$, while each of the $4^k-1$ nonidentity Pauli operators is applied with probability $p/(4^k-1)$.
For initialization, each $|0\rangle$ preparation is followed by an $X$ error with probability $p$. Also, immediately before each $Z$-basis measurement, an $X$ error is applied with probability $p$.
We do not include idling noise in our simulation.

Regarding the Pauli-stabilizer measurement schedule, in our simulations, Pauli stabilizers are measured sequentially by reusing a single ancilla qubit. We adopt this sequential schedule solely to ensure that the simulated protocols form a consistent family. Using one ancilla qubit per stabilizer allows the measurements to be parallelized, which might be more practical.
However, with this choice, obtaining enough samples from the state-vector simulator to reduce the error bars in Fig.~\ref{fig:numerics_statevector} sufficiently to clearly demonstrate agreement with Clifford-stabilizer simulation would be computationally hard.
In Sec.~\ref{app_subsec:cultivation_data}, we provide data showing the difference of the two schedules, although the difference is not important for the purpose of this work.

We next describe the implementation of FEC. For decoding, we use \texttt{PyMatching}~\cite{Higgott2025sparseblossom} to compute corrections based on the syndromes obtained by the fictitious round of stabilizer measurements. 
Then, we apply the corrections as errors in the circuits. These applied errors are further propagated through the final fictitious $\bar{H}_{XY}$ measurement circuit. 
Related to this, in our simulation, we use \texttt{stim.TableauSimulator} to run the Pauli-proxy Clifford circuits because we dynamically change a Clifford circuit when applying a Pauli correction for FEC in the middle of the circuit.
Even if a recovery operation is a Pauli, after propagating the Pauli through the subsequent non-Clifford circuit and mapping that circuit to a Clifford circuit, the resulting Clifford circuit is different depending on what Pauli correction we have applied. 

When evaluating a logical error, we need to check whether the final fictitious logical measurement commutes or anticommutes with the current Clifford-stabilizer group. 
To do this, we use \texttt{TableauSimulator.peek\_observable\_expectation}. 
The commutation relation of a protocol measurement on the data qubits with the data tableau is equivalent to that of an ancilla readout with the entire tableau. Thus, if \texttt{TableauSimulator.peek\_observable\_expectation} for the ancilla readout in the Pauli-proxy Clifford circuit returns $0$, the protocol measurement anticommutes with the current Clifford-stabilizer group. 

Finally, we note a convention used in computing the acceptance rates reported throughout this work.
In our simulation, each Monte Carlo shot simulates the whole circuit including the cat-state preparation, discarding the shot even when the verification measurements for the cat-state preparation detect errors.
In practice, the cat states can be prepared in parallel with the data protocol, so the reported acceptance rates computed in this way may overestimate the practical space-time cost. 
However, this distinction does not affect the objective of our numerical simulation, which is to demonstrate the efficiency and exactness of our simulation method for large logical magic state preparation circuits, rather than to obtain an optimized estimate of protocol performance.

\subsection{Additional data}
\label{app_subsec:cultivation_data}
We now provide additional data on the Pauli-stabilizer measurement schedules discussed above. 
Here, we quantify the effect of the sequential and parallel stabilizer-measurement schedules on the logical error rates and the acceptance rates for reference.
In Fig.~\ref{fig:numerics_parallel}, we compare the sequential and parallel schedules for the $f=3$ protocol, calculated via Clifford-stabilizer simulation.
\begin{figure}[t]
    \centering
    \includegraphics[width=0.9\linewidth]{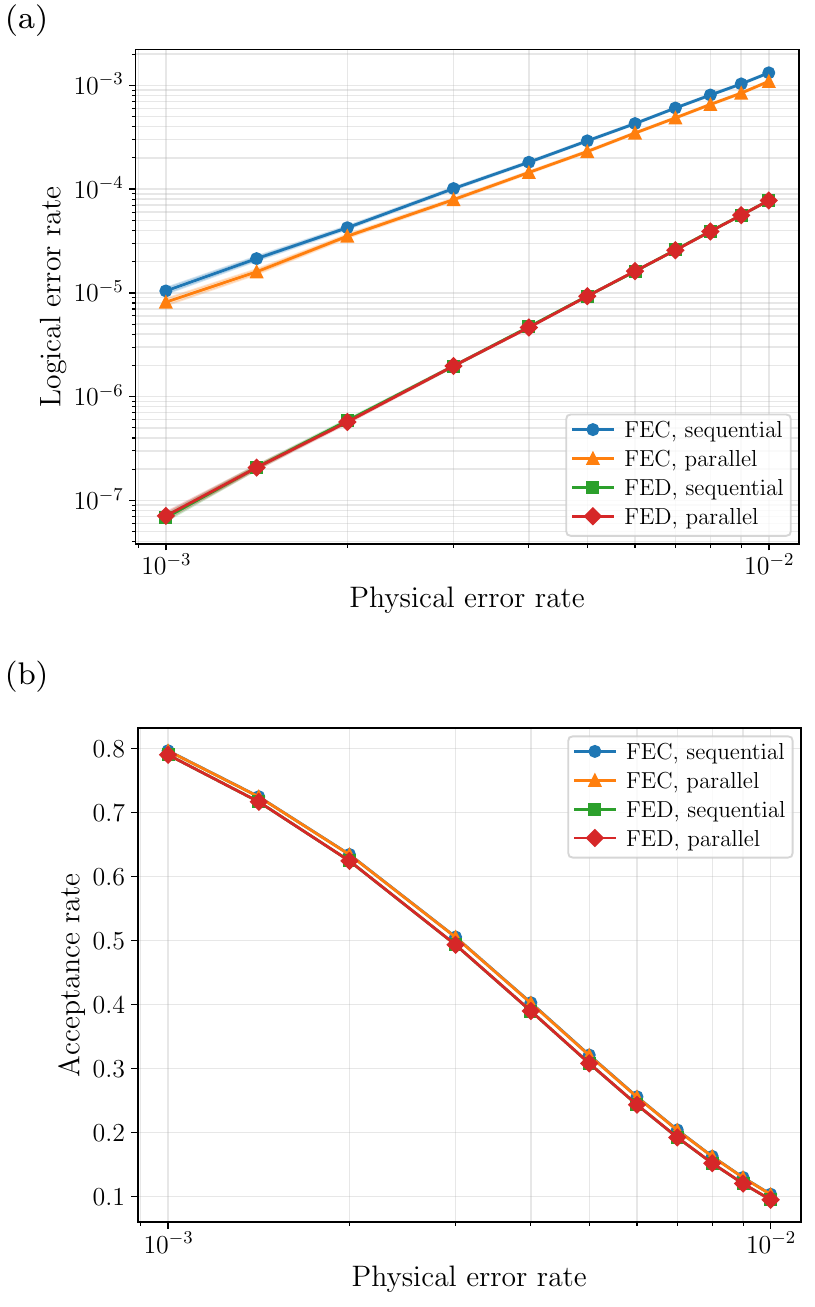}
    \caption{Comparison of sequential and parallel Pauli-stabilizer measurement schedules for the fold-transversal magic state cultivation of $f=3$. The results for the sequential schedule are the same as those shown in Figs.~\ref{fig:numerics_ec} and~\ref{fig:numerics_ed}. For the parallel schedule, the same number of Monte Carlo samples is used as for the corresponding sequential schedule. (a) Logical error rate. (b) Acceptance rate.}
    \label{fig:numerics_parallel}
\end{figure}
The results show that, with FED, the two schedules yield nearly identical values. In contrast, with FEC, the sequential schedule yields a slightly higher logical error rate than the parallel schedule, even in the absence of idling noise. 
This difference arises from the number of errors occurring during the final noisy stabilizer-measurement round that are post-selected. We note that the sequential schedule first measures all $X$ stabilizers, then measures all $Z$ stabilizers. 
In the case of FED, an error, including a hook error occurring during the final noisy stabilizer-measurement round is post-selected regardless of the measurement schedule because it is detected by the fictitious round even if it is not detected by the noisy round itself. 
On the other hand, in the case of FEC with the sequential schedule, some errors, especially hook errors occurring during the final noisy $Z$ stabilizer-measurement round, are not detected by the noisy round itself. They are only detected by the fictitious round, so they are not post-selected and instead corrected. With the parallel schedule, hook errors from both noisy $X$ and $Z$ stabilizer measurements are detected by the noisy round itself, so they are post-selected even in the case of FEC. 
Thus, measuring the stabilizers in parallel using one ancilla qubit per stabilizer can achieve lower logical error rates than those reported in Fig.~\ref{fig:numerics_ec}. We emphasize again, however, that the purpose of our simulation is not to establish the best performance of the protocols. Rather, it is intended to demonstrate that our simulation method can efficiently and exactly simulate large logical magic state preparation circuits.

\clearpage
\section{Supplementary proofs}
\label{app:proofs}
In this appendix we give the proofs of two claims from the main text. 

\subsection{
  \texorpdfstring{
    Eq.~\eqref{eq:dressings_commute_each}
    $\Longrightarrow$
    Eq.~\eqref{eq:commute_anti_condition}
  }{
    Eq. (\ref*{eq:dressings_commute_each})
    => Eq. (\ref*{eq:commute_anti_condition})
  }
}
Let $C$ be a Clifford factor and $D$ a Pauli dressing. For an arbitrary single-qubit Pauli $P$, we write $DP=sPD$, where $s\in\{\pm1\}$. Since $D$ commutes with
the dressing $D_C(P)=P(CPC^\dagger)$,
\begin{equation}
  D(CPC^\dagger)=s(CPC^\dagger)D.
\end{equation}
Conjugating by $C^\dagger$ gives
\begin{equation}
  (C^\dagger DC)P=sP(C^\dagger DC).
\end{equation}
Thus, $D$ and $C^\dagger DC$ have the same commutation relation with
every Pauli $P$ and therefore can differ only by a phase.
Since conjugation preserves $D^2$, this phase can only be $\pm1$:
\begin{equation}
  C^\dagger DC=\pm D.
\end{equation}
Equivalently,
\begin{equation}
[C,D]=\pm I.
\end{equation}

\subsection{
  \texorpdfstring{
    Eq.~\eqref{eq:dressings_commute_each}
    $\Longrightarrow$
    Eq.~\eqref{clifford-clifford-dress}
  }{
    Eq. (\ref*{eq:dressings_commute_each})
    => Eq. (\ref*{clifford-clifford-dress})
  }
}

For two Clifford factors $C$ and $C'$, compare the action of $CC'$ and $C'C$ on an arbitrary single-qubit Pauli $P$.
\begin{align}
  CC'P(CC')^\dagger &=C(P D_{C'}(P))C^\dagger  \\
  &=(CPC^\dagger) (CD_{C'}(P)C^\dagger)\\
  &=P D_C(P)\,(\pm D_{C'}(P)),
\end{align}
whereas
\begin{align}
  C'CP(C'C)^\dagger&=C'(P D_C(P))C'^\dagger  \\
    &=(C'PC'^\dagger)(C' D_C(P)C'^\dagger) \\
  &=P D_{C'}(P)\,(\pm D_C(P)).
\end{align}
Since all Pauli dressings mutually commute due to Eq.~\eqref{eq:dressings_commute_each}, these two results differ only by a sign:
\begin{equation}
    CC'P(CC')^\dagger=\pm C'CP(C'C)^\dagger.
\end{equation}
Now we have
\begin{align}
    [C,C']C'CP(C'C)^\dagger[C,C']^\dagger&=CC'P(CC')^\dagger\\
    &=\pm C'CP(C'C)^\dagger.
\end{align}
Here, the operators $C'CP(C'C)^\dagger$ generate the Pauli group modulo overall phases as $P$ ranges over all single-qubit Paulis.
Thus, $[C,C']$ either commutes or anticommutes with every Pauli, and hence $[C,C']$ is a Pauli operator up to a phase.
Furthermore, since Eq.~\eqref{eq:dressings_commute_each} implies Eq.~\eqref{eq:commute_anti_condition} as proven above, both $C$ and $C'$ either commute or anticommute with every Pauli dressing, and thus $[C,C']$ commutes with every Pauli dressing. 

\end{document}